\documentclass[12pt, margin=1.05in]{article}
 
\usepackage{amsfonts,amssymb,amsmath,amsthm,epsfig,euscript,bbm,amsthm}
\usepackage{algorithm}
\usepackage[dvipsnames]{xcolor}
\usepackage{algpseudocode}
\usepackage{amsthm}
\usepackage{threeparttable,subcaption}
 \usepackage{booktabs}
\newcommand{\ra}[1]{\renewcommand{\arraystretch}{#1}}
 \usepackage{titletoc}

\usepackage{rotating,multirow,makecell,array}

\usepackage {tikz}
\usetikzlibrary {positioning}
\definecolor {processblue}{cmyk}{0.96,0,0,0}
\usepackage{colortbl}

\usepackage{tikz}
\usetikzlibrary{calc}
\usetikzlibrary{arrows}
\usepackage {graphicx}
\usepackage{enumerate}
\pdfoutput=1
\usepackage[margin=1in]{geometry}
\usepackage{subcaption}
\usepackage{caption}
\usepackage[authoryear,sort]{natbib}
\usepackage[titletoc,title]{appendix}
\usepackage[pdftex,colorlinks]{hyperref}
\usepackage{mathtools}
 \usepackage{xcolor}

\usepackage{graphicx}
\usepackage{subcaption}
 
 \usepackage{xr}

 \definecolor{orange}{RGB}{230,170,120}
  \definecolor{green}{RGB}{120,200,120}

 \hypersetup{
 	colorlinks=true,
 	linkcolor=blue,
 	filecolor=blue,      
 	urlcolor=cyan,
 }
 \hypersetup{colorlinks,linkcolor={blue},citecolor={blue},urlcolor={red}} 
 \usepackage{geometry}                
 \usepackage{multirow}
\usepackage[section]{placeins}
 \usepackage{bbm}
 \usepackage{graphicx}
 \usepackage{amssymb}
 
\def\addlegendimage{\csname pgfplots@addlegendimage\endcsname}

\pgfkeys{/pgfplots/number in legend/.style={%
        /pgfplots/legend image code/.code={%
            \node at (0.295,-0.0225){#1};
        },%
    },
}

 \usepackage{epstopdf}
 \usepackage{tikz} 
 \usepackage{amsmath} 
\usepackage[utf8]{inputenc}
\usepackage{algorithm}
\usepackage{pgfplots}
 \theoremstyle{plain}
 \newtheorem{thm}{Theorem}[section]
 \newtheorem{lem}[thm]{Lemma}
 
 \newtheorem{cor}{Corollary}
 
 \theoremstyle{definition}

 \newtheorem{exmp}{Example}[section]
  \newtheorem{ass}{Assumption}[section]
  
 \theoremstyle{definition}
 \newtheorem{rem}{Remark}
 
 \makeatletter
 \def\BState{\State\hskip-\ALG@thistlm}
 \makeatother
 
 \usepackage{authblk}

\def\spacingset#1{\renewcommand{\baselinestretch}%
{#1}\small\normalsize} \spacingset{1}

\algnewcommand\algorithmicforeach{\textbf{for each}}
\algdef{S}[FOR]{ForEach}[1]{\algorithmicforeach\ #1\ \algorithmicdo}
 
\usepackage{setspace}
 \usepackage{epstopdf}

\title{Dynamic covariate balancing:\\ estimating  treatment effects over time with potential local projections  }

\author{Davide Viviano}
\affil{Department of Economics, Harvard University\footnote{Email: dviviano@fas.harvard.edu}}

\author{Jelena Bradic }
\affil{Department of Statistics and Data Science, Cornell University\footnote{Email: jelena.bradic@cornell.edu.}}

\date{This version: February, 2026 \\ First version: March, 2021}

\begin{document}
\maketitle

\begin{abstract}
This paper studies the estimation and inference of treatment effects in panel data settings when treatments change dynamically over time. 
 We propose a balancing method that allows for (i) treatments to be assigned dynamically over time based on high-dimensional covariates, past outcomes, and treatments; (ii) outcomes and time-varying covariates to depend on the trajectory of all past treatments; (iii) heterogeneity of treatment effects. 
 Our approach recursively projects potential outcomes' expectations on past histories. It then controls the bias arising from the non-experimental and sequential nature of this setting by balancing dynamically  observable characteristics over time.  We establish inferential guarantees of the proposed method even when the number of observable characteristics significantly exceeds the sample size.
We study  numerical properties of the estimator and illustrate the benefits of the procedure in an empirical application.   
\end{abstract}


\section{Introduction}
 \spacingset{1.7} 
Researchers collect a panel of $n$ independent observations observed over a finite number of $T$ periods in an observational study. The dataset encompasses time-varying covariates, outcomes, and time-varying treatments. The primary objective is to conduct inference on the average effect of exposure to different treatment histories, such as the effect of being treated for a certain number of periods.

We consider a setting where 
treatments change dynamically over time, and potential outcomes, covariates and treatment may depend on past histories. Two alternative procedures can be considered in this setting. First, researchers may consider explicitly modeling how treatment effects propagate over each period through time-varying covariates and intermediate outcomes. This approach is prone to large estimation error and misspecification in high-dimensions: it requires modeling outcomes and each time-varying covariate as a function of all past covariates, outcomes, and treatment assignments. A second approach is to use inverse-probability weighting estimators for estimation and inference \citep{tchetgen2012semiparametric, vansteelandt2014structural} . However, classical semi-parametric estimators are prone to instability in the estimated propensity score. There are two main reasons. First of all, the propensity score defines the joint probability of the entire treatment history and can be close to zero for moderately long treatment histories. Additionally, the propensity score can be misspecified in observational studies.

This is a common problem both in social sciences and bio-statistics. For example, in a survey of all articles in 2021 top-5 economics journals, more than $20\%$ of studies with time-varying treatments exhibit treatment dynamics.\footnote{This is based on the authors' calculation. Top-5 economics journals are \textit{American Economic Review, Econometrica, Journal of Political Economy, Quarterly Journal of Economics, Review of Economic Studies.}} 
On the other hand, typical approaches in economics and related disciplines often employ, 
Difference-in-Differences designs. If units can dynamically choose treatments in response to their previous outcomes (or treatments), this will lead to violations of the parallel trends assumption required by such designs \citep[][]{ghanem2022selection, marx2022parallel}. We, therefore, introduce an approach that is valid when treatment decisions at period $t$ can depend on the history of outcomes and treatments prior to $t$. The second challenge is that treatment dynamics are difficult to estimate. Individuals may select into treatment arbitrarily based on high-dimensional covariates, outcomes, and treatments, e.g., when maximizing future expected utilities \citep[][]{heckman2007dynamic}. This motivates a method that does not impose modeling assumptions on selection into treatment mechanisms (i.e., propensity score).

This paper studies the estimation and inference of the effects of treatment histories when potential outcomes (and covariates) depend on present and past treatments. Individuals dynamically select into treatment based on past (time-varying) covariates, outcomes, and treatments. There are no unobserved confounders after controlling for high-dimensional past characteristics \citep{ding2019bracketing}. Researchers remain agnostic on the propensity score.

We leverage a model on the \textit{potential} outcomes' conditional expectations as an (approximately) linear function of previous potential outcomes and (high-dimensional) covariates in each period. Our model is motivated by local projection frameworks \citep{jorda2005estimation, montiel2021local}. Local projections impose a (linear) model on observed outcomes conditional on each period observables and do not require estimating how each time-varying covariate changes in response to treatments -- which would be prone to large estimation error in high dimensions. However, different from standard local projections, our model is imposed on expected potential instead of observed outcomes. This difference is important here because of treatments' serial correlation and selection into treatment based on past outcomes and covariates: a model on realized outcomes imposes restrictions on the distribution of the treatment assignments, whereas a potential outcome model does not. 
Building on the literature on marginal structural models \citep{robins2000marginal}, we identify the parameters of interest by \textit{recursively} projecting outcomes' conditional expectations over past histories, allowing for dynamic selection into treatment.

 Our estimation method, Dynamic Covariate Balancing (DCB), estimates the parameters of the model by using recursive penalized projections through lasso \citep{hastie2015statistical}. It then reweights observations to guarantee balance between treated and control units. Balancing covariates is intuitive and common in practice: in cross-sectional studies, treatment, and control units are comparable when the two groups have similar characteristics \citep{imai2014covariate, li2018balancing, hainmueller2012entropy}. 
We generalize covariate balancing in the absence of dynamics of \cite{zubizarreta2015stable, athey2018approximate, ben2018augmented, hirshberg2017augmented} to a dynamic setting. We show that balancing with potential local projections corresponds to constructing weights \textit{sequentially} in time by first balancing treated and control units' covariates in the first period and then balancing histories in the next periods \textit{reweighted} by the weights obtained in the previous period. The estimated balancing weights solve a sequence of quadratic programs to minimize the weights' variance.

Our estimation procedure guarantees a vanishing bias of order faster than $n^{-1/2}$ and a parametric rate of convergence of the estimated treatment effect in high-dimensional settings. In addition, the optimization problem over the set of balancing weights admits a feasible solution, with the true propensity score being one such solution (and without requiring knowledge of it). This result highlights the benefits of balancing over propensity score reweighting here: the proposed balancing weights have a smaller variance than inverse probability weights and -- by leveraging an (approximate) high-dimensional linear outcome model -- do not require the correct specification of the propensity score.\footnote{Typical methods in high dimensions require conditions on the product of the rates of estimators for the propensity score and coefficients of the linear model to be faster than $n^{-1/4}$, and also require consistent estimation of \textit{both} the outcome model and propensity score model \citep[what known as rate-doubly robustness, see e.g.,][]{athey2017efficient}. Compared to estimating the propensity score with a semi-parametric model, our guarantees do not depend on the estimation error of the propensity score (only require that the estimation error of the coefficients is $o(n^{-1/4})$), by leveraging the high dimensional linear outcome model. } This is an advantage especially in dynamic settings: the propensity score defines the joint probability that units are assigned to a given treatment history, and therefore inverse probability weights can exhibit large variance in finite sample (see e.g., Figure \ref{fig:overlap}). Finally, we provide guarantees for inference. Relative to cross-sectional studies, our dynamic structure necessitates novel considerations for identification, balancing, and derivations, that require analyzing joint distributions of correlated residuals from sequential projections.

   We illustrate our method in an empirical application using data from \cite{acemoglu2019democracy} on studying the effects of democracy on economic growth. Here, the authors assume a  dynamic selection model. Whereas effects are in magnitude and sign consistent with \cite{acemoglu2019democracy}, we show that standard local projections and \cite{acemoglu2019democracy}'s linear regression lead to significantly smaller point estimates compared to our approach. We also show that (A)IPW methods lead to a more substantial imbalance (and bias) compared to DCB due to the instability of the propensity score in both high and low-dimensions.

\section{Related Literature}

The goal of this paper is to conduct inference on dynamic treatment effects, while being robust to the misspecification of the propensity score. To achieve this goal, we leverage a high dimensional linear model and derive the first dynamic balancing equations within the local projection model proposed in this paper. 

In the econometrics and statistics literature,  
\cite{imbens2019panel} propose balancing assuming no treatment dynamics, whereas here, treatment dynamics require different (and novel) balancing conditions. In the context of dynamics, different from \cite{imai2015robust}, who estimate a \textit{single} set of balancing weights over all possible combinations of time periods and covariates, here the number of moment conditions grows linearly with $T$ and not exponentially. Unlike \cite{zhou2018residual}, who extend entropy balancing of \cite {hainmueller2012entropy} to dynamic settings, and \cite{li2024toward} who propose a single set of balancing by regressing each covariate on past information, we do not estimate one model for each covariate in the past (which can be prone to large estimation error in high dimensions).  
 DCB explicitly characterizes the high-dimensional model's bias in a dynamic setting to avoid overly conservative moment conditions, while \cite{kallus2018optimal} design conservative balancing conditions for the worst-case bias. Different from \cite{yiu2018covariate}, we do not require estimating the propensity score. 
 Our insight with respect to all these references (in low and high dimensions) is that with a linear model and \textit{sequential} weights, balancing reduces to few and novel dynamic restrictions. This insight is even more relevant with high-dimensional covariates, which none of these references study with dynamics.

Compared to cross-sectional studies, we generalize balancing in \cite{athey2018approximate, ben2018augmented}, and consider an arbitrary class of weights. Therefore, our residual balancing procedure does not reduce to linear estimators as in settings with linear balancing weights \citep[e.g.][]{bruns2023augmented}, and our analysis differs from cross-sectional studies with low dimensions in \cite{wang2020minimal}.

More broadly, this paper connects to the literature on 
DiD, local projections, and dynamic treatments. 
Different from the literature on DiD \citep{rambachan2023more,  de2022difference, callaway2019difference,  abraham2018estimating, athey2022design, caetano2022difference} or subsequent work on potential projections with DiD \citep{dube2023local}, here we allow for dynamic treatment regimes. This literature imposes the parallel trends assumption, violated with dynamic treatments \citep{marx2022parallel, ghanem2022selection}. 
Different from the time-series literature \citep{montiel2021local,  stock2018identification, rambachan2019nonparametric}, this paper uses information from panel data and allows for arbitrary dependence of outcomes, covariates, and treatment assignments over time. 


 References in bio-statistics include \cite{robins2000marginal}, \cite{hernan2001marginal}, \cite{boruvka2018assessing},  \cite{blackwell2013framework}, \cite{bang2005doubly} \citep[for a review, ][]{vansteelandt2014structural}. \cite{bojinov2020panel} study IPW estimators from a design-based perspective. 
Doubly robust estimators for dynamic treatments have been studied by \cite{nie2021learning, zhang2013robust, jiang2015doubly, tchetgen2012semiparametric, babino2019multiple}. Here, we focus on studying the effect of a given treatment path, as in \cite{robins2000marginal} or \cite{blackwell2013framework}, different from and complementary to studying optimal policies (e.g., \cite{murphy2003optimal}, \cite{nie2021learning}). 

Specifically, studies with high-dimensional panels require correct specification of the propensity score
 \citep{lewis2020double, zhu2017high, shi2018high, bodoryevaluating, belloni2016inference, chernozhukov2017orthogonal}, or impose homogeneous treatment effects \citep{high_dim_IRF, kock2015inference}. 
(\cite{lewis2020double} also illustrate bounds on misspecification). More
generally, prior works that formally study properties of dynamic doubly-robust methods in
high dimensions require product of rates conditions for the estimated propensity score and
conditional mean function, and consistent estimation of both; see for example follow up work
by \cite{bradic2021high} who provide tight rates of convergence of dynamic AIPW. 
 Different from above, our framework does not require consistent estimation of the propensity score.   
 
 Finally, in both works subsequent to the first version of this paper, \cite{chernozhukov2022automatic} generalize the use of riesz representers with arbitrary non-linear outcome models, and \cite{zhang2021dynamic} study doubly robustness to model misspecification through moment restrictions.
Different from these references, here we do not require conditions on the balancing weights motivated by our goal of allowing for inference with a possibly completely misspecified propensity score function, whereas \cite{chernozhukov2022automatic} and \cite{zhang2021dynamic} require functional form restrictions on the balancing weights to obtain a product of rates conditions. Our focus on the high-dimensional linear model (which we view as a linear approximation to conditional expectations in high dimensions) is motivated by its large use in applications.

\section{Dynamics and potential local projections} \label{sec:2}

\subsection{Setup} 
 
We start with the analysis of two time periods, deferring multiple periods to Section \ref{sec:multiple}. We observe a panel with $n$ $i.i.d.$ copies of $ 
\Big(  X_{i,1}, D_{i, 1}, Y_{i, 1}, X_{i, 2}, D_{i, 2}, Y_{i, 2} \Big)$, each distributed according to $\mathcal{P}$. Here
  $ D_{i,1}, D_{i,2} \in \{0,1\}$ denote binary treatments at time $t = 1,t = 2$, respectively, $X_{i,t}, Y_{i,t}$ denote covariates and the outcome at time $t$. 
We allow for any nonstationarity and dependencies that may occur over time within each unit. When indices are not specified, such as in \(D_t\), this refers to the collective observations for all $n$ units.

 \begin{figure}[!ht]
 \centering
    \begin{tikzpicture}

\coordinate (1) at (-4,3);
\coordinate (2) at (-2,3);
\coordinate (3) at (-2,5);
\coordinate (4) at (-4,5);
\coordinate (5) at ($(1)!.5!(2)$); 
\coordinate (6) at ($(2)!.5!(3)$);
\coordinate (7) at ($(3)!.5!(4)$);
\coordinate (8) at ($(1)!.5!(4)$);
\coordinate (9) at ($(1)!.5!(3)$);

\coordinate (10) at (-4,0);
\coordinate (11) at (-2,0);
\coordinate (12) at (-2,2);
\coordinate (13) at (-4,2);
\coordinate (14) at ($(10)!.5!(11)$); 
\coordinate (15) at ($(11)!.5!(12)$);
\coordinate (16) at ($(12)!.5!(13)$);
\coordinate (17) at ($(10)!.5!(13)$);
\coordinate (18) at ($(10)!.5!(14)$);

\coordinate (21) at (-10,4);
\coordinate (22) at (3.5,4);
\coordinate (23) at (3.5,5);
\coordinate (24) at (-10, 5);

\coordinate (31) at (-10,3.5);
\coordinate (32) at (-7,3.5);
\coordinate (33) at (-7,2.5);
\coordinate (34) at (-10, 2.5);

\coordinate (41) at (-6.5,3.5);
\coordinate (42) at (-3.5,3.5);
\coordinate (43) at (-3.5,2.5);
\coordinate (44) at (-6.5, 2.5);

\coordinate (51) at (-3,3.5);
\coordinate (52) at (0,3.5);
\coordinate (53) at (0,2.5);
\coordinate (54) at (-3, 2.5);

\coordinate (61) at (0.5,3.5);
\coordinate (62) at (3.5,3.5);
\coordinate (63) at (3.5,2.5);
\coordinate (64) at (0.5, 2.5);

\coordinate (71) at (-10,2);
\coordinate (72) at (-8.6,2);
\coordinate (73) at (-8.6,1);
\coordinate (74) at (-10, 1);

\coordinate (81) at (-8.4,2);
\coordinate (82) at (-7,2);
\coordinate (83) at (-7,1);
\coordinate (84) at (-8.4, 1);

\coordinate (91) at (-6.5,2);
\coordinate (92) at (-5.1,2);
\coordinate (93) at (-5.1,1);
\coordinate (94) at (-6.5, 1);

\coordinate (101) at (-4.9,2);
\coordinate (102) at (-3.5,2);
\coordinate (103) at (-3.5,1);
\coordinate (104) at (-4.9, 1);

\coordinate (111) at (-3,2);
\coordinate (112) at (-1.6,2);
\coordinate (113) at (-1.6,1);
\coordinate (114) at (-3, 1);

\coordinate (121) at (-1.4,2);
\coordinate (122) at (0,2);
\coordinate (123) at (0,1);
\coordinate (124) at (-1.4, 1);

\coordinate (131) at (0.5,2);
\coordinate (132) at (1.9,2);
\coordinate (133) at (1.9,1);
\coordinate (134) at (0.5, 1);

\coordinate (141) at (2.1,2);
\coordinate (142) at (3.5,2);
\coordinate (143) at (3.5,1);
\coordinate (144) at (2.1, 1);


   

\draw[->] (-8,4.3)  -- (-1,4.3);
 

  \node[circle] (g) at (-8,4.3) {$|$};
   \node[circle] (g) at (-5,4.3) {$|$};
     \node[circle] (g) at (-2,4.3) {$|$};
  \node[circle] (g) at (-8,4) {$t = 0$};
   \node[circle] (g) at (-5,4) {$t = 1$};
    \node[circle] (g) at (-2,4) {$t = 2$};
   \node[circle] (g) at (-8,5) {$X_1,$};
    \node[circle] (g) at (-6.8,5) {$D_1,$};
    \node[circle] (g) at (-5,5) {$(Y_1, X_2),$};
    \node[circle] (g) at (-3.3,5) {$D_2,$};
      \node[circle] (g) at (-2,5) {$Y_2$};






    \end{tikzpicture}
\caption{Sampling process in two periods. First, baseline covariates $X_{i,1}$ realize at $t = 0$. Then, treatment $D_{i,1}$ is assigned and the outcomes and covariates $(Y_{i,1},X_{i,2})$ realize at $t = 1$. Finally, the treatment $D_{i,2}$ is assigned and, afterwards, the endline outcome $Y_{i,2}$ realizes.  } \label{fig:time}
\end{figure}

We consider potential outcomes that are functions of the entire treatment history with $Y_{i,2}(d_1, d_2)$  denoting the potential outcome at time $t = 2$, under treatment $d_1$ in the first  and $d_2$ in the second period.
Our goal is to conduct inference on the estimand(s)
$$
\small 
\begin{aligned} 
\mathrm{ATE}(d_{1:2}, d_{1:2}') = \mu_2(d_1, d_2) - \mu_2(d_1', d_2'), \quad 
\mu_2(d_1, d_2) = \mathbb{E}\Big[Y_{i,2}(d_1, d_2)\Big], 
\end{aligned} 
$$ 
for given treatment histories $(d_1, d_2), (d_1', d_2')$. For example, researchers may be interested in estimating $
\mathrm{ATE}((1,1), (0,0)),
$
which denotes the \textit{total} effect of treating an individual for two consecutive periods \citep{athey2022design}; or the \textit{direct} effect $\mathrm{ATE}((1, 0), (0,0))$. 
 Figure \ref{fig:seqign} shows that the overall treatment effects capture the direct effect of the treatment on the outcomes and the indirect effect. For longer histories, one could also consider weighted combinations of relevant treatment effects, omitted for brevity (see Section \ref{sec:app}).

\begin{figure}[!ht]
\centering 
\scalebox{0.7}{
    \begin{tikzpicture}[scale = 1.16]
    \node[draw, black,ultra thick, inner sep=0pt,
  text width=14mm,
  align=center,   circle] (h) at (-3,-2) {$D_1$};
  
    \node[draw, black,ultra thick, inner sep=0pt,
  text width=14mm,
  align=center, circle] (e) at (-0.5,-2) {$Y_1, X_2$};
 
  \node[draw, black,ultra thick, inner sep=0pt,
  text width=14mm,
  align=center, circle] (f) at (-0.5,-4) {$D_2$};

   \node[draw, black,ultra thick, inner sep=0pt,
  text width=13.5mm,
  align=center, circle] (d) at (2,-2) {$Y_2$};

    \draw[->, -triangle 90]     (e) edge (f) (h) edge (f);
      \draw[->, -triangle 90]    (h) edge (e) (e) edge (d) (f) edge (d) ;
   \draw[->,  -triangle 90]  (h) edge[bend right=-30] node [left] {} (d);

    \end{tikzpicture}}
    \scalebox{0.7}{
        \begin{tikzpicture}[scale = 1.16]
    \node[draw, black,ultra thick, inner sep=0pt,
  text width=14mm,
  align=center,   circle] (h) at (-3,-2) {$D_1$};
  
    \node[draw, black,ultra thick, inner sep=0pt,
  text width=14mm,
  align=center, circle] (e) at (-0.5,-2) {$Y_1, X_2$};
 
  \node[draw, black,ultra thick, inner sep=0pt,
  text width=14mm,
  align=center, circle] (f) at (-0.5,-4) {$D_2$};

   \node[draw, black,ultra thick, inner sep=0pt,
  text width=13.5mm,
  align=center, circle] (d) at (2,-2) {$Y_2$};

    \draw[->, -triangle 90]     (e) edge (f) (h) edge (f)  ;
      \draw[->, -triangle 90, red]    (h) edge (e) (e) edge (d) ;
   \draw[->,  -triangle 90, red]  (h) edge[bend right=-30] node [left] {} (d);
 \draw[->, -triangle 90, red, dotted]     (f) edge (d) ;

    \end{tikzpicture} }
    \caption{The left panel illustrates all the possible causal paths under Sequential Ignorability (Assumption \ref{ass:seqign}). Here, past treatments may affect intermediate covariates, and future treatments may depend on past treatments, covariates and outcomes. The right panel presents two estimands of interest. In particular,  $\mathrm{ATE}(\mathbf{1}, \mathbf{0})$ (the effect of increasing treatments in both periods) denotes the effect mediated through all red edges, including the dotted red edge. Instead, $\mathrm{ATE}((1, 0), (0, 0))$ (the \textit{direct} effect of only increasing treatment in the first period) denotes the effect mediated through all red edges excluding the dotted red edge. } \label{fig:seqign}
    \end{figure}

\subsection{Dynamic treatment assignments} 

 
Treatment histories can impact both outcomes and covariates at intermediate stages. Let \(Y_{i,1}(d_1, d_2)\) represent the intermediate potential outcome and \(X_{i,2}(d_1, d_2)\) represent the potential covariates following a sequence of \(d_1\) then \(d_2\). Here, \(X_{i,1}\) refers to the baseline covariates.
 
     \begin{ass} \label{ass:noant} For $d_1 \in \{0,1\}$, let 
     $Y_{i,1}(d_1, 1) = Y_{i,1}(d_1, 0)$, $X_{i,2}(d_1,1) = X_{i,2}(d_1,0)$.
      \end{ass}  
      
     Assumption \ref{ass:noant} is a no-anticipation restriction: (i) intermediate potential outcomes only depend on past but not future treatments; (ii) the treatment status at $t = 2$ has no contemporaneous effect on covariates.

      Assumption \ref{ass:noant} allows for anticipatory effects governed by \textit{expectations} (e.g., individuals may choose treatments based on \textit{expected} future utilities), but not on the future treatment \textit{realizations} \citep[see][for a discussion]{athey2022design}. 
       
\begin{exmp}[Observed outcomes] \label{exmp:a_0}
Consider a dynamic model of the form  
$$
\small 
\begin{aligned} 
Y_{i,2} = g_2\Big(Y_{i,1}, X_{i,1}, X_{i,2}, D_{i,1}, D_{i,2}, \varepsilon_{i,2}\Big), \quad Y_{i,1} = g_1\Big(X_{i,1}, D_{i,1}, \varepsilon_{i,1}\Big),  \quad X_{i,2} = g_0\Big(X_{i,1}, D_{i,1}, \varepsilon_{i, X}\Big)
\end{aligned} 
$$  
for some arbitrary functions $g_2(\cdot), g_1(\cdot), g_0(\cdot)$ and unobservables 
$(\varepsilon_{i,2}, \varepsilon_{i,1}, \varepsilon_{i,X})$, with 
$
\varepsilon_{i,2} \perp D_{i,2} | Y_{i,1}, X_{i,1}, X_{i,2}, D_{i,1}, $ and $ (\varepsilon_{i,X}, \varepsilon_{i,1}) \perp D_{i,1} | X_{i,1}. 
$
 We can write 
$$ 
\small 
\begin{aligned} 
Y_{i,2}(d_1, d_2) = g_2\Big(Y_{i,1}(d_1), X_{i,1}, X_{i,2}(d_1), d_1, d_2, \varepsilon_{i,2}\Big), 
\end{aligned} 
$$
where $Y_{i,1}(d_1) = g_1\Big(X_{i,1}, d_1, \varepsilon_{i,1}\Big), X_{i,2} = g_0\Big(X_{i,1}, d_1, \varepsilon_{i, X}\Big)$. 
Since $g_1(\cdot), g_0(\cdot)$ are not functions of $d_2$, Assumption \ref{ass:noant} holds. (Assumption \ref{ass:linearity} below will impose restrictions on $\mathbb{E}[g_2(\cdot)]$.)
\qed 
\end{exmp}       
      
In the rest of our discussion, we index potential outcomes and covariates by past treatment history under Assumption \ref{ass:noant}. We define  $
H_{i,2} = \Big[D_{i,1}, X_{i,1}, X_{i,2}, Y_{i,1}\Big], 
$
 the vector of past treatment assignments, covariates, and outcomes in the previous period. We refer to 
 $
 H_{i,2}(d_1) = \Big[d_1, X_{i,1}, X_{i,2}(d_1), Y_{i,1}(d_1)\Big]
$ as the \textit{potential history} under treatment status $d_1$ in the first period. Here, $H_{i,2}$ can include interaction terms, omitted for brevity.

\begin{ass}[Sequential Ignorability] \label{ass:seqign} Assume that for all $(d_1, d_2) \in \{0,1\}^2$ ,  
$$
\small 
\begin{aligned} 
(A) \quad &Y_{i,2}(d_1, d_2) \perp D_{i,2} \Big | D_{i,1}, X_{i,1}, X_{i,2}, Y_{i,1}, \quad  
(B) &\Big(Y_{i,2}(d_1, d_2), H_{i,2}(d_1)\Big) \perp D_{i,1} \Big | X_{i,1},   
\end{aligned} 
$$
\end{ass} 

Sequential ignorability states that treatment in the first period is unconfounded conditional on baseline covariates, and the treatment in the second period is unconfounded conditional on all observable characteristics at $t= 2$. It assumes no unobserved factors after controlling for high dimensional observable characteristics and arbitrary past information. Note that we could also state (A), conditioning on $D_{i,1} = d_1$ and potential history $H_{i,1}(d_1)$.

In Example \ref{exmp:a_0}, Assumption \ref{ass:seqign} holds if 
$
D_{i,2} \perp \varepsilon_{i,2} \Big| D_{1, i}, X_{i,1}, X_{i,2}, Y_{i,1}, \quad D_{i, 1} \perp (\varepsilon_{i,1}, \varepsilon_{i,2}) \Big| X_{i,1}. 
$

\subsection{Potential local projections}


Following in spirit, \cite{jorda2005estimation}, we approximate the expectation of potential outcomes as linear functions of (high-dimensional) past characteristics. Different from \cite{jorda2005estimation}, linearity is imposed on expected potential instead of realized outcomes. Modeling potential outcomes directly avoids functional form restrictions on the treatment assignment mechanism.

\begin{ass} \label{ass:linearity} 
For some $\beta_{d_1, d_2}^{(1)} \in \mathbb{R}^{p_1}, \beta_{d_1, d_2}^{(2)} \in \mathbb{R}^{p_2}$ 
$$
\small 
\begin{aligned} 
& \mathbb{E}\Big[Y_{i,2}(d_1, d_2)\Big| X_{i,1} = x_1\Big]  = x_1\beta_{d_1, d_2}^{(1)}, \quad  \\ & \mathbb{E}\Big[Y_{i,2}(d_1, d_2) \Big| X_{i,1} = x_1, X_{i,2} = x_2, Y_{i,1} = y_1, D_{i,1} = d_1\Big]  = \Big[d_1, x_1, x_2, y_1\Big] \beta_{d_1, d_2}^{(2)}.  
\end{aligned} 
$$ 
\end{ass}

Assumption \ref{ass:linearity} allows for heterogeneity in $(d_1, d_2)$, and the dimensions $p_1, p_2$ can grow with $n$ (because of additional covariates and/or covariates transformations). 
As for MSMs \citep{robins2000marginal}, Assumption \ref{ass:linearity} (i) does not require estimating a structural model for each time-varying-covariate, that would be prone to large estimation error in high dimensions; and (ii) it is agnostic on the treatment assignment mechanism because the model is imposed on potential outcomes. Coefficients can vary with time in the model. 



\begin{lem}[Identification] \label{lem:identification_model1} Let Assumptions \ref{ass:noant}, \ref{ass:seqign}, \ref{ass:linearity} hold. Then
$$
\small 
\begin{aligned} 
& \mathbb{E}\Big[Y_{i,2} \Big| H_{i,2}, D_{i,2} = d_2, D_{i,1} = d_1\Big] = \mathbb{E}\Big[Y_{i,2}(d_1, d_2) \Big| H_{i,2}, D_{i,1} = d_1\Big] = 
 H_{i,2}(d_1) \beta_{d_1, d_2}^{(2)}  \\ 
& \mathbb{E}\Big[\mathbb{E}\Big[Y_{i,2} \Big| H_{i,2}, D_{i,2} = d_2, D_{i,1} = d_1\Big] \Big| X_{i,1}, D_{i,1} = d_1\Big]  =
 \mathbb{E}\Big[Y_{i,2}(d_1, d_2) \Big| X_{i,1} \Big] = X_{i,1} \beta_{d_1, d_2}^{(1)}. 
\end{aligned} 
$$ 
\end{lem}

The proof is in Appendix \ref{sec:lemma1}. Lemma \ref{lem:identification_model1} builds on results in the literature on marginal structural models \citep[e.g.][]{robins2000marginal, bang2005doubly, tran2019double, kallus2020double}, where, here, we make a connection between marginal structural models and local projections in economics as a contribution of independent interest. Lemma \ref{lem:identification_model1} motivates a recursive estimation strategy discussed in Section \ref{sec:two_p}.   


\begin{exmp}[Linear Model] \label{exmp:sem1}
Let $X_{i,1}, X_{i,2}$ also contain an intercept. Let 
$\mathbb{E}\Big[Y_{i,1}(d_1) \Big| X_{i,1}\Big] = X_{i,1} \alpha_{d_1}$,
$\mathbb{E}\Big[ X_{i,2}(d_1) \Big|X_{i,1}\Big] =  W_{d_1} X_{i,1}$, $\mathbb{E}\Big[Y_{i,2}(d_1, d_1) \Big| X_{i,1}, X_{i,2}, Y_{i,1}, D_{i,1} = d_1\Big] = \\ \Big(X_{i,1}, X_{i,2}(d_1), Y_{i,1}(d_1)\Big) \beta_{d_1, d_2}^{(2)},  
$ 
for some arbitrary parameters $\alpha_{d_1} \in \mathbb{R}^{p_1}$ and $\beta_{d_1, d_2}^{(2)}   \in \mathbb{R}^{p_2}$; $W_{d_1}, V_{d_1} $ denote unknown matrices in $ \mathbb{R}^{p_2 \times p_1}$ . The model satisfies Assumption \ref{ass:linearity}. 
\qed    
 \end{exmp}

\begin{rem}[Linearity in high-dimensions as an approximation to the true model] \label{rem:approximation} In the same spirit of \cite{belloni2014inference}, 
our results also directly extend to the case where we relax Assumption \ref{ass:linearity} and assume only approximate linearity up to an order $\mathcal{O}_p(r_p), $ where $r_p$ is an arbitrary sequence which depends on $p$ with $r_p = o(n^{-1/2})$. This setting embeds empirical applications where many covariates (and their transformation) can \textit{approximate} the conditional mean function as linear. As we further discuss in Section \ref{sec:theory}, we consider a high-dimensional setting where we will only require weak conditions on the estimated coefficients $||\hat{\beta}_{d_1, d_2}^{(t)} - \beta_{d_1, d_2}^{(t)}||_1 = O_p(n^{-1/4})$ with unbounded covariates and $||\hat{\beta}_{d_1, d_2}^{(t)} - \beta_{d_1, d_2}^{(t)}||_1 = o_p(1/\log(n))$ with bounded covariates, common for standard high-dimensional estimators.
\qed 
\end{rem} 

\begin{rem}[Comparison with standard local projections and DiD] Appendix \ref{sec:lp1} presents an extensive discussion and comparison of Lemma \ref{lem:identification_model1} with standard local projections and DiD common in economics that, as we show, would return biased estimates in a dynamic context. The reason is because standard local projections impose a linear model on the \textit{observed} outcomes $Y_{i,T}$ instead of potential outcomes, which therefore would also depend on the distribution of treatment assignments.  
\end{rem}

\section{Estimation with dynamic balancing}  \label{sec:two_p}

\subsection{Estimation with two periods} 

This section studies estimation. We defer to Section \ref{sec:app} a complete guide for practice, including discussion about the model, tuning parameters, and complexity. Appendix \ref{sec:extended_discussion} presents a more detailed description of each step to construct the estimator. 

Consider a two periods setting first, with $T = 2$. 
The estimator for $\mu(d_1,d_2)$ (and symmetrically for $\mu(d_1', d_2')$) proceeds in the following steps: 

\begin{itemize}
\item \textit{Estimation of the coefficients:} We estimate the coefficients $\hat{\beta}_{d_1, d_2}^{(2)}$ by regressing $Y_2$ onto $H_2$ (controlling for $D_2 = d_2$). Following Lemma \ref{lem:identification_model1}, we then regress $H_2 \hat{\beta}_{d_1,d_2}^{(2)}$ onto $X_1$ (controlling for $D_1 = d_1$) to estimate $\hat{\beta}_{d_1,d_2}^{(1)}$. These regression may allow for high-dimensional coefficients as for lasso. 
The full algorithm is in Algorithm 2. 
\vspace{2mm} 

\item \textit{Sequential estimation via regression adjustments:} Because linearity may only hold as we control for high-dimensional covariates, we cannot directly use our predictions $H_2 \hat{\beta}^{(2)}$ or $X_1 \hat{\beta}^{(1)}$ for valid causal inference. We instead must guarantee a vanishing high-dimensional bias through reweighting. 

\vspace{2mm} 

For weights in each period $\hat{\gamma}_2(d_1,d_2)$ and $\hat{\gamma}_1(d_1,d_2) \in \mathbb{R}^n$, denoting $\bar{X}_1$ the sample mean of $X_1$, an equivalent of a AIPW-type estimator takes the form \citep[e.g.][]{tchetgen2012semiparametric, zhang2013robust, jiang2015doubly, nie2021learning} 
\begin{equation} \label{eqn:myestimator}
\small 
\begin{aligned} 
\hat{\mu}_2(d_1, d_2; \hat{\gamma}_1, \hat{\gamma}_2) &= \hat{\gamma}_2(d_{1:2})^\top \Big(Y_2 - H_2 \hat{\beta}_{d_{1:2}}^{(2)}\Big) + \hat{\gamma}_1(d_{1:2})^\top \Big(H_2 \hat{\beta}_{d_{1:2}}^{(2)} - X_1 \hat{\beta}_{d_{1:2}}^{(1)} \Big) + \bar{X}_1 \hat{\beta}_{d_{1:2}}^{(1)}. 
\end{aligned} 
\end{equation}
We will omit the arguments $(\hat{\gamma}_1, \hat{\gamma}_2)$ in $\hat{\mu}_2$ whenever clear. Its construction directly follows from properties of influence functions \citep{tchetgen2012semiparametric}. To gain insights, note that 
a simple estimator for $\mathbb{E}[Y_2(d_1, d_2)]$ is $\bar{X}_1 \hat{\beta}_{d_1, d_2}^{(1)}$. This estimator is consistent and asymptotically normal for low dimensional $\hat{\beta}_{d_1, d_2}^{(1)}$, but not in high-dimensional settings. 
Instead, the estimator in \eqref{eqn:myestimator} uses regression adjustments over \textit{each} period to control the high-dimensional bias. 

\vspace{2mm} 

A choice of the weights from previous literature are inverse probability weights (IPW). These weights for the first and second period are
\begin{equation} \label{eqn:ipw} 
\small 
\begin{aligned} 
\frac{1\{D_{i,1} = d_1\}}{n P(D_{i,1}  = d_1 | X_{i,1})}, \quad \frac{1\{D_{i,1} = d_1\}}{n P(D_{i,1}  = d_1 | X_{i,1})} \times  \frac{1\{D_{i,2} = d_2\}}{P(D_{i,2} = d_2 | Y_{i,1}, X_{i,1}, X_{i,2}, D_{i,1})}. 
\end{aligned} 
\end{equation}
 
However, in high dimensions, AIPW weights require the correct specification of both the propensity score, which in practice may be unknown, and  also of the conditional mean function (see Remark \ref{rem:aipw}). Also, in small sample, IPW are sensitive to poor overlap (high variance). 
  Motivated by these considerations, we leverage the linear  structure to replace IPW with more stable balancing weights that we introduce here.

\vspace{2mm} 

\item \textit{Main balancing conditions:} The main step is in the choice of the balancing weights. The core idea is to decompose 
 \begin{equation} \label{eq:decomposition}
  \small 
  \begin{aligned} 
  \hat{\mu}_2(d_1, d_2)  =  \bar{X}_1\beta_{d_1, d_2}^{(1)}+ T_1 + T_2 +T_3,
  \end{aligned} 
\end{equation} 
where $\bar{X}_1\beta_{d_1, d_2}^{(1)}$ converges to $\mu(d_1,d_2)$ under standard $\sqrt{n}$-asymptotics, and  
$$ 
\small 
\begin{aligned} 
T_2 =  \hat{\gamma}_2(d_1, d_2)^\top \Big[Y_2 - H_2 \beta_{d_1,d_2}^{(2)}\Big] , \quad T_3=   \hat{\gamma}_1(d_1, d_2)^\top \Big[H_2 \beta_{d_1,d_2}^{(2)} -  X_1 \beta_{d_1, d_2}^{(1)} \Big] 
\end{aligned} 
$$  
do not depend on the estimation error for $\beta$. 
The remaining term $T_1$ is the key component that determines the high-dimensional bias due to the estimation error of $\hat{\beta}$. In particular, writing \small $T_1 =  \Big( \hat{\gamma}_1(d_1, d_2)^\top X_1 - \bar{X}_{1}\Big)  (\beta_{d_1, d_2}^{(1)} - \hat{\beta}_{d_1, d_2}^{(1)})
+ \Big(\hat{\gamma}_2(d_1, d_2)^\top H_2 - \hat{\gamma}_1(d_1, d_2)^\top H_2\Big)   (\beta_{d_1, d_2}^{(2)} - \hat{\beta}_{d_1, d_2}^{(2)})$, \normalsize we have
\begin{equation} \label{eqn:eqbalance}
\small 
\begin{aligned} 
T_1 \le&
  \underbrace{\| \hat{\beta}_{d_1,d_2}^{(1)} - \beta_{d_1, d_2}^{(1)}\| _1 \Big| \Big| \bar{X}_1 - \hat{\gamma}_1(d_1, d_2)^\top X_1 \Big| \Big|_{\infty}}_{(i)} +
 \underbrace{\| \hat{\beta}_{d_1,d_2}^{(2)} - \beta_{d_1, d_2}^{(2)}\| _1 \Big| \Big| \hat{\gamma}_2(d_1, d_2)^\top H_2 - \hat{\gamma}_1(d_1, d_2)^\top H_2 \Big| \Big|_{\infty}}_{(ii)}.
 \end{aligned} 
\end{equation} 

  The estimation error depends on the \textit{product} between the imbalance of covariates characterized by the expressions in $(i), (ii)$ and the estimation error of the coefficients, in the spirit of strong doubly-robustness properties. 
 Therefore a key insight is that to make the estimation error of $\hat{\beta}$ asymptotically negligible over each period, we want to guarantee that 
 \begin{equation} \label{eqn:bb} 
 \small 
 \begin{aligned} 
 \Big| \Big| \bar{X}_1 - \hat{\gamma}_1(d_1, d_2)^\top X_1 \Big| \Big|_{\infty}, \quad   \Big| \Big| \hat{\gamma}_2(d_1, d_2)^\top H_2 - \hat{\gamma}_1(d_1, d_2)^\top H_2 \Big| \Big|_{\infty}  
 \end{aligned} 
 \end{equation} 
 are sufficiently small. 
 The first term in Equation \eqref{eqn:bb} coincides with the balancing term in \cite{athey2018approximate}. However, 
 here we also require that histories in the second period are balanced, once \textit{reweighted} by the weights in the previous period.  This motivates the dynamic (sequential) balancing weights. To our knowledge, this paper is the first to derive such dynamic balancing conditions. 

\vspace{2mm}

\item \textit{Additional conditions on the weights:} 
The remaining two terms $T_2, T_3$
are mean zero under two conditions: (i) the weights $\hat{\gamma}_2$ are only functions of $H_2$ (and therefore $D_1,D_2$) but not functions of $Y_2$, and weights $\hat{\gamma}_1$ are only functions of $(D_1,X_1)$; (ii) $\hat{\gamma}_{i,2}$ differs from zero only for units with $\{D_{i,2} = d_2, D_{i,1} = d_1\}$ and $\hat{\gamma}_{i,1}$ differs from zero only for units with $D_{i,1} = d_1$. A special case of weights satisfying such conditions are IPW. 


 \vspace{2mm} 
 
\item \textit{Complete algorithm:} Algorithm 1 presents the algorithmic details in generic $T$ periods. We choose weights sequentially: each period $t$ we minimize the $l_2$ norm of the weights to find stable weights. Such weights are non-zero only for units with treatment history equal to the target path $d_{1:t}$ up-to time $t$; they sum to one and are bounded not to assign too large weight to a few units. The main balancing condition requires that the history $H_{i,t}$ are balanced once reweighting $H_{i,t}$ by the balancing weights estimated in the previous period $t-1$.

\end{itemize}

We formalize this discussion below (Appendix \ref{lem:33} contains a formal proof). 

\begin{thm}[Balancing weights] \label{thm:residual_final} \label{lem:balancing1} Let assumptions \ref{ass:noant} - \ref{ass:linearity} hold. Then Equation \eqref{eq:decomposition} holds with $T_1$ bounded as in Equation \eqref{eqn:eqbalance}. 

In addition, suppose that $\hat{\gamma}_1$ is measurable with respect to the sigma algebra $\sigma(X_1, D_1)$ and $\hat{\gamma}_2$ is measurable with respect to the sigma algebra $\sigma(X_1, X_2, Y_1, D_1, D_2)$. Suppose in addition that $\hat{\gamma}_{i,1}(d_1, d_2) = 0$ if $D_{i,1} \neq d_1$ and $\hat{\gamma}_{i,2}(d_1, d_2) = 0$ if $(D_{i,1}, D_{i,2}) \neq (d_1, d_2)$. Then 
$
\mathbb{E}\Big[T_2 \Big| X_1, D_1, Y_1, X_2, D_2\Big] = \mathbb{E}\Big[T_3 \Big| X_1, D_1\Big] = 0.    
$
\end{thm}

\begin{rem}[Estimating when-to-treat policies] \label{rem:when_to_treat} The goal and focus of this paper is to study the effect of a specific counterfactual treatment-assignment sequence, often the object of interest in applications, e.g., \cite{robins2000marginal}, \cite{imai2015robust}, \cite{acemoglu2019democracy}. 
This differs from studying the effect of an optimal treatment path as in \cite{murphy2003optimal} or \cite{nie2021learning}. To gain further insights on the latter problem, define $\pi_t: \mathbb{R}^{p_t} \mapsto \{0,1\}$ a binary policy decision with $\pi(H_{i,t}) \in \{0,1\}$ indicating whether to treat individual $i$ given an observed history $H_{i,t}$ as defined in Equation \eqref{eqn:H_it}, and 
$$
\small 
\begin{aligned} 
V(\pi) = \mathbb{E}\left[Y_{i,2}(\pi)\right], \quad Y_{i,2}(\pi):= Y_{i,2}\Big(\pi_1(X_{i,1}), \pi_2(H_{i,2}(\pi_1(X_{i,1}))\Big) 
\end{aligned} 
$$ 
as the expected value of potential outcome $Y_{i,2}$ evaluated under a treatment trajectory $\pi_1(X_{i,1}), \pi_2(H_{i,2}(\pi(X_{i,1})))$. 
Different from our framework, where the choice of the balancing weight depends on the model for $\mathbb{E}[Y_{i,2}(d_1, d_2) | X_{i,1}]$, in this case the balancing weights must depend on the 
model for $\mathbb{E}[Y_{i,2}(\pi)| X_{i,1}]$. In the latter case, we must integrate over $\pi_2(H_{i,2}(\pi_1))$ as a function of $H_{i,2}$.  Appendix \ref{sec:appendix_when_to_treat} discusses an extension for this setting. 
\end{rem}

\subsection{Generalization to multiple time periods} \label{sec:multiple}

\begin{figure}[!ht]
\centering
\includegraphics[scale=0.8, page = 1, 
trim={1cm 17cm 1cm 2.5cm},clip]{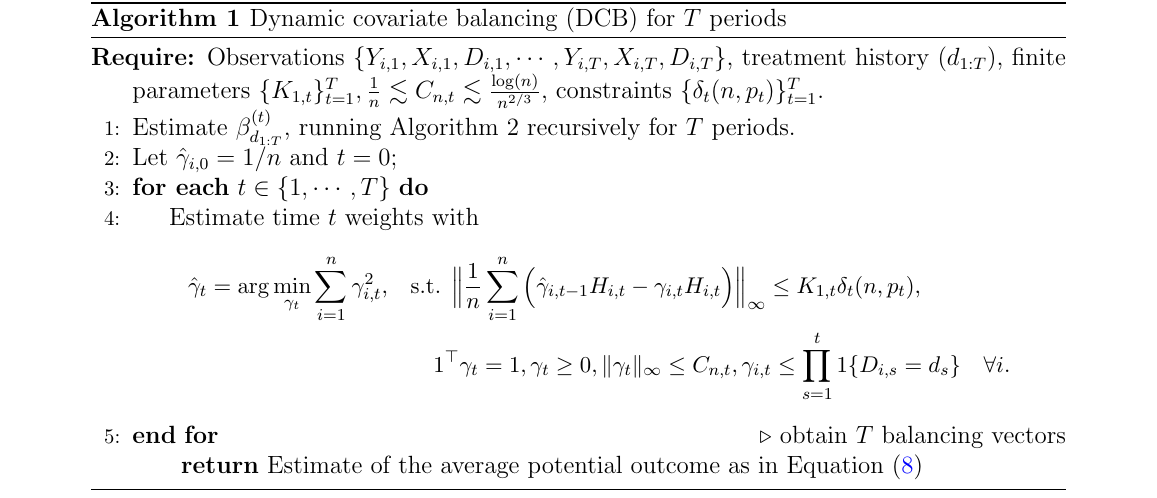} 
\end{figure}


We now describe in details the procedure with finite $T$ periods. Let $d_{1:T} = (d_1, \cdots, d_T)$, 
\begin{equation}
\small 
\begin{aligned} 
\mathrm{ATE}(d_{1:T}, d_{1:T}') = \mu_T(d_{1:T}) - \mu_T(d_{1:T}'), \quad   \mu_T(d_{1:T}) =  \mathbb{E}\Big[Y_T(d_{1:T}) \Big]. 
\end{aligned} 
\end{equation}
This estimand denotes the difference in potential outcomes for two treatment histories $d_{1:T}, d_{1:T}'$. We denote 
\begin{equation} \label{eqn:H_it}
\small 
\begin{aligned} 
H_{i,t} = \Big[D_{i,1}, \cdots , D_{i,t - 1}, X_{i,1}, \cdots, X_{i,t}, Y_{i,1}, \cdots,  Y_{i,t-1} \Big] \in \mathbb{R}^{p_t}
\end{aligned} 
\end{equation}  
the vector containing information from time one to time $t$, after excluding the treatment assigned in the present period $D_t$. Interaction components may also be considered, omitted here for brevity. 
We let the potential history be
$
H_{i,t}(d_{1:(t-1)}) = \Big[d_{1:(t-1)}, X_{i,1:t}(d_{1:(t-1)}), Y_{i,1:(t-1)}(d_{1:(t-1)}) \Big]. 
$

\begin{ass} \label{ass:seqignm} For any $d_{1:T} \in \{0,1\}^T$,  and $t \le T$, 
\begin{itemize} 
\item[(A)] (No-anticipation) The potential history $H_{i,t}(d_{1:T})$ is constant in $d_{t:T}$;
\item[(B)]   (Sequential ignorability) $\Big(Y_{i,T}(d_{1:T}), H_{i,t+1}(d_{1:(t+1)}), \cdots, H_{i,T-1}(d_{1:(T-1)})\Big) \perp D_{i,t} | H_t$;
\item[(C)]  (Potential projections) For some $\beta_{d_{1:T}}^{(t)} \in \mathbb{R}^{p_t}$,  
$$
\small 
\begin{aligned} 
\mathbb{E}\Big[Y_{i,T}(d_{1:T}) | D_{i, 1:(t-1)} = d_{1:(t-1)}, X_{i,1:t}, Y_{i,1:(t-1)}\Big] = H_{i,t}(d_{1:(t-1)})  \beta_{d_{1:T}}^{(t)}.
\end{aligned}
$$
\end{itemize} 
\end{ass} 
Assumption \ref{ass:seqignm} generalizes Assumptions \ref{ass:noant}-\ref{ass:linearity} from the two-period setting.  Identification follows similarly to Lemma \ref{lem:identification_model1}. For given weights $\hat{\gamma}_{1:T}$, and coefficients $\hat{\beta}^{(1:T)}$, we estimate
\begin{equation} \label{eqn:generalest}
\small 
\begin{aligned} 
\hat{\mu}_T(d_{1:T}) = & \sum_{i=1}^n \left\{\hat{\gamma}_{i,T}(d_{1:T}) Y_{i,T} -   \sum_{t = 2}^T \Big(\hat{\gamma}_{i,t}(d_{1:T}) - \hat{\gamma}_{i,t-1}(d_{1:T})\Big) H_{i,t} \hat{\beta}_{d_{1:T}}^{(t)} -  \Big(\hat{\gamma}_{i,1}(d_{1:T}) - \frac{1}{n} \Big) X_{i,1} \hat{\beta}_{d_{1:T}}^{(1)} \right\}.
\end{aligned} 
\end{equation}

Coefficients are estimated recursively as in the two periods setting (see Algorithm 2). 

\begin{lem} \label{lem:lemmam1} 
Let $\hat{\gamma}_{i,T}(d_{1:T}) = 0$ if $D_{i,1:T} \neq d_{1:T}$. For $\varepsilon_{i,t}(d_{1:T}) =  Y_{i,T}(d_{1:T}) - H_{i,t}(d_{1:(t-1)})  \beta_{d_{1:T}}^{(t)}$, 

\begin{equation}
\small 
\begin{aligned} 
\hat{\mu}_T(d_{1:T}) - \bar{X}_1 \beta_{d_{1:T}}^{(1)} = 
&\underbrace{\sum_{t = 1}^T \Big(\hat{\gamma}_t(d_{1:T}) H_t - \hat{\gamma}_{t-1} (d_{1:T})H_t\Big) (\beta_{d_{1:T}}^{(t)} - \hat{\beta}_{d_{1:T}}^{(t)})}_{(I_1)} 
+ \underbrace{\hat{\gamma}_T^\top(d_{1:T}) \varepsilon_T}_{(I_2)}  \\ 
&
+ \underbrace{\sum_{t=2}^T \hat{\gamma}_{t-1}(d_{1:T})\Big(H_t \beta_{d_{1:T}}^{(t)} - H_{t-1} \beta_{d_{1:T}}^{(t-1)}\Big)}_{(I_3)} 
\end{aligned} 
\end{equation}  
 
\end{lem} 
The proof is in Appendix \ref{app:lem22}. 
Lemma \ref{lem:lemmam1} decomposes the estimation error into three components.  First, ($I_1$), depends on the estimation error of the coefficient and on balancing properties of the weights. ($I_1$) suggests imposing balancing conditions on \\ 
$ 
 \Big| \Big| \hat{\gamma}_t(d_{1:T}) H_t - \hat{\gamma}_{t-1}(d_{1:T}) H_t \Big| \Big|_{\infty}
$
each period. The components characterizing the estimation error are $ (I_2)=\hat{\gamma}_T(d_{1:T})^\top \varepsilon_T$,  and ($I_3$), which are mean zero as described in Appendix Lemma \ref{lem:lemmam2}.  Note that $\bar{X}_1 \beta_{d_{1:T}}^{(1)}$ converges to $\mu_T(d_{1:T})$ under standard $\sqrt{n}$-asymptotics.

\section{Theoretical properties and inference} \label{sec:theory}

Next, we study the theoretical properties of the estimator in finite $T$ periods. We consider a high dimensional regime where the dimension covariates in each period $p_1, \cdots, p_T$ can grow to infinity, as long as  $\log (\max_t p_t n)/n^{1/4} \to 0$. We impose the following conditions.

\begin{ass}[Overlap and tails' conditions] \label{ass:weakoverlap}  
Assume that (i) ${P(D_{i,t} = d_t | H_{i,t})}  \in (\delta,1 - \delta), \delta \in (0,1)$ for each $t \in \{1, \cdots, T\}$; and (ii) $H_{i,t}^{(j)}, \forall j$ is Sub-Gaussian given $H_{i,t-1}$ and $X_{i,1}^{(j)}, j \in \{1, \cdots, p_1\}$ is Sub-Gaussian.
 \end{ass}

Condition (i) is the overlap condition, standard in the causal inference literature. The overlap condition is sufficient (but not necessary, see the discussion of \cite{athey2018approximate} in cross-sectional settings) to show existence of a feasible solution of Algorithm 1; see Remark \ref{rem:overlap}. Condition (ii) is a tail restriction. Assumption \ref{ass:weakoverlap} can be relaxed by assuming that the product of the inverse probability weights times the covariates is sub-exponential  at the expense of more tedious derivations.

\begin{thm}[Existence of feasible weights] \label{thm:overlap} 
 Let Assumptions \ref{ass:seqignm}, \ref{ass:weakoverlap} hold. Consider $\delta_t(n,p_t) \ge c_{0,t} {n^{-1/2}}{\log^{3/2}(p_tn)}$ for a finite constant $c_{0,t} < \infty$, and $C_{n,t} \ge \frac{\bar{c}}{n \delta^t}$ for some sufficiently large constant $\bar{c} \in (0,\infty)$. 
Then with probability $\eta_n \rightarrow 1$, for each $t \in \{1, \cdots, T\}, T < \infty$, for some $N >0$, $n > N$, there exists a $\hat{\gamma}_t^*(\hat{\gamma}_{t-1})$ satisfying the constraint in Algorithm 1 as a function of the solution $\hat{\gamma}_{t-1}$ of Algorithm 1, where 
\begin{equation}\label{eqn:hh} 
\small 
\begin{aligned} 
\hat{\gamma}_{i,0}^* = 1/n, \quad \hat{\gamma}_{i,t}^*(\hat{\gamma}_{t-1}) = \hat{\gamma}_{i,t-1} w_{i,t}^* \Big/ \sum_{i=1}^n \hat{\gamma}_{i,t-1} w_{i,t}^*, \quad w_{i,t}^*:=  \frac{1\{D_{i,t} = d_t\}}{P(D_{i,t} = d_t | H_{i,t})} 
\end{aligned} 
\end{equation} 
 \end{thm}
 
 The proof is in Appendix \ref{app:overlap_proof}. In particular, the reader may refer to Appendix Lemma \ref{lem:firststat} for additional details on the construction of the feasible weights.  The algorithm thus finds weights that minimize the small sample variance, with the IPW weights $w_{i,t}^*$ (reweighted by the solution in the previous iteration $\hat{\gamma}_{i,t}$) being \textit{one} possible solution. 
 Minimizing the $l_2$ norms of the weights is a natural objective when the goal is to minimize the variance of the ATE estimator: following Theorem \ref{thm:thm_asym_t} below, under homoskedasticity of the residuals from each projection, the variance is proportional to a weighted sum of $||\hat{\gamma}_t||^2$. However, homoskedasticity is not necessary for our results.

\begin{cor} \label{cor:comparison} Let the conditions in Theorem \ref{thm:overlap} hold. Then  for some $N > 0$, $n > N$, with probability $\eta_n \rightarrow 1$,   there exist a sequence of feasible solutions $\Big\{\hat{\gamma}_t\Big\}_{t=1}^T$ such that 
 $
n ||\hat{\gamma}_{t}||^2 \le n||\hat{\gamma}_{t}^*(\hat{\gamma}_{t-1})||^2$ for each $t$, where  $\hat{\gamma}_{i,t}^*(\hat{\gamma}_{t-1})$ is as defined in Equation \eqref{eqn:hh}, and $n ||\hat{\gamma}_t||^2 \le n c_t ||\hat{\gamma}_{t-1}||^2$ for a finite constant $c_t < \infty$. 
\end{cor} 
Corollary \ref{cor:comparison} provides a desiderable stability property. It shows that the $l_2$ norm of the weights is upper bounded by the (stabilized) IPW weights reweighted by the solution in the previous step. From basic concentration inequalities, for $n$ sufficiently large, $\mathbb{E}[\hat{\gamma}_{i,t}^{*2}(\hat{\gamma}_{t-1})] \approx \mathbb{E}\Big[\frac{ \hat{\gamma}_{i,t-1}^2 }{P(D_{i,t} = d_t | H_{t})}\Big]$. In addition, the last part of the corollary shows that the weights' norm is controlled over each period by the norm in the previous period up to a finite multiplicative constant.

\begin{ass} \label{ass:4m} Let the following hold: for every $t \in \{1, \cdots, T\}$, $d_{1:T} \in \{0,1\}^T$, 
\begin{itemize} 
 \item[(i)]  $\max_t \| \hat{\beta}_{d_{1:T}}^{(t)} - \beta_{d_{1:T}}^{(t)} \| _1 \delta_t(n,p_t) = o_p(1/\sqrt{n})$, $\delta_t(n,p_t) \ge c_{0, t} {n^{-1/2}}{\log^{3/2}(p_tn)}$ for a finite constant $c_{0, t}$; In addition let either of the two following conditions hold: (a) $\max_t \| \hat{\beta}_{d_{1:T}}^{(t)} - \beta_{d_{1:T}}^{(t)} \|_1 = O_p(n^{-1/4})$; or (b) $\max_t \| \hat{\beta}_{d_{1:T}}^{(t)} - \beta_{d_{1:T}}^{(t)} \|_1 = o_p(1/\log(n))$ and $||H_t||_{\infty} \le \bar{H}$ almost surely for a finite constant $\bar{H} < \infty$ for all $t \ge 1$; 
 \item[(ii)]  Let $\nu_{i,t} = (H_{i,t +1} \beta_{d_{1:T}}^{(t+1)} - H_{i,t} \beta_{d_{1:T}}^{(t)}), \varepsilon_{i,T} = Y_{i,T} - H_{i,T} \beta_{d_{1:T}}^{(T)}$. For a finite constant $C$, $\mathbb{E}[\varepsilon_{i, T}^4 | H_{i, T}, D_{i, T}] < C$, $\mathbb{E}[\nu_{i,t}^4 | H_{i,t-1}, D_{i,t-1}] < C$, almost surely. In addition, $Y_{i,T}$ is a sub-gaussian random variable. 
 \item[(iii)]  $\mbox{Var}(\varepsilon_{i,T} | H_{i,T}, D_{i, T}), \mbox{Var}(H_{i,t} \beta_{d_{1:T}}^{(t)} - H_{i,t-1} \beta_{d_{1:T}}^{(t-1)} | H_{i,t-1}, D_{i,t-1}) > u_{min}$, almost surely, for some constant $u_{min} > 0$.  
 \end{itemize}  
\end{ass} 

Assumption \ref{ass:4m} imposes the consistency in estimating the outcome models. Condition (i) is attained for many high-dimensional estimators, such as the lasso method, under sparsity and restricted eigenvalues restrictions; see, e.g.,   \cite{buhlmann2011statistics}. An example and derivation for condition (i) for Lasso under sparsity is included in Example \ref{lem:suff}  (Appendix \ref{aa:suff}). As we discuss in Appendix Remark \ref{rem:restricted}, the restricted eigenvalue condition is imposed for all $H_t, t \ge 1$; references that study this condition from different angles include \cite{deshpande2023online}, Section C.4.

 \begin{thm}[Parametric convergence rate] \label{thm:convergence_rate_first}  Let the conditions in Theorem \ref{thm:overlap} and Assumption \ref{ass:4m} hold. 
Then, whenever $\log(n(\sum_t p_t))/n^{1/4} \to 0$ with $n,p_1, \cdots,p_T\to \infty$, it follows that 
 $
 \hat{\mu}_T(d_{1:T}) - \mu_T(d_{1:T}') = \mathcal{O}_P\Big(n^{-1/2}\Big).  
 $
 \end{thm}

Theorem \ref{thm:convergence_rate_first} guarantees a parametric convergence rate with high-dimensional covariates. 

\begin{thm}[Inference]  \label{thm:thm_asym_t} 
Let the conditions in Theorem \ref{thm:overlap} and Assumption \ref{ass:4m} hold. 
Then, whenever $\log(n \sum_t p_t)/n^{1/4} \to 0$, as $n, p_1, \cdots, p_T \rightarrow \infty$, 
\begin{equation}
\small  
\begin{aligned} 
{  \frac{\sqrt{n} \Big(\hat{\mu}(d_{1:T}) - \mu(d_{1:T})\Big)}{\hat{V}_T(d_{1:T})^{1/2}}} \rightarrow_d \mathcal{N}(0,1) 
\end{aligned} 
\end{equation} 
where 
 $$
\small  
 \begin{aligned} 
& \hat{V}_T(d_{1:T}) =  \\ 
& \sum_{i = 1}^n \left\{ n\hat{\gamma}_{i,T}^2(d_{1:T})  (Y_{i, T} - H_{i,T} \hat{\beta}_{d_{1:T}}^{(T)})^2 + \sum_{t=1}^{T-1} n  \hat{\gamma}_{i,t}^2(d_{1:t}) (H_{i,t+1} \hat{\beta}_{d_{1:T}}^{t+1} - H_{i,t} \hat{\beta}_{d_{1:T}}^{t})^2 +  \frac{1}{n}   (\bar{X}_1 \hat{\beta}_{d_{1:T}}^{(1)} - X_{i,1} \hat{\beta}_{d_{1:T}}^{(1)})^2\right\}
\end{aligned}.  
$$ 
\end{thm}

Inference on ATE follows as a direct corollary for two histories $d_{1:T}, d_{1:T}'$ with $d_1 \neq d_1'$ (see Theorem \ref{cor:ate}), as described in Appendix \ref{sec:main_theorem}. Also, for inference conditional on baseline covariates in the first period $X_1$, the relevant variance is $\hat{V}_T(d_{1:T}) - \frac{1}{n} \sum_i (\bar{X}_1 \hat{\beta}_{d_{1:T}}^{(1)} - X_{i,1} \hat{\beta}_{d_{1:T}}^{(1)})^2$ (since we condition on $X_1$) and for the corresponding ATE is the sum of these two variances.

  \begin{rem}[Strict overlap assumption] \label{rem:overlap} Strict overlap (Assumption \ref{ass:weakoverlap} $(i)$) is not necessary to achieve parametric convergence rates whenever a feasible solution to Algorithm 1 exists (see for instance \cite{athey2018approximate}, Lemma 2 in cross-sectional settings). \qed 
  \end{rem}

\begin{rem}[Rate conditions and comparison with AIPW] \label{rem:aipw} As noted in \cite{hirshberg2021augmented} in static settings, the advantages of balancing weights compared to AIPW  is to be able to estimate causal effects of interest under essentially the same conditions for AIPW on the conditional mean function, but weaker conditions on the balancing weights. Specifically, in high-dimensional settings, AIPW requires conditions of the form $||\hat{e} - e|| = o_P(n^{-1/4}), ||\hat{\beta} - \beta|| = o_P(n^{-1/4})$, where $e$ denotes the propensity score. Here, we only require $||\hat{\beta} - \beta||_1 = O_p(n^{-1/4})$ with sub-gaussian histories $H_t$ and $||\hat{\beta} - \beta||_1 = o_p(1/\log(n))$ with uniformly bounded covariates and no condition on the propensity score.  We could think of balancing weights as inheriting a ``product-of-rate" condition, where here the estimation error takes the form: 
$
||\hat{\beta} - \beta||_1 \delta(n,p). 
$ 
\end{rem} 

\begin{rem}[Propagation of error over multiple periods]  A natural question is how estimation error varies as the number of periods increase. To shed light on this question, note that the estimation bias for $\hat{\mu}(d_{1:T})$ is bounded above by 
$$
\sum_{t=1}^T ||\hat{\beta}_{d_{1:T}}^{(t)} - \beta_{d_{1:T}}^{(t)} ||_1 \Big| \Big| \sum_{i=1}^n \hat{\gamma}_{i,t-1} H_{i,t} - \sum_{i=1}^n \hat{\gamma}_{i,t} H_{i,t} \Big| \Big|_{\infty}.
$$ 
Therefore, multiple time periods may affect the error through  the estimation error of the coefficient, and through the weights $\hat{\gamma}_t$ in the balancing component. The properties of estimated coefficients for Lasso are presented in Appendix \ref{aa:suff}. For the balancing component instead, existence of a feasible solution requires that the balancing component grows at rate $\sqrt{\log(t)}$, so that effectively the constants are of order $K_{1,t} = \log^{1/2}(t)$ (see Appendix Lemma \ref{lem:firststat}). Intuitively, as we move along the path over multiple time periods it becomes harder to guarantee approximate balance, reflecting into weaker constraints on the balancing set.  
\end{rem}

\section{Guide to practice: numerical studies and application} \label{sec:app}

\subsection{Implementation guide} 

 The complete Algorithm 1 is implemented off-the-shelf in the R-package {\tt DynBalancing}. 

It requires researchers to specify four main parameters: the length $h$ of the treatment history considered (i.e., carry-over effects), two treatment histories of length $h$, $d_{(T-h):T}, d_{(T-h):T}'$ to compare, the model used to estimate the coefficients ({\tt linear} or {\tt fully interacted}) as described in Algorithm 2, and whether to consider a {\tt pooled} regression.

\textit{Choosing the length of the treatment history with long panel} With short panels, selecting the length of the treatment history $h = T$ is natural. With long panels, this may reduce the effective sample size or be infeasible (as the effective sample becomes ``thinner"). This is because, as for IPW, the weights at time $t$ can be non zero only for those units observed over a given treatment path up to time $t$. Therefore, we recommend selecting a treatment history $h$ shorter than the number of periods $T$ (i.e., $h < T$), and estimate causal effects of the form
\begin{equation} \label{eqn:average_effect} 
\small 
\begin{aligned} 
\mathbb{E}\left[Y_{i,T}\Big(D_{1:(T - h)}, d_{T- h + 1}, \cdots, d_T\Big)\right] - 
\mathbb{E}\left[Y_{i,T}\Big(D_{1:(T - h)}, d_{T- h + 1}', \cdots, d_T'\Big)\right]
\end{aligned} 
\end{equation}  
for given treatment histories $d_{(T-h):T}, d_{(T-h):T}'$. Equation \eqref{eqn:average_effect} estimates the effect of exposing an individual to two different histories over the last $h$ periods and average over previous assignments. Our analysis and estimation follow similarly to Algorithm 1, with the difference that we construct balancing weights starting from period $T - h$ and proceed sequentially until time $T$ (observable characteristics before time $T - h$ can be used as additional controls). As in \cite{imai2018matching}, the focus on Equation \eqref{eqn:average_effect} makes our procedure robust to long panels. 

As a rule of thumb, as we illustrate in our application, we recommend report results for different choices of $h$ (say $h \in \{1, \cdots, 10\}$ in a long panel); our package reports and plots the estimated effects along-side standard errors which can help disentangle the trade-offs between identification of long-run effects against precision. In addition, it is useful to report $1/(n ||\hat{\gamma}_t||^2)$ as a measure of effective sample size at time $t$ for different values of $t$, which can help guide the choice of $h$ (larger choices of $1/(n ||\hat{\gamma}_t||^2)$ indicates more accurate treatment effects estimates).  \\ 
\textit{Choice of the model specification ({\tt linear} or {\tt fully interacted})} The estimation error $||\hat{\beta}_{d_{1:T}}^{(t)} - \beta_{d_{1:T}}^{(t)}||_1$ depends on modeling assumptions. For the {\tt fully interacted} model, $||\hat{\beta}_{d_{1:T}}^{(t)} - \beta_{d_{1:T}}^{(t)}||_1$ scales exponentially with $T$ as it considers all possible interactions with the treatment assignments $d_1, \cdots, d_T$. The {\tt linear} model avoids that the effective sample size shrinks exponentially in $T$ but imposes homogeneity restrictions of treatment effects as for example in \cite{acemoglu2019democracy}, by modeling treatment effects as additive and linear. See Algorithm 2 for more details. In addition, when {\tt pooled} is true, we consider a regression
$$
\begin{aligned} 
Y_{i,t}(d_{1:t}) =  \beta_0 + \beta_1 d_t + \beta_2 Y_{i,t-1}(d_{1:(t-1)}) + X_{i,t}(d_{1:(t-1)}) \gamma +  \tau_t + \varepsilon_{i,t},
\end{aligned}   
$$  
where $\tau_t$ denotes fixed effects, pooling together effects estimated in different periods.  We then cluster standard errors at the individual level to allow for correlation over time.

\begin{figure}[!ht]
\centering
\includegraphics[scale=0.8, page = 2, 
trim={1cm 15cm 1cm 2.5cm},clip]{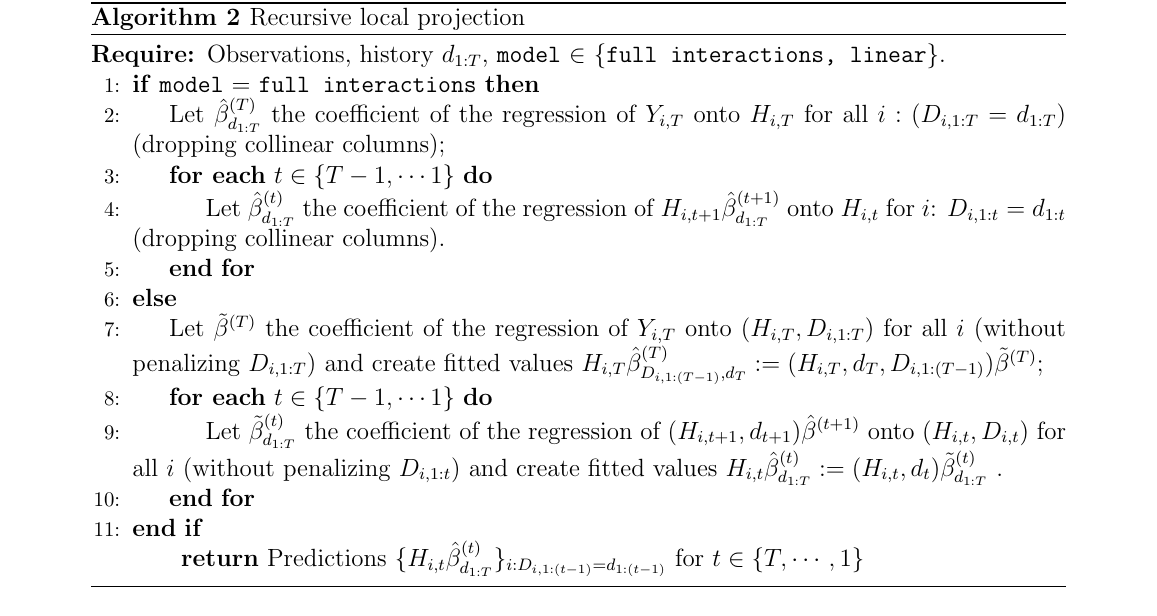} 
\end{figure}

\begin{rem}[Tuning parameters] \label{rem:tuning} Similarly to one-dimensional setting \citep{athey2018approximate}, Algorithm 1 requires choosing tuning parameters. A complete description is in Algorithm \ref{alg:algtuning} and uses a data-adaptive procedure (i.e., researchers do not need to specify the tuning parameters). In a nutshell, we choose $\delta_t(n,p) = \log^{3/2}(p_t n)/n^{1/2}$ (here $p_t$ is the dimension of covariates at time $t$) as prescribed by the theoretical analysis in Section \ref{sec:theory} and for simplicity we set $C_{n,t} = \log(n) n^{-2/3}$. To guarantee balance with many covariates, we iterate over a grid of values for $K_1$, and select the smallest constant within this grid such that the optimization program admits a feasible solution (in Algorithm 1 we also refine the algorithm to weight more the constraints for which the coefficients are non-zero). This approach minimizes the estimator's bias and, within the set of weighting estimators with the smallest bias, selects the one with the smallest variance. Below and in Appendix \ref{sec:num} we illustrate the benefits of this procedure through numerical studies. \qed 
 \end{rem}
 
 \begin{rem}[Computational complexity] Algorithm 1 is a sequence of $T$ quadratic programs with linear constraints. Its complexity scales only polynomially with $n, p$. Appendix Figure \ref{fig:long_T2} shows that the computational time is between a few seconds and a few minutes for $T \in \{1, \cdots, 10\}$ on a personal laptop (including choosing the tuning parameters).  \qed 
 \end{rem} 
 
 \subsection{Numerical studies} \label{sec:numerics}

Next, we collect results from numerical experiments.
 We estimate 
$
\mathbb{E}\Big[Y_{i,T}(1, \cdots, 1) - Y_{i,T}(0, \cdots,0)\Big], T \in \{2,3\}.  
$
We let the baseline covariates $X_{i,1}$ be drawn from as i.i.d. $\mathcal{N}(0, \Sigma)$ with $\Sigma^{(i,j)} = 0.5^{|i-j|}$.  
Covariates in the subsequent period are generated according to an auto-regressive model
$
\{X_{i,t}\}_j =0.5 \{X_{i,t-1}\}_j + \mathcal{N}(0,  1),  j=1,\cdots,p_t.
$
Treatments are drawn from a logistic model that depends on all previous treatments and past covariates:
$
D_{i,t} \sim \mbox{Bern}\Big((1 + e^{\iota_{i,t}})^{-1}\Big)
$
with
\begin{equation}  \label{eqn:propmodel}
\small 
\begin{aligned} 
\iota_{i,t} =   \eta  \sum_{s=1}^t X_{i,s} \phi+ \sum_{s=1}^{t-1} \delta_s (D_{i,s} - \bar{D}_s) +  \xi_{i,t}, \quad \bar{D}_s =n^{-1} \sum_{i=1}^n D_{i,s}
\end{aligned} 
\end{equation}
 and $\xi_{i,t} \sim \mathcal{N}(0,1)$, for $t \in \{1, 2,3\}$.
 Here, $\eta, \delta$ controls the association between covariates and treatment assignments. 
We consider values of $\eta \in \{0.1, 0.3, 0.5\}$, $\delta_1  = 0.5, \delta_2 = 0.25$.  
We let $\phi \propto 1/j$, with $\|\phi \|_2^2 = 1$, similarly to balancing conditions presented in \cite{athey2018approximate}.  
The larger $\eta$ corresponds to weaker overlap (see Table~\ref{tab:summaries} in the Appendix).


We generate the outcome as
$
 Y_{i,t}(d_{1:t}) = \sum_{s = 1}^t \Big(X_{i,s} \beta + \lambda_{s, t} Y_{i,s-1} + \tau d_s\Big) +  \varepsilon_{i,t}(d_{1:t}),  \quad t=1,2,3,
$
where  elements of $\varepsilon_{i,t}(d_{1:t})$ are i.i.d. $ \mathcal{N}(0,1)$ and $\lambda_{1,2} = 1, \lambda_{1,3}, \lambda_{2,3} = 0.5$. We consider three different settings: \textit{Sparse} with  $\beta^{(j)} \propto 1\{j \le 10\}$, \textit{Moderate}  with  moderately sparse    $\beta^{(j)} \propto 1/j^2$ and the \textit{Harmonic} setting with $\beta^{(j)} \propto 1/j$. We set $\| \beta \|_2 =1, \tau = 1$. 

We consider the following competing methodologies: 

\begin{itemize} 
\item \textit{Augmented IPW}, with \textit{known} propensity score and with \textit{estimated} propensity score. The method replaces the balancing weights in Equation \eqref{eqn:myestimator} with the (estimated or known) propensity score. Estimation of the propensity score is performed using a logistic regression (denoted as aIPWl) and a penalized logistic regression (denoted as aIPWh) \citep[e.g.][for a discussion]{nie2021learning, bodoryevaluating}. For both AIPW and IPW we consider stabilized inverse probability weights. 
\item   \textit{CAEW (MSM)}: Although our  balancing weights in Algorithm 1 are novel (with or without regression adjustment), we can compare to other balancing procedures. In particular, we consider Marginal Structural Model (MSM) with balancing weights computed using the method in   \cite{yiu2018covariate, yiu2020joint} that, different from ours, require information about the propensity score. We follow Section 3 in \cite{yiu2020joint} for its implementation. (We do not also compare to \citealt{imai2015robust} for MSM since it is not feasible in high-dimensions.)

\item
\textit{``Dynamic" Double Lasso}: it estimates the effect of each treatment assignment separately, after conditioning on the present covariate and past history for each period using the double lasso discussed in one period from \cite{belloni2014inference}.

 \item \textit{Naive Lasso}: it runs a regression controlling for covariates and treatment assignments. 
 
\item \textit{Sequential Estimation}: it estimates the conditional mean in each time period sequentially using the lasso method, and it predicts end-line potential outcomes as a function of the estimated potential outcomes in previous periods.

\item \textit{DiD switchback}: it is a DiD estimator similar to \cite{de2024difference}. 

\item \textit{Simple LP} (Local Projection): it projects $Y_T$ onto baseline covariates $X_{i,1}$ and treatment $D_{i,1}$ and take the coefficient multiplying $D_{i,1}$ as the estimated effect, while penalizing the coefficients for $X_{i,1}$ via Lasso.  
\end{itemize} 
For Dynamic Covariate Balancing, \textit{DCB}, the choice of tuning parameters is data adaptive, and it uses a grid-search method discussed in Appendix \ref{app:algorithms} and Remark \ref{rem:tuning}.  We estimate coefficients as in Algorithm 1 for DCB and (a)IPW, with a linear model in treatment assignments. Estimation of the penalty for the lasso methods is performed via cross-validation.

We consider $\mathrm{dim}(\beta) = \mathrm{dim}(\phi) = 100$ and set the sample size to be  $n = 400$. We set $p_1=101, p_2=203, p_3=305$ as number of covariates in each period.

In Table \ref{tab:mse2} we collect results for the average mean squared error for estimating the average treatment effect in two and three periods. 
 Throughout all simulations, the proposed method significantly outperforms any other competitor, with one single exception for $T = 2$, good overlap and harmonic design. It also outperforms using known propensity score, consistently with our findings in Theorem \ref{thm:overlap}, where we show that the propensity score is a feasible solution of DCB weights (and in the absence of knowledge of the propensity score).  

Finally, in Appendix \ref{sec:num} we consider more extensive simulation studies with a longer time horizon, a misspecified (non-linear) model, low and high dimensional settings among additional simulation designs.

  \definecolor{glaucous}{rgb}{0.38,0.51,0.71}

\begin{table*}[!ht]\centering
\caption{Mean Squared Error (MSE) for estimating the average treatment effect of always vs never being under treatment of Dynamic Covariate Balancing (DCB) across $200$ repetitions with sample size $400$ and $101$  variables in time period 1. This implies that the number of variables in time period $2$ and $3$ are $203$ and $304$. Oracle Estimator is denoted with aIPW$^*$ whereas aIPWh(l) denote AIPW with high(low)-dimensional estimated propensity. CAEW (MSM) corresponds to the method in \cite{yiu2020joint}, D.Lasso is adaptation of  Double Lasso \citep{belloni2014inference}.}    \label{tab:mse2} 
\scalebox{0.7}{\begin{tabular}{@{}lrrrcrrrcrrr@{}}\toprule
& \multicolumn{3}{c}{$\eta=0.1$} & \ & \multicolumn{3}{c}{$\eta=0.3$} & \ & \multicolumn{3}{c}{$\eta=0.5$} \\
\cmidrule{2-4} \cmidrule{6-8}  \cmidrule{10-12}  
&\small sparse&\small mod &\small harm && \small sparse&\small  mod  & \small harm && \small sparse& \small mod  & \small harm  \\\Xhline{.8pt} 
 \rowcolor{lightgray}\multicolumn{12}{c}{$T=2$} \\ \Xhline{.8pt} 
 \small aIPW$^*$  & $0.069$ & $0.092$ & $0.071$ &&$0.102$ & $0.104$ & $0.118$ &&0.131&$0.127$ & $0.132$\\
 \cline{2-4}  \cline{6-8}  \cline{10-12} 
 \rowcolor{glaucous!60}  \small DCB & $0.060$ & ${0.077}$ & $0.075$&& ${0.092}$ & ${0.076}$ & ${0.084}$&&0.099&0.077&0.085\\
\small aIPWh & ${0.064}$ & $0.091$ & ${0.070}$&&$0.180$ & $0.204$ & $0.218$ &&$0.265$ & $0.312$ & $0.368$\\
\small aIPWl & $0.260$ & $0.229$ & $0.212$ && $0.157$ & $0.201$ & $0.165$&& $0.214$ & $0.234$ & $0.213$  \\
\small IPWh & $2.37$ & $1.78$ & $2.80$&&$10.19$ & $6.49$ & $11.72$ &&$15.25$ & $8.09$ & $16.67$\\
\small Seq.Est. & $0.932$ & $1.333$ & $0.692$ && $1.388$ & $1.787$ & $1.152$&& $1.759$ & $1.795$ & $1.664$  \\
\small  Lasso & $0.247$ & $0.410$ & $0.132$ && $0.509$ & $0.710$ & $0.298$ && $0.762$ & $0.948$ & $0.560$ \\ 
\small CAEW & $0.432$ & $0.444$ & $0.517$ &&$1.934$ & $1.274$ & $1.974$&& $3.376$ & $2.168$ & $4.423$  \\
\small Dyn.D.Lasso & $0.124$ & $0.118$ & $0.256$ && $0.208$ & $0.147$ & $0.430$&& $0.218$ & $0.153$ & $0.554$ \\
\small DiD Switchback & $2.06$ & $1.71$ & $1.60$  && $14.52$ & $6.98$ & $20.72$ && $38.79$ & $6.98$ & $52.48$ \\
\small Simple LP & $1.28$ & $1.40$ & $1.37$  && $1.613$ & $1.682$ & $1.573$ && $1.777$ & $1.689$ & $1.808$ \\
\Xhline{.8pt}  \rowcolor{lightgray}\multicolumn{12}{c}{$T=3$} \\ \Xhline{.8pt} 
  \small aIPW$^*$  & $0.226$ & $0.296$ & $0.261$&&  $0.403$ & $0.251$ & $0.339$&& $0.472$ & $0.496$ & $0.562$\\
 \cline{2-4}  \cline{6-8}  \cline{10-12} 
  \rowcolor{glaucous!60}  \small DCB & ${0.155}$ & ${0.208}$ & ${0.199}$  &&  ${0.257}$ & ${0.217}$ & ${0.329}$  && $0.294$ & $0.267$ & $ 0.455$\\
\small aIPWh  & $0.201$ & $0.273$ & $0.280$&& $0.595$ & $0.747$ & $0.835$&& $0.999$ & $1.328$ & $1.607$  \\
\small aIPWl & $0.823$ & $0.625$ & $0.829$ &&$0.623$ & $0.704$ & $0.638$&& $1.078$ & $1.396$ & $1.234$ \\
\small IPWh &  $11.03$ & $8.09$ & $12.84$&&$34.65$ & $20.34$ & $39.37$ &&$47.65$ & $23.30$ & $45.47$\\
\small Seq.Est. & $2.608$ & $4.016$ & $2.316$  && $3.722$ & $5.269$ & $3.818$&& $5.279$ & $6.829$ & $5.467$  \\
\small   Lasso & $0.409$ & $0.492$ & $0.514$ && $0.559$ & $0.732$ & $0.507$ && $1.290$ & $1.315$ & $1.174$ \\
\small CAEW  & $3.580$ & $2.446$ & $4.279$  && $18.50$ & $12.07$ & $22.85$ && $30.07$ & $18.71$ & $33.01$  \\
\small Dyn.D.Lasso & $0.471$ & $0.344$ & $0.679$  && $0.694$ & $0.378$ & $1.182$ && $0.964$ & $0.383$ & $1.594$ \\
\small DiD Switchback & $24.9$ & $27.98$ & $20.52$  && $21.07$ & $7.503$ & $38.33$ && $59.23$ & $22.15$ & $89.07$ \\
\small Simple LP & $7.38$ & $7.54$ & $7.38$  && $8.180$ & $8.154$ & $8.294$ && $9.131$ & $9.087$ & $9.217$ \\
\bottomrule
\end{tabular} }
\end{table*}

 \subsection{Empirical illustration} 
 
In this section, we present an empirical application for studying the effect of democracy on GDP growth using data from \cite{acemoglu2019democracy}. \cite{acemoglu2019democracy} studied dynamic treatment effects of democracy under GDP growth under sequential ignorability  \citep[Assumption 1 in][]{acemoglu2019democracy}.  Figure \ref{fig:treatment_path2} illustrates the dynamics of treatments. Many units switch treatment over time, violating standard event studies designs.

 The data (available at \url{https://www.journals.uchicago.edu/doi/suppl/10.1086/700936}) consist of a collection of countries observed between $1960$ and $2010$. We consider observations starting from $1989$. After removing missing values, we run regressions with 141 countries. The outcome is the log-GDP in the country $i$ in period $t$ as in \cite{acemoglu2019democracy}. 
We use the same treatment specification as in \cite{acemoglu2019democracy}, which is binary. We study the effect of exposing countries at time $t$ to democracy for in $s$ years before (and including) $t$ versus not exposing them to democracy for the previous $s$ years. Namely, the estimand is the $s$-long run effect of democracy, after averaging over past assignments.%
  We let $s \in \{1, \cdots, 20\}$ to study the impact from one to twenty years of democracy.

For each country, we condition on lag outcomes in the past four years as in the preferred specification of \cite{acemoglu2019democracy}, and past four treatments. We consider a pooled regression and two alternative specifications. The first is parsimonious and includes dummies for different regions (continents) and different intercepts for different periods. The second one controls for the past four outcomes, past four treatments, for the geographical region, and colonial history as in  \cite{acemoglu2019democracy}. 
Coefficients are estimated as in Algorithm 2 with {\tt model} $=$ {\tt linear}.

\definecolor{aero}{rgb}{0.49,0.73,0.91}
\definecolor{airsuperiorityblue}{rgb}{0.45,0.63,0.76}
\definecolor{babyblueeyes}{rgb}{0.63,0.79,0.95}
\definecolor{beaublue}{rgb}{0.74,0.83,0.9}
\definecolor{glaucous}{rgb}{0.38,0.51,0.71}

\begin{figure}[!ht]
\centering 
\includegraphics[scale=0.7]{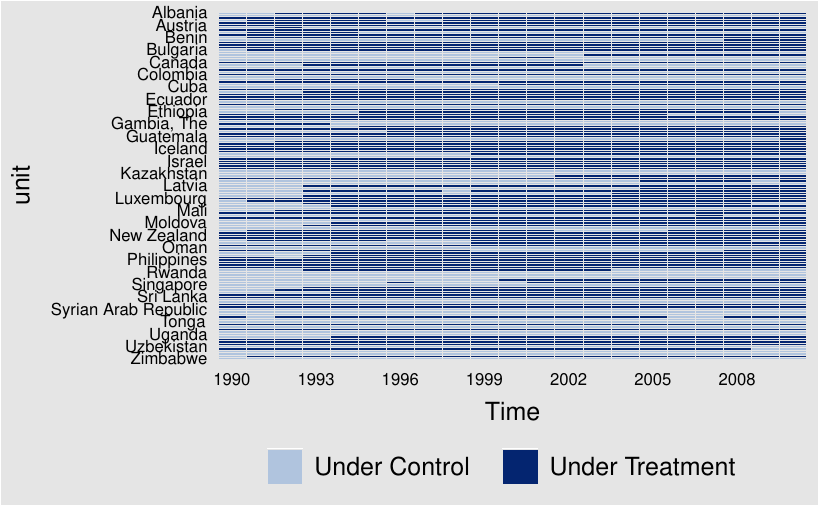}

\caption{The figure illustrates the dynamics of treatments. } 
\label{fig:treatment_path2} 
\end{figure}

\begin{figure}[!ht] 
\centering 
\includegraphics[scale=0.35]{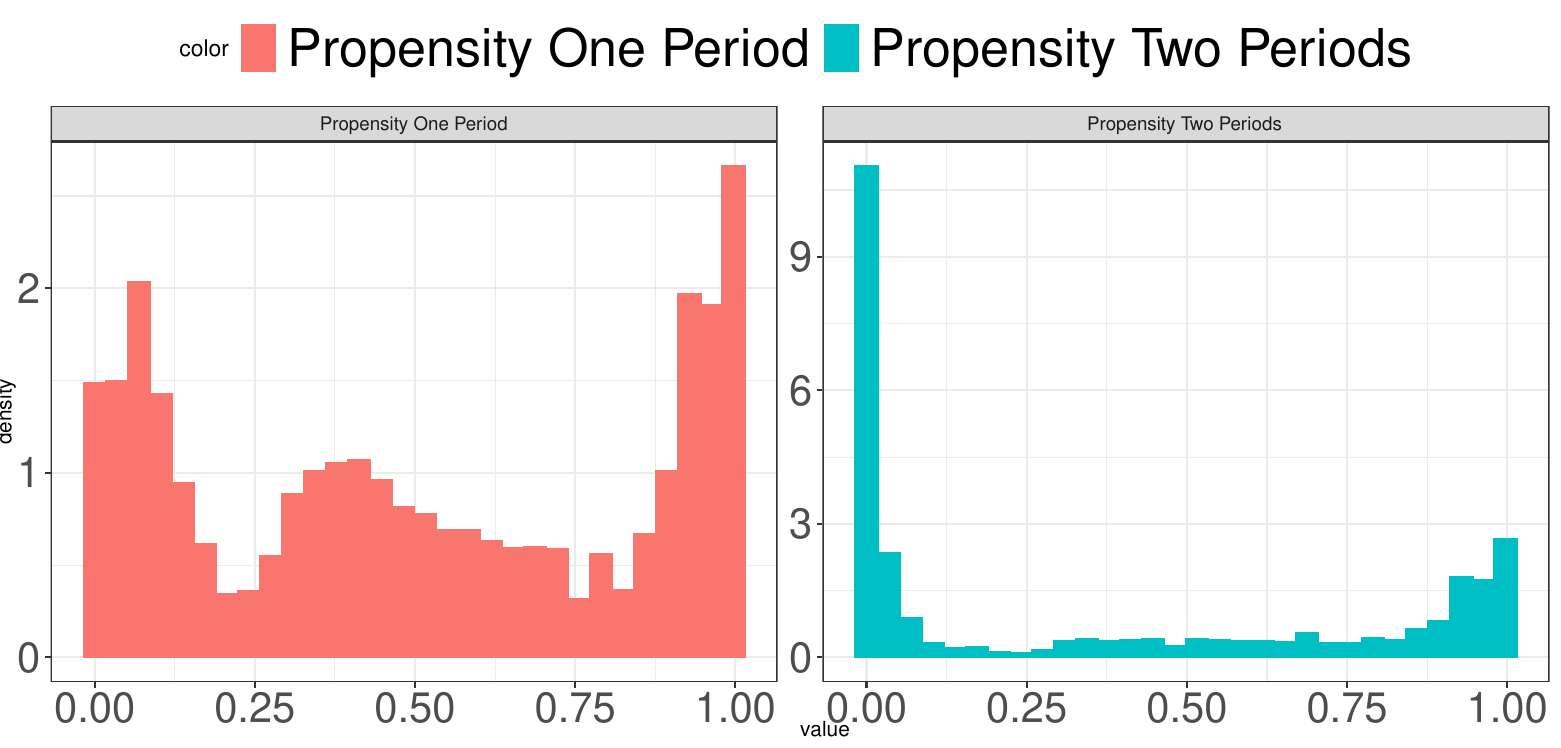} 
\caption{Estimated probability of treatment for one year (left-panel) and two consecutive years (right-panel). Estimation is performed via logistic regression with a pooled regression with year, region fixed effects and four lagged outcomes. The right panel also controls for the past treatment assignment. The figure illustrates the sensitivity of inverse probability weights to longer time horizon, motivating more stable balancing weights.} \label{fig:overlap}
\end{figure}

\begin{table}[!htbp]
\centering
\caption{Diagnostics across horizons with time horizon $h= 3$: inverse of weights dispersion for DCB weights and $1/||w_t||^2$ for IPW where $w_t$ is the stabilized inverse probability weight. Estimation of IPW is performed via logistic regression with a pooled regression with year, region fixed effects and four lagged outcomes. Larger number indicates better precision.}
\label{tab:diagnostics_h}
\begin{tabular}{lccc|ccc}
\toprule
$h = 3$ & \multicolumn{3}{c|}{$\mathbb{E}[Y_T(\mathbf{1}_3, D_{T-3:1})]$} & \multicolumn{3}{c}{$\mathbb{E}[Y_T(\mathbf{0}_3, D_{(T-3):1}))]$} \\
\cmidrule(lr){2-4}\cmidrule(lr){5-7}
& $1/||\hat{\gamma}_1||^2$ & $1/||\hat{\gamma}_2||^2$ & $1/||\hat{\gamma}_3||^2$ & $1/||w_1||^2$ & $1/||w_2||^2$ & $1/||w_3||^2$ \\
\midrule
DCB & 67 & 64 & 62 & 34 & 34 & 34 \\
IPW & 73 & 17 & 10  & 32 & 23 & 23 \\
\bottomrule
\end{tabular}
\end{table}

Figure \ref{fig:acemoglu} collects our results. Democracy has a statistically insignificant effect over the first few years and a statistically significant positive impact on long-run GDP growth after three years. Point estimates are in sign and magnitude consistent with what found by \cite{acemoglu2019democracy}, and results are robust across the two specifications for DCB. 
 
 We compare our method to (i) the \textit{linear estimator} reported by \cite{acemoglu2019democracy} (Table 2, Column 3), where dynamic effects are estimated by propagating the effect over past outcomes at each period (we consider two specifications, with and without unit fixed effects -- both report similar results); (ii) the \textit{simple local projection}, that projects the outcome on the treatment and the past outcome $s$ periods before, with and without country fixed effects, time fixed effects and controlling for lagged outcome at time $s$.

 The simple local projection approach reports small point estimates compared to other methods. This result is consistent with our theoretical discussion: local projections average over the distribution of \textit{future} assignments. Therefore, the causal effects estimated by the local projection differ from the target long-run effect, which instead \textit{fixes} future treatment assignments. 
 The effect estimated as in \cite{acemoglu2019democracy} is larger than the local projection when including country fixed effects, but significantly smaller than the effect estimated through DCB. Therefore, the specification in \cite{acemoglu2019democracy} may capture some but not all the long-run effects. After controlling for imbalance with DCB, average treatment effects are twice as large. The results from \cite{acemoglu2019democracy} with and without unit fixed effects report almost identical results. 

To investigate differences with (A)IPW methods, the right panel in Figure \ref{fig:acemoglu} presents comparisons in terms of the imbalance over the lagged outcome at time $t - 1$ when using balancing or inverse probability weights. As \cite{acemoglu2019democracy} note, the lags outcome may capture most variation in treatment. Therefore, an imbalance in lagged GDP may suggest the presence of bias. We report the relative improvement in absolute imbalance (average across the potential outcomes under treatment and control) and observe substantial gain over using inverse probability weights. Such gains illustrate the advantage of balancing in small sample.

Figure \ref{fig:overlap} complements Figure \ref{fig:acemoglu} showing instability of inverse probability weights, and Figure \ref{fig:dispersion} in the Appendix show that DCB weights present less dispersion than IPW weights. As a helpful diagnostic, in Table \ref{tab:diagnostics_h} we report $1/(||\hat{\gamma}_t||^2)$ for the DCB method as well as for IPW weights, where $\hat{\gamma}_t$ is replaced by the AIPW weight. This measure is indicative of the level of precision and effective sample size (a larger number indicates better precision). We find substantial improvements of DCB over IPW especially for weights estimated for longer time horizons.

 \begin{figure}[!ht]
 \centering
 \includegraphics[scale=0.35]{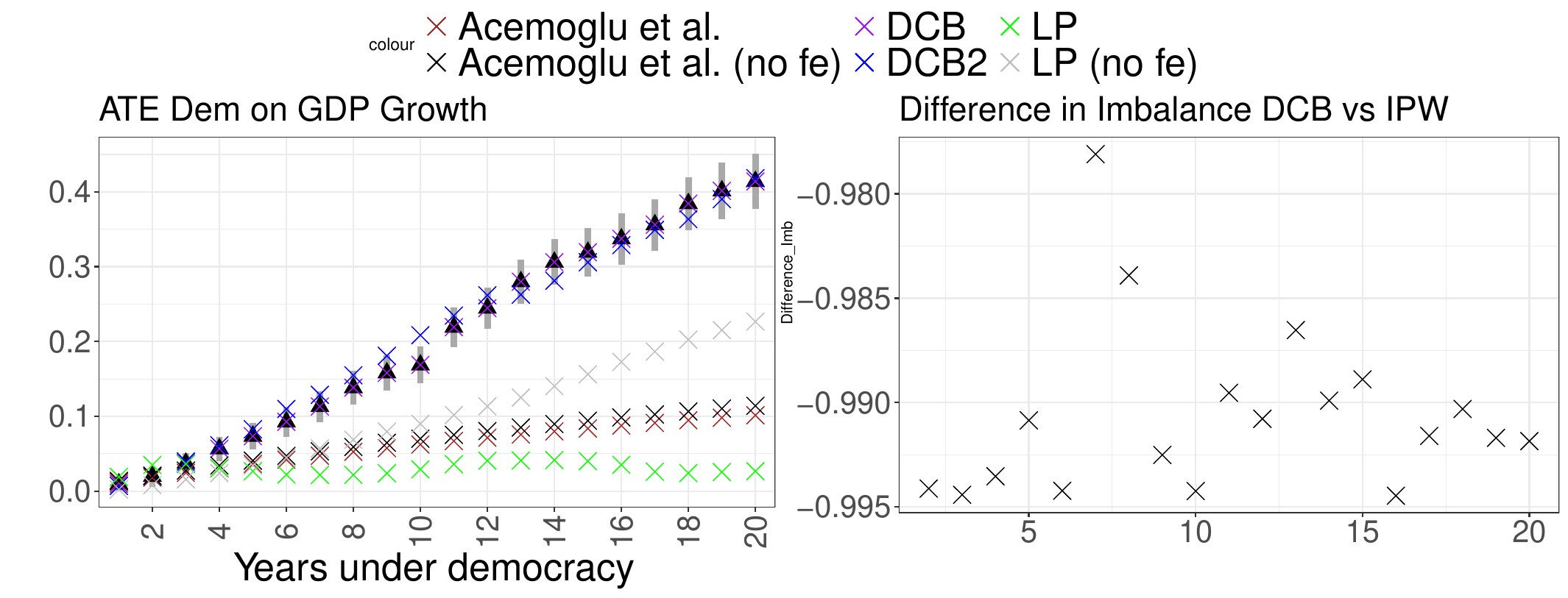}
 \caption{Left-hand side: pooled regression from $t \in \{1989, \cdots, 2010\}$. Gray region denotes the $90\%$ confidence band for the least parsimonious model. DCB and DCB2 refer to two separate specification, with DCB corresponding to the more parsimonious one. LP denotes a local projection on $t$ periods before. (no fe) indicates a specification without country fixed effects. Right-hand side reports $\log(|I(dcb)| + 1)/\log(|I(aipw)| + 1) - 1 \approx |I(dcb)|/|I(aipw)| - 1$, where $I(\cdot)$ denotes the imbalance (as in Lemma \ref{lem:balancing1}) in the lagged outcome usin DCB or IPW weights.  } \label{fig:acemoglu} 
 \end{figure}
 
 \appendix 
 
 \newpage 
\begin{center}
\vspace{-1.8cm}{\Large \textbf{Appendix}} 
\bigskip
\end{center}

\renewcommand\thefigure{\thesection.\arabic{figure}} 
\setcounter{figure}{0}   
\renewcommand\thetable{\thesection.\arabic{table}} 
\setcounter{table}{0} 
\setcounter{page}{0}
\startcontents[appendix]
\printcontents[appendix]{}{1}{\setcounter{tocdepth}{3}}
\numberwithin{equation}{section}
\numberwithin{table}{section}
\numberwithin{figure}{section}
\newpage

\section{Extended discussion for dynamic balancing}  \label{sec:extended_discussion}

In this section we provide an extended discussion of the method. 

\subsection{Estimation of the coefficients} 



As a first step, we estimate the regression coefficients with Algorithm \ref{alg:coefficients1} for $T = 2$ (and recursively as in Algorithm 2 in the main text for general $T$). Recall, that from Lemma \ref{lem:identification_model1}, we have 
$$
\small 
\begin{aligned} 
& \mathbb{E}\Big[Y_{i,2} \Big| H_{i,2}, D_{i,2} = d_2, D_{i,1} = d_1\Big] = \mathbb{E}\Big[Y_{i,2}(d_1, d_2) \Big| H_{i,2}, D_{i,1} = d_1\Big] = 
 H_{i,2}(d_1) \beta_{d_1, d_2}^{(2)}  \\ 
& \mathbb{E}\Big[\mathbb{E}\Big[Y_{i,2} \Big| H_{i,2}, D_{i,2} = d_2, D_{i,1} = d_1\Big] \Big| X_{i,1}, D_{i,1} = d_1\Big]  =
 \mathbb{E}\Big[Y_{i,2}(d_1, d_2) \Big| X_{i,1} \Big] = X_{i,1} \beta_{d_1, d_2}^{(1)}. 
\end{aligned} 
$$ 
Intuitively, identification of the linear model requires recursive projections of each conditional expectation (controlling for the treatment sequence of interest) onto past information. Therefore, 
 the algorithm recursively projects the estimated conditional mean functions on past histories. The penalization is by using lasso for estimation \citep{hastie2015statistical}. The algorithm considers two separate model specifications. The first allows for arbitrary heterogeneity in observable characteristics. This specification is cumbersome for longer time horizons because the effective sample size shrinks exponentially with the number of periods. The second specification assumes separable treatment effects, and it is more parsimonious. It is possible to model heterogeneity in treatment effects in the second specification by including interaction terms between observable characteristics and treatments. We require that the parameters of the estimated model converge to the true parameters of the linear model for each regression at a rate of order $O(n^{-1/4})$. This condition is typically attained for lasso and discussed in detail in  Section \ref{sec:theory} (Assumption \ref{ass:4m} (i) and discussion therein). Therefore, given the convergence rate requirement, a more parsimonious model such as the linear model in Algorithm \ref{alg:coefficients1} imposes stronger restrictions on the estimation error of the parameters.

\begin{algorithm} [h]   \caption{Recursive local projection for $t = 2$}\label{alg:coefficients1}
    \begin{algorithmic}[1]
    \Require Observations, history $d_{1:2} = (d_1, d_2)$, {\tt model} $\in$ {\tt \{full interactions, linear\}}. 
    \If{{\tt model} $=$ {\tt full interactions}}
    \State Let $\hat{\beta}_{d_{1:2}}^{(2)} $ the coefficient of the regression of $Y_{i,2}$ onto $H_{i,2}$ for all $i: (D_{i, 1:2} = d_{1:2})$ (dropping collinear columns); 
     \State Let $\hat{\beta}_{d_{1:2}}^{(1)} $ the coefficient of the regression of $H_{i,2} \hat{\beta}_{d_{1:2}}^{(2)}$ onto $X_{i,1}$ for $i$  that has  $D_{i,1} = d_{1}$ (dropping collinear columns). 
    \Else 
    \State Let $\tilde{\beta}^{(2)} $ the coefficient of the regression of $Y_{i,2}$ onto $(H_{i,2}, D_{i,2}, D_{i,1})$ for all $i$ (without penalizing $(D_{i,1}, D_{i,2})$) and create fitted values $H_{i,2} \hat{\beta}_{D_{i,1}, d_2}^{(2)} := (H_{i,2}, d_2, D_{i,1}) \tilde{\beta}^{(2)}$; 
     \State Let $\tilde{\beta}_{d_2}^{(1)} $ the coefficient of the regression of $(H_{i,2}, d_2, D_{i,1}) \hat{\beta}^{(2)}$ onto $(X_{i,1}, D_{i,1})$ for all  $i$ (without penalizing $D_{i,1}$) and create fitted values $X_{i,1} \hat{\beta}_{d_1, d_2}^{(1)} := (X_{i,1}, d_1)\tilde{\beta}_{d_{2}}^{(1)}$ for all $i$. 
    \EndIf  
    \State  
    \Return Predictions $\{X_{i,1} \hat{\beta}_{d_1, d_2}^{(1)}\}_{i=1}^n, \{H_{i,2} \hat{\beta}_{d_1, d_2}^{(2)}\}_{i: D_{i,1} = d_1}$
         \end{algorithmic}
\end{algorithm}


\begin{rem}[Estimation error of the coefficients in $T$ periods] \label{rem:ee} The estimation error $||\hat{\beta}_{d_{1:T}}^{(t)} - \beta_{d_{1:T}}^{(t)}||_1$ in $T$ periods depends on modeling assumptions. For the fully interacted model, $||\hat{\beta}_{d_{1:T}}^{(t)} - \beta_{d_{1:T}}^{(t)}||_1$ scales exponentially with $T$ since this model requires running regressions over the subsample with treatment histories $D_{1:t} = d_{1:t}$. Additional assumptions permit to estimate $\hat{\beta}_{d_{1:T}}^{(t)}$ using more information from the sample. A simple example is to explicitly model the effect of the treatment history $d_{1:T}$ on the outcome as in the linear model, or to assume that carry-over effects of treatments on future outcomes at time $t$ occur over a finite number of $h < t$ periods, as, for example, in \cite{imai2018matching}. \qed 
\end{rem}

\subsection{Dynamic estimator: illustration with IPW weights }

 Because linearity may only hold as we control for high-dimensional covariates, we cannot directly use our predictions $H_2 \hat{\beta}^{(2)}$ or $X_1 \hat{\beta}^{(1)}$ for valid causal inference. We instead must guarantee a vanishing high-dimensional bias through reweighting. 
 
Given the estimated coefficients $\hat{\beta}^{(1)}, \hat{\beta}^{(2)}$, and following previous literature on doubly-robust scores \citep{tchetgen2012semiparametric, zhang2013robust, jiang2015doubly, nie2021learning}, we propose an estimator that exploits linearity while reweighting observations to guarantee balance. Formally, we consider an estimator 
\begin{equation} \label{eqn:myestimatorb}
\small 
\begin{aligned} 
\hat{\mu}_2(d_1, d_2; \hat{\gamma}_1, \hat{\gamma}_2) &= \sum_{i=1}^n  \Big\{ \hat{\gamma}_{i,2}(d_1, d_2) Y_{i,2} -  \Big(\hat{\gamma}_{i,2}(d_1, d_2) - \hat{\gamma}_{i,1}(d_1, d_2)\Big)  H_{i,2} \hat{\beta}_{d_1, d_2}^{(2)}\Big\} \\& -  \sum_{i=1}^n \Big(\hat{\gamma}_{i,1}(d_1, d_2) - \frac{1}{n} \Big)  X_{i,1} \hat{\beta}_{d_1, d_2}^{(1)}. 
\end{aligned} 
\end{equation}
The estimator in \eqref{eqn:myestimator} uses regression adjustments over each period, and reweight observations by weights $\hat{\gamma}_1, \hat{\gamma}_2$ (inputs of the estimator). The construction of such estimator leverages properties of influence functions \citep{tchetgen2012semiparametric}.   We will omit the arguments $(\hat{\gamma}_1, \hat{\gamma}_2)$ in $\hat{\mu}_2$ whenever clear from the context.

A first choice of the weights are inverse probability weights (IPW). As for multi-valued treatments \citep{imbens2000role}, these weights for the first and second period can be written as 
\begin{equation} \label{eqn:ipw} 
\frac{1\{D_{i,1} = d_1\}}{n P(D_{i,1}  = d_1 | X_{i,1})}, \quad \frac{1\{D_{i,1} = d_1\}}{n P(D_{i,1}  = d_1 | X_{i,1})} \times  \frac{1\{D_{i,2} = d_2\}}{P(D_{i,2} = d_2 | Y_{i,1}, X_{i,1}, X_{i,2}, D_{i,1})}. 
\end{equation}

However, in high dimensions, IPW weights require the correct specification of the propensity score, which in practice may be unknown. In particular, the conditions of the AIPW are that \citep[e.g.][]{bradic2021high}
$$
||\hat{e}^{(t)} - e^{(t)}||_2 = o_p(n^{-1/4}), \quad  || \hat{\beta}^{(t)} - \beta^{(t)}||_2 = o_p(n^{1/4})
$$ 
where $e^{(t)}$ denotes the propensity score in each period $t$. These conditions therefore, in high dimensions impose restrictions on the propensity score. Our goal is to find a set of weights so that we can achieve consistent estimation even without $o_p(n^{-1/4})$ estimation of the propensity score, and only require $ || \hat{\beta}^{(t)} - \beta^{(t)}||_1 = O_p(n^{-1/4})$. 

In addition, we would like weights that are stable in finite sample, by minimizing their norm.

\subsection{Dynamic covariate balancing}

  Motivated by these considerations, we leverage the local projection model and propose replacing IPW with more stable weights.

We start studying covariate balancing conditions induced by the local projection model. By denoting $\bar{X}_1$ the sample average of covariates $X_1$, 
we can write 
  \begin{equation} \label{eq:decomposition}
  \hat{\mu}_2(d_1, d_2)  =  \bar{X}_1\beta_{d_1, d_2}^{(1)}+ T_1 + T_2 +T_3, 
\end{equation} 
where 
\begin{equation} 
\small 
\begin{aligned} 
T_1&=  \Big( \hat{\gamma}_1(d_1, d_2)^\top X_1 - \bar{X}_{1}\Big)  (\beta_{d_1, d_2}^{(1)} - \hat{\beta}_{d_1, d_2}^{(1)})
+ \Big(\hat{\gamma}_2(d_1, d_2)^\top H_2 - \hat{\gamma}_1(d_1, d_2)^\top H_2\Big)   (\beta_{d_1, d_2}^{(2)} - \hat{\beta}_{d_1, d_2}^{(2)})
\end{aligned} 
\end{equation} 
and
\begin{align*} 
T_2 &=  \hat{\gamma}_2(d_1, d_2)^\top \Big[Y_2 - H_2 \beta_{d_1,d_2}^2\Big] , \quad T_3=   \hat{\gamma}_1(d_1, d_2)^\top \Big[H_2 \beta_{d_1,d_2}^2 -  X_1 \beta_{d_1, d_2}^{(1)} \Big]. 
\end{align*}

  The estimation error enters through $T_1$. We bound this as follows. 
  
 \begin{lem}[Covariate balancing conditions] \label{lem:balancing1} The following holds
\begin{align*} \label{eqn:eqbalance}
T_1 \le&
  \underbrace{\| \hat{\beta}_{d_1,d_2}^1 - \beta_{d_1, d_2}^{(1)}\| _1 \Big| \Big| \bar{X}_1 - \hat{\gamma}_1(d_1, d_2)^\top X_1 \Big| \Big|_{\infty}}_{(i)} +
\nonumber \\
 &  \underbrace{\| \hat{\beta}_{d_1,d_2}^2 - \beta_{d_1, d_2}^{(2)}\| _1 \Big| \Big| \hat{\gamma}_2(d_1, d_2)^\top H_2 - \hat{\gamma}_1(d_1, d_2)^\top H_2 \Big| \Big|_{\infty}}_{(ii)}.
\end{align*} 
\end{lem} 
 Element $(i)$ is equivalent to what is discussed in \cite{athey2018approximate} in one period setting. Element $(ii)$ depends on the additional error induced by dynamics in the second period. 
 Therefore  the above suggests controlling the following norms 
 \begin{equation} \label{eqn:bbb} 
 \Big| \Big| \bar{X}_1 - \hat{\gamma}_1(d_1, d_2)^\top X_1 \Big| \Big|_{\infty}, \quad   \Big| \Big| \hat{\gamma}_2(d_1, d_2)^\top H_2 - \hat{\gamma}_1(d_1, d_2)^\top H_2 \Big| \Big|_{\infty} .
 \end{equation} 
 By imposing that the first norm converges to zero, the weights in the first-period balance covariates in the first period only. The second condition requires that histories in the second period are balanced, \textit{given} the weights in the previous period. 


The remaining terms in \eqref{eq:decomposition} are mean zero under the following conditions. 

\begin{lem}[Balancing error] \label{lem:residual} Let assumptions \ref{ass:noant} - \ref{ass:linearity} hold. Suppose that $\hat{\gamma}_1$ is measurable with respect to the sigma algebra $\sigma(X_1, D_1)$ and $\hat{\gamma}_2$ is measurable with respect to the sigma algebra $\sigma(X_1, X_2, Y_1, D_1, D_2)$. Suppose in addition that $\hat{\gamma}_{i,1}(d_1, d_2) = 0$ if $D_{i,1} \neq d_1$ and $\hat{\gamma}_{i,2}(d_1, d_2) = 0$ if $(D_{i,1}, D_{i,2}) \neq (d_1, d_2)$. Then 
$$ 
\mathbb{E}\Big[T_2 \Big| X_1, D_1, Y_1, X_2, D_2\Big] = 0, \quad \mathbb{E}\Big[T_3 \Big| X_1, D_1\Big] = 0.   
$$ 
\end{lem} 

The proof is in Appendix \ref{lem:33}. 
Lemma \ref{lem:residual} conveys a key insight: if we can guarantee that each component in Equation \eqref{eqn:bb} is sufficiently small, $\hat{\mu}$ is centered around the target estimand plus a small estimation error (since $\mathbb{E}[\bar{X}_1] \beta_{d_1, d_2}^{(1)} = \mathbb{E}[Y_{i,2}(d_1, d_2)]$). Lemma \ref{lem:residual} imposes the following intuitive conditions. The balancing weights in the first period are non-zero only for those units whose assignment in the first period coincide with the target assignment $d_1$, and similarly for the assignments $(d_1, d_2)$ in the second period. Moreover, we can only balance based on information observed before the realization of potential outcomes but not based on future information. A special case of balancing weights satisfying such conditions are IPW weights in Equation \eqref{eqn:ipw}. 

Figure \ref{fig:illustration} provides a visual description of the balancing condition, when the goal is to estimate $\mathbb{E}[Y_2(1,1)]$. In the first period we balance covariates of those individuals with shaded areas (both light and dark gray) with those of all the individuals (red box). That is, we balance the covariates in the first period between the treated and control units \textit{in the first period only}. In the second period we balance covariates between the gray and black box, once we reweight the covariates by the weights obtained in the first period. 

 \begin{figure}[!ht]
 \centering
 \vspace{-8mm}
    \begin{tikzpicture}[auto,
 block/.style ={rectangle, draw=blue, thick, fill=blue!20, text width=5em,align=center, rounded corners, minimum height=2em},
 block1/.style ={rectangle, draw=blue, thick, fill=blue!20, text width=5em,align=center, rounded corners, minimum height=2em},
 line/.style ={draw, thick, -latex',shorten >=2pt},
 cloud/.style ={draw=red, thick, ellipse,fill=red!20,
 minimum height=1em}]

\coordinate (1) at (-4,3);
\coordinate (2) at (-2,3);
\coordinate (3) at (-2,5);
\coordinate (4) at (-4,5);
\coordinate (5) at ($(1)!.5!(2)$); 
\coordinate (6) at ($(2)!.5!(3)$);
\coordinate (7) at ($(3)!.5!(4)$);
\coordinate (8) at ($(1)!.5!(4)$);
\coordinate (9) at ($(1)!.5!(3)$);

\coordinate (10) at (-4,3);
\coordinate (11) at (-2,3);
\coordinate (12) at (-2,1);
\coordinate (13) at (-4,1);
\coordinate (14) at ($(10)!.5!(11)$); 
\coordinate (15) at ($(11)!.5!(12)$);
\coordinate (16) at ($(12)!.5!(13)$);
\coordinate (17) at ($(10)!.5!(13)$);
\coordinate (18) at ($(10)!.5!(14)$);

\coordinate (21) at (-2,3);
\coordinate (22) at (0,3);
\coordinate (23) at (0,5);
\coordinate (24) at (-2, 5);
\coordinate (25) at ($(21)!.5!(22)$); 
\coordinate (26) at ($(22)!.5!(23)$);
\coordinate (27) at ($(23)!.5!(24)$);
\coordinate (28) at ($(21)!.5!(24)$);
\coordinate (29) at ($(21)!.5!(23)$);

\coordinate (30) at (-2,3);
\coordinate (31) at (0,3);
\coordinate (32) at (0,1);
\coordinate (33) at (-2,1);
\coordinate (34) at ($(30)!.5!(31)$); 
\coordinate (35) at ($(31)!.5!(32)$);
\coordinate (36) at ($(32)!.5!(33)$);
\coordinate (37) at ($(30)!.5!(33)$);
\coordinate (38) at ($(30)!.5!(34)$);


  \draw [fill=,lightgray,lightgray] (-4,5) rectangle (-2,3); 
    
   \draw [fill=darkgray] (-2,5) rectangle (0,3);

    \node[circle] (a) at (-3, 5.8) {$D_{i,2} = 1$};
   \node[circle] (a) at (-1, 5.8) {$D_{i,2} = 0$};
   
   \node[circle] (a) at (-5.3, 4) {$D_{i,1} = 1$};
     \node[circle] (a) at (-5.3, 2) {$D_{i,1} = 0$};

\draw (1)--(2)--(3)--(4)-- cycle;

\draw (21)--(22)--(23)--(24)-- cycle;
\draw (10)--(11)--(12)--(13)-- cycle;
 \draw (30)--(31)--(32)--(33)-- cycle;
  
 \draw[red,thick] ($(-2.3, 2.2)+(-2,3)$)  rectangle ($(0,0.8)+(0.3,0)$);
   \draw[thick] ($(-2.1, 2.2)+(-2,3)$)  rectangle ($(0,2.7)+(0.1,0)$);


    \end{tikzpicture}
   
  \caption{Illustration for balancing when estimating $\mathbb{E}[Y_2(1,1)]$. }  \label{fig:illustration}
\end{figure}

\begin{algorithm} [!htp]   \caption{Dynamic  covariate balancing (DCB): multiple time periods}\label{alg:alg3}
    \begin{algorithmic}[1]
    \Require Observations $\{Y_{i,1}, X_{i,1}, D_{i,1},\cdots, Y_{i,T}, X_{i,T}, D_{i,T}\}$, treatment history $(d_{1:T})$, finite parameters $\{K_{1,t}\}_{t=1}^T$, constraints $\delta_1(n,p_1), \delta_2(n,p_2), \cdots, \delta_T(n,p_T)$. 
    \State   Estimate $\beta_{d_{1:T}}^{(t)} $, running Algorithm \ref{alg:coefficients1} recursively for $T$ (instead of two) periods. 
    \State Let $\hat{\gamma}_{i,0} = 1/n$ and $t=0$; 
    \ForEach{$t \in \{1, \cdots, T\}$}
   
     \State  Estimate time $t$ weights with
\begin{equation}  \label{eqn:constraint_set}
\begin{aligned} 
 \hat{\gamma}_t  = \arg\min_{\gamma_t} \sum_{i = 1}^n \gamma_{i,t}^2, \quad  
& \text{s.t. }   \Big\|\frac{1}{n} \sum_{i = 1}^n \Big(\hat{\gamma}_{i,t-1} H_{i,t} - \gamma_{i,t}H_{i,t}\Big) \Big\|_{\infty} \le K_{1,t} \delta_t(n,p_t),  
\\
&  \quad 1^\top \gamma_{t} = 1, \gamma_{t} \ge 0,  \| \gamma_{t}\| _{\infty} \le  \log(n) n^{-2/3},   \gamma_{i,t} \le \prod_{s=1}^t 1\{D_{i,s} = d_s\} \forall i
\end{aligned}
\end{equation} 
\EndFor
  \Comment{obtain $T$ balancing vectors}

        \Return Estimate of the average potential outcome as in Equation \eqref{eqn:generalest}

         \end{algorithmic}
\end{algorithm}

\subsection{Complete algorithm description} \label{sec:parameters}

Algorithm \ref{alg:alg3} presents the details to estimate the balancing weights with $T$ periods. 
First, we construct weights in the first period that are nonzero only for those individuals with treatment at time $t = 1$ equal to the target treatment status $d_1$. In the first period, we balance baseline covariates between the treated and control groups as in \cite{athey2018approximate}. Second, we estimate $\hat{\gamma}_{2}$ for the desired treatment history $(D_{i,1}, D_{i,2}) = (d_1, d_2)$. The weights $\hat{\gamma}_{i,2}$ are not zero only for individuals with treatment history $(D_{i,1}, D_{i,2}) = (d_1, d_2)$  as discussed in Lemma \ref{lem:residual}. The estimated weights $\hat{\gamma}_{2}$ balance observable characteristics between different treatment groups at time $t = 2$, after \textit{reweighting} with the weights estimated in the previous period. We choose weights that sum to one, are positive, and do not assign the largest weight to a few observations. For each period, the optimization problem solves a quadratic program recursively 
that minimizes the weights' variances.

The program in Algorithm \ref{alg:alg3} is a sequence of $T$ quadratic programs with linear constraints. Its computational complexity scales polynomially with $n, p$. Figure \ref{fig:long_T2} provides an example showing that the computational time is between a few seconds and a few minutes for $T \in \{1, \cdots, 10\}$ on a personal laptop.  For $T = 20$ in our empirical application the running time takes approximately 30 minutes to run on one core on a personal laptop (in practice, researchers can run the optimization on ten or more cores).

\begin{rem}[Long panels, pooled estimation and imbalanced panels] \label{rem:pooled} In the presence of long panels, researchers may consider a 
regression of the form 
$$
Y_{i,t}(d_{1:t}) =  \beta_0 + \beta_1 d_t + \beta_2 Y_{i,t-1}(d_{1:(t-1)}) + X_{i,t}(d_{1:(t-1)}) \gamma +  \tau_t + \varepsilon_{i,t},  
$$  
where $\tau_t$ denotes fixed effects, and an estimand of the form
\begin{equation} \label{eqn:estimand2}
 \mathbb{E}[Y_{i,t + h}(D_{1:t}, d_{t+1}, \cdots, d_{t+h})] - \mathbb{E}[Y_{i,t + h}(D_{1:t}, d_{t+1}', \cdots, d_{t+h}')],   
\end{equation}
for given $d_{t+1:(t+h)}, d_{(t+1):(t +h)}'$. Equation \eqref{eqn:estimand2} denotes the effect of changing the treatment path in the last $h$ periods, averaging over past assignments. The regression specification and the estimand in Equation \eqref{eqn:estimand2} reduce the dimensionality as $T$ grows. It is possible to include groups fixed effects directly in the covariates of interest. 
In particular, because of the definition of the estimand, researchers only have to estimate $h$ instead of $T$ balancing weights. Inference must cluster standard errors at the individual level. 

Finally, note that our method allows for imbalanced panels since estimation is performed sequentially (both for the coefficients and for the weights). If some observations are missing over some periods, the algorithm will exclude such units when estimating the coefficients and weights for that period(s) but not for the remaining ones. However, a sufficiently large (proportional to $n$) number of observations must be observed for $h$ consecutive periods. \qed 
\end{rem}

\subsection{Tuning parameters: practical implementation} \label{app:algorithms}


 Similarly to balancing in one-dimensional setting \citep{athey2018approximate}, Algorithm \ref{alg:alg3} requires choosing tuning parameters for Equation \eqref{eqn:constraint_set}. A complete description is in Algorithm \ref{alg:algtuning} and uses a data-adaptive procedure. In a nutshell, we choose $\delta_t(n,p) = \log^{3/2}(p_t n)/n^{1/2}$ (here $p_t$ is the dimension of covariates at time $t$) as prescribed by the theoretical analysis in Section \ref{sec:theory}. To guarantee balance with many covariates, we first select the smallest constant $K_1$ for covariates with non-zero estimated coefficients and second the smallest constant for the remaining covariates until a feasible solution is reached. This choice minimizes the estimator's bias and, within the set of weighting estimators with the smallest bias, selects the one with the smallest variance while prioritizing balance on covariates with non-zero coefficients. Section \ref{sec:num} illustrates the benefits of this procedure.


\begin{algorithm}[!h]   \caption{Tuning Parameters for DCB}\label{alg:algtuning}
    \begin{algorithmic}[1]
    \Require Observations $\{Y_{i,1}, X_{i,1}, D_{i,1},..., Y_{i,T}, X_{i,T}, D_{i,T}\}$, $\delta_t(n,p)$, treatment history $(d_{1:T})$,  $L_t, U_t$, grid length $G$, number of grids $R$. 
    \State Estimate coefficients as in Algorithm 2  and let $\hat{\gamma}_{i,0} = 1/n$;
    \State Define $R$ grids of length $G$, denoted as $\mathcal{G}_1, ..., \mathcal{G}_R$, equally between $L_t$ an $U_t$. 
    \State Define  
    $$
    \mathcal{S}_1 = \{j: |\hat{\beta}^{t, (j)}| \neq 0\}, \quad \mathcal{S}_2 = \{j: |\hat{\beta}^{t, (j)}| = 0\}.
    $$
    \State (Non-sparse regression): if $|\mathcal{S}_1|$ is too large (i.e., $>  \mathrm{dim}(\hat{\beta}^t)/3$), select $\mathcal{S}_1$ the set of the $1/3^{rd}$ largest coefficients in absolute value and $\mathcal{S}_2 = \mathcal{S}_1^c$.  
    \ForEach {$s_1 \in 1:G $}
    \ForEach {$K_{1,t}^a \in \mathcal{G}_{s_1}$}
    \ForEach {$K_{1,t}^b \in \mathcal{G}_{s_1}$}
    \State  Let  $\hat{\gamma}_{i,t} = 0, \text{ if } D_{i,1:t} \neq d_{1:t}$ and define $\hat{\gamma}_t:= \text{argmin}_{\gamma_t} \sum_{i = 1}^n \gamma_{i,t}^2$
\begin{equation} 
\begin{aligned} 
\text{s.t. } & \Big|\frac{1}{n} \sum_{i = 1}^n \hat{\gamma}_{i,t-1} H_{i,t}^{(j)} - \gamma_{i,t}H_{i,t}^{(j)} \Big| \le K_{1,t}^a \delta_t(n,p), \quad \forall j: \hat{\beta}^{t, (j)} \in \mathcal{S}_1 \\ &\Big|\frac{1}{n} \sum_{i = 1}^n \hat{\gamma}_{i,t-1} H_{i,t}^{(j)} - \gamma_{i,t}H_{i,t}^{(j)} \Big| \le K_{1,t}^b \delta_t(n,p) \quad \forall j:\hat{\beta}^{t, (j)} \in \mathcal{S}_2  \\ &\sum_{i = 1}^n \gamma_{i,t} = 1, \quad ||\gamma_{t}||_{\infty} \le  \log(n) n^{-2/3}, \gamma_{i,t} \ge 0. 
\end{aligned}
\end{equation} 
\State \textbf{Stop if}: a feasible solution exists. 
\EndFor
\EndFor
\EndFor

        \Return $\hat{\mu}_T(d_{1:T})$ 

         \end{algorithmic}
\end{algorithm}

\subsection{Comparisons with standard local projections} \label{sec:lp1}

We pause here and highlight key differences from local projections in economics. 
Different from standard local projections, and building on the literature on marginal structural models \citep{robins2000marginal} here we assume and identify a local projection model for potential outcomes $Y_{2}(d_1, d_2)$ instead of realized outcomes $Y_2$ \citep[see][for reviews of these models]{jorda2005estimation, montiel2021local}. To illustrate this difference, consider a two periods settings without time-varying covariates, and
\begin{equation} \label{eqn:y} 
\begin{aligned} 
Y_{i,t} & =  Y_{i,t - 1} \alpha + D_{i,t} \beta + X_{i,1} \gamma  + \varepsilon_{i, t}, \quad \mathbb{E}\Big[\varepsilon_{i,t} \Big| Y_{i,t-1}, D_{i,t}, D_{i,t-1}, X_{i,1}\Big] = 0. 
\end{aligned} 
\end{equation} 
Under Equation \eqref{eqn:y}, we can write 
\begin{equation}  \label{eqn:cc} 
\begin{aligned} 
\mathbb{E}\Big[Y_{i,2} | X_{i,1}, D_{i,1}\Big] & = \alpha \beta D_{i,1} + \beta \mathbb{E}\Big[D_{i,2} | X_{i,1}, D_{i,1}\Big] + X_{i,1} (\gamma + \alpha \gamma) \\ 
\mathbb{E}\Big[Y_{i,2} | X_{i,1}, D_{i,2}, D_{i,1}\Big] & = \alpha \beta D_{i,1} + \beta D_{i,2}  + X_{i,1} (\gamma + \alpha \gamma) + \alpha \mathbb{E}[\varepsilon_{i, 1} | D_{i,2}, X_{i,1}] \\ 
\mathbb{E}\Big[Y_{i,2}(d_1, d_2) | X_{i,1}, D_{i,1}\Big] & = \alpha \beta d_1 + \beta d_2 + X_{i,1} (\gamma + \alpha \gamma)
\end{aligned}
\end{equation} 
The first equation is a reduced form corresponding to regressing the outcome at time $t = 2$ on past treatments and covariates, the second equation is its equivalent also controlling for $D_{i,2}$, and the third is the reduced form for potential outcomes.  

The first equation in \eqref{eqn:cc} shows that the interpretation of the parameters of the local projection of the observed outcome $Y_{i,2}$ onto $(D_{i,1}, X_{i,1})$ depends on properties of $\mathbb{E}\Big[D_{i,2} | X_{i,1}, D_{i,1}\Big]$. Once we project $Y_{i,2}$ onto $(D_{i,1}, X_{i,1})$, the estimated coefficient for $D_{i,1}$ denotes the effect of treating an individual at time $t = 1$ only if $D_{i,2}$ and $D_{i,1}$ are independent.

The second equation shows that controlling for $D_{i,2}$ may lead to omitted variable bias on the coefficient multiplying $D_{i,2}$ if future treatments depend on past outcomes. 
Therefore, standard local projections recover estimands whose interpretation depends on the distribution of the treatments, whereas the treatments' distribution may change with a change in policy. A third possible approach and different from the specifications in \eqref{eqn:cc} is to estimate each equation for $Y_{i,t}, Y_{i,t-1}, X_{i,t}$ and obtain the desired $\mathrm{ATE}(\cdot)$ through products and sums of the coefficients. This approach is prone to large estimation error with high-dimensional time-varying covariates because it estimates a separate model for each covariate. This follows similarly to discussion motivating local projections over vector auto-regression models \citep{jorda2005estimation, montiel2021local}.  Our approach leverages the potential outcome model in the third equation in \eqref{eqn:cc}, whose parameters do not depend on the realized treatments $D's$.

Our problem also connects to the literature on two-way fixed effects and Difference-in-Differences \citep[see][for an overview]{roth2023s}. This literature focuses on staggered adoption, whereas treatments here can change arbitrarily over time. In particular, the parallel trend assumption in DiD designs prohibits dynamic selection into treatment considered here \citep[see][for a discussion]{ghanem2022selection, marx2022parallel}.

Using Equation \eqref{eqn:cc} for a simple illustration, we can interpret a version of parallel trends as 
\begin{equation} \label{eqn:pt} 
\mathbb{E}\Big[\varepsilon_{i,t} - \varepsilon_{i,t-1}\Big| D_{i, 1:T} = \mathbf{1}\Big] = \mathbb{E}\Big[\varepsilon_{i,t} - \varepsilon_{i,t-1} \Big| D_{i, 1:T} = \mathbf{0}\Big], 
\end{equation}  
violated when \textit{future} assignments depend on past outcomes (and therefore $\varepsilon_{i,t}$).

Therefore, DiD designs are suited in the presence of unobserved additive unit-level confounders but lack of dynamic selection different from (and complementary to) our framework.

\subsection{Discussion: Estimation of policy treatment rules} \label{sec:appendix_when_to_treat} 

In this subsection we briefly illustrate balancing weights can be extended to study treatment policy sequences in a simple linear model, focusing on the two periods settings for brevity. We leave to future research studying this extension to non-linear settings. Following Remark \ref{rem:when_to_treat}, consider a simple linear model of the form (which we interpret as an approximation in high dimensional settings) 
$$
\mathbb{E}\Big[Y_{i,2}(\pi) | X_{i,1}\Big] = X_{i,1} \tilde{\beta}_\pi^{(1)}, \quad \mathbb{E}\Big[Y_{i,2}(\pi) | H_{i,2}\Big]  = H_{i,2} \tilde{\beta}_{\pi}^{(2)}. 
$$ 
Identification and estimation of the coefficients follows similarly what discussed in Lemma \ref{lem:identification_model1}, where we can estimate $\tilde{\beta}_{\pi}^{(2)}$ by regressing 
$$
\Big[Y_{i,2}(\pi) | H_{i,2}, D_{i,1} = \pi_1(X_{i,1}), D_{i,2} = \pi_2(H_{i,2})\Big] = \mathbb{E}\Big[Y_{i,2}(\pi) | H_{i,2}\Big] 
$$ 
(noting that $X_{i,1}$ is measurable with respect to $H_{i,2}$). For $\tilde{\beta}_{\pi}^{(1)}$ we follow the identification strategy 
$$
\mathbb{E}\left[\Big[Y_{i,2}(\pi) | H_{i,2}, D_{i,1} = \pi_1(X_{i,1}), D_{i,2} = \pi_2(H_{i,2})\Big]\Big|X_{i,1}, D_{i,1} = \pi_1(X_{i,1})\right] = \mathbb{E}\Big[Y_{i,2}(\pi) | X_{i,1}\Big]. 
$$ 
We can then construct balancing weights as in Algorithm 1 where we replace the constraint $\gamma_{i,t} \le \prod_{s = 1}^t 1\{D_{i,s} = d_s\}$ with the constraint $\gamma_{i,t} \le \prod_{s=1}^t 1\{D_{i,s} = \pi_s(H_{i,s})\}$. We construct an estimator of the following form 
$$
\hat{V}(\pi) =  \sum_{i=1}^n \left\{\hat{\gamma}_{i,2}   Y_{i,2} -   \Big(\hat{\gamma}_{i,2} - \hat{\gamma}_{i,1} \Big) H_{i,2} \hat{\beta}_{\pi}^{(2)} -  \Big(\hat{\gamma}_{i,1}  - \frac{1}{n} \Big) X_{i,1} \hat{\beta}_{\pi}^{(1)} \right\}.
$$ 
The core intuition is that because $\hat{\gamma}_{i,2} > 0$ only if $D_{i,1} = \pi_1(X_{i,1})$ \textit{and} $D_{i,2} = \pi_2(H_{i,2})$, and $\hat{\gamma}_{i,1} > 0$ only if $D_{i,1} = \pi_1(X_{i,1})$, we can equivalently write 

$$
\hat{V}(\pi) =  \sum_{i=1}^n \left\{\hat{\gamma}_{i,2}   Y_{i,2}(\pi) -   \Big(\hat{\gamma}_{i,2} - \hat{\gamma}_{i,1} \Big) H_{i,2}(\pi) \hat{\beta}_{\pi}^{(2)} -  \Big(\hat{\gamma}_{i,1}  - \frac{1}{n} \Big) X_{i,1} \hat{\beta}_{\pi}^{(1)} \right\}, 
$$ 
where we used consistency of potential outcomes for $Y_{i,2}(\pi), H_{i,2}(\pi)$. We can then follow verbatim our argument in the main text to motivate and justify the (linear) balancing constraints.

 \section{Lemmas} \label{sec:lem}

\subsection{Proof of Lemma \ref{lem:identification_model1}} \label{sec:lemma1}

 The first equation in Lemma \ref{lem:identification_model1} is a direct consequence of condition (A) in Assumption \ref{ass:seqign}, and the linear model assumption (Assumption \ref{ass:linearity}). Consider the second equation in Lemma \ref{lem:identification_model1}. By condition (B) in Assumption \ref{ass:seqign}, we have  
$$
\begin{aligned} 
 \mathbb{E}\Big[Y_{i,2}(d_1, d_2) \Big| X_{i,1} \Big] = \mathbb{E}\Big[ Y_{i,2}(d_1, d_2) \Big| X_{i,1}, D_{i,1} = d_1\Big]. 
\end{aligned} 
$$ 
Using the law of iterated expectations (since $X_{i,1}$ is measurable with respect to $H_{i,2}$)
$$
\begin{aligned} 
\mathbb{E}\Big[ Y_{i,2}(d_1, d_2) \Big| X_{i,1}, D_{i,1} = d_1\Big] = \mathbb{E}\Big[\mathbb{E}[Y_{i,2}(d_1,d_2) | H_{i,2}, D_{i,1} = d_1] | X_{i,1}, D_{i,1} = d_1\Big].  
\end{aligned} 
$$ 
Using condition (A) in Assumption \ref{ass:seqign}, we have 
$$
\mathbb{E}[Y_{i,2}(d_1,d_2) | H_{i,2}, D_{i,1} = d_1] = \mathbb{E}[Y_{i,2}(d_1,d_2) | H_{i,2}, D_{i,1} = d_1, D_{i,2} = d_2]
$$ 
the proof completes as $\mathbb{E}[Y_{i,2}(d_1,d_2) | H_{i,2}, D_{i,1} = d_1, D_{i,2} = d_2] = \mathbb{E}[Y_{i,2} | H_{i,2}, D_{i,1} = d_1, D_{i,2} = d_2]$ as a consequence of condition (A) in Assumption \ref{ass:seqign}.

\subsection{Proof of Lemma \ref{lem:lemmam1} }  \label{app:lem22}

We prove the main lemmas for multiple periods as the two-periods case is a special case. 

 Throughout the proof we omit the argument $d_{1:T}$ of $\hat{\gamma}_t(d_{1:T})$ for notational convenience. 
 Recall that $\hat{\gamma}_{i,T} = 0$ if $D_{i, 1:T} \neq d_{1:T}$. Therefore, 
 $$
 \small 
 \begin{aligned} 
 \hat{\gamma}_{i,T} Y_{i,T} =  \hat{\gamma}_{i,T} Y_{i,T}(d_{1:T}) = \hat{\gamma}_{i,T}(H_{i,T} \beta_{d_{1:T}}^T + \varepsilon_{i,T}). 
 \end{aligned} 
 $$
 We can write 
$$
\small 
\begin{aligned} 
& \sum_{i=1}^n \Big(\hat{\gamma}_{i,T} Y_{i,T} - \sum_{t = 2}^T (\hat{\gamma}_{i,t} - \hat{\gamma}_{i,t-1}) H_{i,t} \hat{\beta}_{d_{1:T}}^t - (\hat{\gamma}_{i,1} - \frac{1}{n}) X_{i,1} \hat{\beta}_{d_{1:T}}^1\Big) \\
&=  \sum_{i=1}^n \Big(\hat{\gamma}_{i,T} H_{i,T} \beta_{d_{1:T}}^T - \sum_{t = 2}^T (\hat{\gamma}_{i,t} - \hat{\gamma}_{i,t-1}) H_{i,t} \hat{\beta}_{d_{1:T}}^t - (\hat{\gamma}_{i,1} - \frac{1}{n}) X_{i,1} \hat{\beta}_{d_{1:T}}^1\Big)  + \hat{\gamma}_T^\top \varepsilon_T.
\end{aligned} 
$$
Consider first the term 
$$
\begin{aligned} 
& \sum_{i=1}^n  (\hat{\gamma}_{i,T} H_{i,T} \beta_{d_{1:T}}^T -(\hat{\gamma}_{i,T} - \hat{\gamma}_{i,T-1})H_{i,T} \hat{\beta}_{d_{1:T}}^T) = (\hat{\gamma}_{T}^\top H_T - \hat{\gamma}_{T-1}^\top H_{T}) (\beta_{d_{1:T}}^T - \hat{\beta}_{d_{1:T}}^T) + \hat{\gamma}_{T-1}^\top H_T \beta_{d_{1:T}}^T. 
\end{aligned} 
$$
For any $s > 1$, 
$$
\small 
\begin{aligned} 
&\sum_{i=1}^n  (\hat{\gamma}_{i,s} - \hat{\gamma}_{i,s-1})H_{i,s} \hat{\beta}_{d_{1:T}}^s = (\hat{\gamma}_{s}^\top H_{s} - \hat{\gamma}_{s-1}^\top H_{s}) (\hat{\beta}_{d_{1:T}}^s - \beta_{d_{1:T}}^s) + \hat{\gamma}_{s}^\top H_{s} \beta_{d_{1:s}}^s - \hat{\gamma}_{s - 1}^\top H_s \beta_{d_{1:s}}^s. 
\end{aligned} 
$$
For $s = 1$, it follows 
$$
\small 
\begin{aligned} 
&\sum_{i=1}^n  (\hat{\gamma}_{i,1} - \frac{1}{n})X_{i,1} \hat{\beta}_{d_{1:T}}^1 = (\hat{\gamma}_{1}^\top X_{1} - \bar{X}_1) (\hat{\beta}_{d_{1:T}}^1 - \beta_{d_{1:T}}^1) + \hat{\gamma}_{1}^\top X_{1} \beta_{d_{1:s}}^1 - \bar{X}_1 \beta_{d_{1:s}}^1. 
\end{aligned} 
$$
Therefore, we can write 
$$
\small 
\begin{aligned} 
&\sum_{i=1}^n \Big(\hat{\gamma}_{i,T} Y_{i,T} - \sum_{t = 2}^T (\hat{\gamma}_{i,t} - \hat{\gamma}_{i,t-1}) H_{i,t} \hat{\beta}_{d_{1:T}}^t - (\hat{\gamma}_{i,1} - \frac{1}{n}) X_{i,1} \hat{\beta}_{d_{1:T}}^1\Big) \\ &= 
 (\hat{\gamma}_{T}^\top H_T - \hat{\gamma}_{T-1}^\top H_{T}) (\beta_{d_{1:T}}^T - \hat{\beta}_{d_{1:T}}^T) + \sum_{s=2}^{T - 1} (\hat{\gamma}_{s}^\top H_{s} - \hat{\gamma}_{s-1}^\top H_{s}) (\beta_{d_{1:T}}^s - \hat{\beta}_{d_{1:T}}^s) + (\hat{\gamma}_{1}^\top X_{1} - \bar{X}_1) (\beta_{d_{1:T}}^1 - \hat{\beta}_{d_{1:T}}^1)  \\ &  + \hat{\gamma}_{T-1}^\top H_T \beta_{d_{1:T}}^T  + \gamma_T^\top \varepsilon_T - \Big[\sum_{s=2}^{T - 1}  \hat{\gamma}_{s}^\top H_{s} \beta_{d_{1:s}}^s - \hat{\gamma}_{s-1}^\top H_s \beta_{d_{1:s}}^s\Big]  - \hat{\gamma}_{1}^\top X_{1} \beta_{d_{1:s}}^1 + \bar{X}_1 \beta_{d_{1:s}}^1. 
\end{aligned} 
$$ 
The proof completes after collecting the desired terms.

\subsection{Additional lemma in multiple periods and proof of Lemma \ref{lem:residual}} \label{lem:33}

We prove Lemma \ref{lem:residual} as a special case of the following lemma. 

\begin{lem} \label{lem:lemmam2}
Let Assumption \ref{ass:seqignm} hold. Suppose that the sigma algebra $\sigma(\hat{\gamma}_t(d_{1:T})) \subseteq \sigma(H_t ,D_t)$. Suppose in addition that $\hat{\gamma}_{i,t}(d_{1:T}) = 0$ if $D_{i,1:t} \neq d_{1:t}$. Then 
$$
\small 
\begin{aligned} 
\mathbb{E}\Big[\hat{\gamma}_{i,t-1}(d_{1:T})H_t \beta_{d_{1:T}}^{(t)} - \hat{\gamma}_{i,t-1}(d_{1:T}) H_{t-1} \beta_{d_{1:T}}^{(t-1)} \Big| H_{t-1}, D_{t-1}\Big] = 0.
\end{aligned} 
$$
\end{lem} 

\begin{proof} 
Since $\hat{\gamma}_{i,t}(d_{1:T})$ is equal to zero if $D_{i,1:t} \neq d_{1:t}$ we can focus to the case where $D_{i,1:t} = d_{1:t}$. Since weights at time $t-1$ are measurable with respect to $H_{t-1}, D_{t-1}$, we only need to show 
\begin{equation} 
 \mathbb{E}[\hat{\gamma}_{i,t - 1}(d_{1:T}) H_t \beta_{d_{1:T}}^t | H_{t-1}, D_{-i, t-1}, D_{i,(1:(t-1))} = d_{1:(t-1)}] =  \hat{\gamma}_{i,t}(d_{1:T}) H_{t-1} \beta_{d_{1:T}}^{t-1}.  
\end{equation} 
 On the event that $D_{i,(1:(t-1))} \neq d_{1:(t-1)}$ the expression is zero on both sides and the result trivially holds. Therefore, we can implicitely assume that $D_{i, (1:(t-1))} = d_{1:(t-1)}$ since otherwise the result trivially holds. 
Under Assumption \ref{ass:seqignm} we can write 
\begin{equation}
\begin{aligned} 
 \mathbb{E}[\hat{\gamma}_{i,t - 1}(d_{1:T}) H_t \beta_{d_{1:T}}^t | H_{t-1}, D_{t-1}] & = \mathbb{E}\Big[ \hat{\gamma}_{i,t - 1}(d_{1:T}) \mathbb{E}[ Y_{i,T}(d_{1:T}) | H_t, D_t] \Big |H_{t-1}, D_{t-1}\Big] \\ & = 
\hat{\gamma}_{i,t - 1}(d_{1:T}) \mathbb{E}[ Y_{i,T}(d_{1:T}) | H_{t-1},D_{t-1}] 
\end{aligned} 
\end{equation}  
by the tower property of the expectation and the definition of $H_t$. Now notice that under Assumption \ref{ass:seqignm} (B), 
$
\begin{aligned} 
\mathbb{E}[ Y_{i,T}(d_{1:T}) | H_{t-1}, D_{t-1}] = \mathbb{E}[ Y_{i,T}(d_{1:T}) | H_{t-1} ] .
\end{aligned} 
$
Therefore  
\begin{equation}
\hat{\gamma}_{i,t - 1}(d_{1:T}) \mathbb{E}[ Y_{i,T}(d_{1:T}) | H_{t-1} ] = \hat{\gamma}_{i,t - 1}(d_{1:T}) H_{i,t-1} \beta_{d_{1:(t-1)}}^{t-1} .
\end{equation} 
  \end{proof}

\subsection{Sufficient conditions for lasso} \label{aa:suff}

In this section we provide sufficient conditions for the convergence rate of lasso in two periods as in Assumption \ref{ass:4m} (i). $T$-periods follows similarly. 
\begin{lem}[Sufficient conditions for Lasso] \label{lem:suff} Suppose that $||H_2||_{\infty}, ||X_1||_{\infty}$ are uniformly bounded almost surely and $||\beta_{d_{1:2}}^{(2)}||_0, ||\beta_{d_{1:2}}^{(1)}||_0 \le s, ||\beta_{d_{1:2}}^{(2)}||_{\infty}, ||\beta_{d_{1:2}}^{(1)}||_{\infty} \le s$. Suppose that $H_2, X_1$ both satisfy the restricted eigenvalue assumption, and the column normalization condition \citep{negahban2012unified}. (Sufficient conditions that guarantee that the restricted eigenvalue assumption holds are discussed in \citep{negahban2012unified}.) Suppose that $\hat{\beta}_{d_{1:2}}^{(1)}, \hat{\beta}_{d_{1:2}}^{(2)}$ are estimated with Lasso as in Algorithm 2 with a full interaction model and with penalty parameter $\lambda_n \asymp s\sqrt{\log(p)/n}$. Let Assumptions \ref{ass:noant} - \ref{ass:linearity} hold. Let $\varepsilon_2(d_{1:2}) | H_2$ be subgaussian almost surely and  $\nu_1(d_1) | X_1$ be sub-gaussian almost surely.   Then for $t \in \{1, 2\}$, 
$$
\Big| \Big| \hat{\beta}_{d_{1:2}}^{(t)} - 
\beta_{d_{1:2}}^{(t)} \Big| \Big|_1 = \mathcal{O}_p\Big(s^2 \sqrt{\log(p)/n}\Big) . 
$$ 
Therefore, 
$$
\| \hat{\beta}_{d_{1:2}}^{(t)} - \beta_{d_{1:2}}^{(t)} \| _1 \delta_t(n, p) = o_p(1/\sqrt{n}),$$ 
for $\delta_t(n,p) \asymp \log(np)/n^{1/4}$ and $s^2 \log^{3/2}(np)/n^{1/4} = \mathcal{O}(1)$. 
\end{lem} 
The proof is discussed below and follows similarly to \cite{negahban2012unified}, with minor modifications. The above result provides a set of sufficient conditions such that Assumption \ref{ass:4m} (i) holds for a feasible choice of $\delta_t$.
Interestingly, the estimation error propagates each period through stricter restrictions on the sparsity parameter $s$ compared to the standard lasso method. The reason is because of the recursive approach.

\begin{proof} The result for 
$$
\Big| \Big|\hat{\beta}_{d_{1:2}}^2 - \beta_{d_{1:2}}^2 \Big| \Big|_1 = \mathcal{O}_p\Big(s \sqrt{\log(p)/n}\Big)
$$
follows verbatim from \cite{negahban2012unified} Corollary 2. For the result for $\hat{\beta}_{d_{1:2}}^1$ it suffices to notice, following the same argument from \cite{negahban2012unified} (Corollary 1), that 
$$
\Big| \Big|\hat{\beta}_{d_{1:2}}^1 - \beta_{d_{1:2}}^1 \Big| \Big|_1 = O(s \lambda_n ), \text{ for } \lambda_n \ge \Big| \Big| \frac{1}{n} X_1^\top (H_2 \hat{\beta}_{d_{1:2}}^2 - X_1 \beta_{d_{1:2}}^1)\Big| \Big|_{\infty},  
$$ 
since here we used the estimated outcome $H_2 \hat{\beta}_{d_{1:T}}^2$ as the outcome of interest in our estimated regression instead of the true outcome and the error is defined as $H\hat{\beta}_{d_{1:2}}^{(t)} - X_1 \beta_{d_{1:2}}^1$. (Formally, here to compute $\mathcal{R}^*(\nabla \mathcal{L}(\theta^*))$ in \cite{negahban2012unified}'s notation we need to account for the loss function to depend on the estimated outcome.) The upper bound as a function of $\lambda_n$ follows directly from Theorem 1 in \cite{negahban2012unified}. (Note that Theorem 1 in  \cite{negahban2012unified} does not depend on the distribution of the data and is a deterministic statement which holds under strong convexity at the true regression parameter. For a linear model, strong convexity is satisfied under the restricted eigenvalue assumption which does not depend on the regression parameter.)
 We note that we can write 
$$
\begin{aligned} 
\Big| \Big|\frac{1}{n} X_1^\top (H_2 \hat{\beta}_{d_{1:2}}^2 - X_1 \beta_{d_{1:2}}^1)\Big| \Big|_{\infty} & \le \Big| \Big| \frac{1}{n} X_1^\top \nu_1\Big| \Big|_{\infty}  + \Big| \Big| \frac{1}{n} X_1^\top (H_2 \hat{\beta}_{d_{1:2}}^2 - H_2 \beta_{d_{1:2}}^2)\Big| \Big|_{\infty} \\ &= \Big| \Big| \frac{1}{n} X_1^\top \nu_1\Big| \Big|_{\infty}  + \Big| \Big| \frac{1}{n} X_1^\top H_2 (\beta_{d_{1:2}}^2 - \hat{\beta}_{d_{1:2}}^2)\Big| \Big|_{\infty}  \\
&\le \Big| \Big| \frac{1}{n} X_1^\top \nu_1\Big| \Big|_{\infty}  + ||X_1||_{\infty} ||H_2||_{\infty}  ||\beta_{d_{1:2}}^2 - \hat{\beta}_{d_{1:2}}^2||_1. 
\end{aligned}
$$ 
We now study each component separately. By sub-gaussianity, since $\mathbb{E}[\nu_1 | X_1] = 0$ by Assumption \ref{ass:linearity}, we have for all $t > 0$, by Hoeffding inequality and the union bound, 
$$
P\Big(\Big| \Big| \frac{1}{n} X_1^\top \nu_1\Big| \Big|_{\infty} > t\Big| X_1\Big) \le p \exp\Big(-M \frac{t^2 n}{s}\Big)
$$ 
for a finite constant $M$. This result follows since $\nu_1 \le ||\beta_1||_1 ||X_{1}^{(j)}||_{\infty} \le M s$. It implies that
$$
\Big| \Big| \frac{1}{n} X_1^\top \nu_1\Big| \Big|_{\infty} = O_p(\sqrt{s\log(p)/n})
$$ 
The second component instead is $O_p(s \sqrt{\log(p)/n})$ by the bound on $||\beta_{d_{1:2}}^2 - \hat{\beta}_{d_{1:2}}^2||_1$. This complete the proof. Finally, observe also that the same argument follows recursively for any finite $T$, with the estimation error depending on $T$. 
\end{proof} 

\begin{rem}[Restricted eigenvalue condition] \label{rem:restricted}  Note that here we impose the restricted eigenvalue condition on both $H$ and $X$. Assume that $H$ and $X$ have bounded entries for simplicity. Because $H$ has independent rows and similarly $X$ has independent rows (units are independent across $i$, but not over time necessarily), and because $H$ and $X$ have bounded entries, a sufficient condition for the restricted eigenvalue condition to hold with high probability is 
$$
\lambda_{min}\Big(\mathbb{E}[(X_i - \mathbb{E}[X_i])(X_i - \mathbb{E}[X_i])^\top]\Big), \lambda_{min}\Big(\mathbb{E}[(H_{i,2} - \mathbb{E}[H_{i,2}])(H_{i,2} - \mathbb{E}[H_{i,2}])^\top]\Big) \ge \bar{\lambda} > 0
$$ 
where $\lambda_{min}$ indicates the smallest eigenvalue and $\bar{\lambda}$ is a constant bounded away from zero (see Theorem 10.5.11 in \cite{vershynin2018high}). This condition poses a restriction on the time dependence (since $H_{i,2}$ contains covariates over both periods): we cannot have that $H_{i,2}$ is perfectly collinear across different columns. It allows however the column of $H_{i,2}$ to be  dependent but not collinear. References studying this condition include \cite{deshpande2023online}. \qed 
\end{rem} 

\subsection{Auxiliary Lemmas}

In the lemmas below we will refer to $a, c_0, C$ as finite constants. 

\begin{lem}[Existence of Feasible $\hat{\gamma}_1$] \label{lem:gamma_1_star} 
Suppose that $X_{i,1}^{(j)}$ is subgaussian for all $j \in \{1, \cdots, p_1\}$, $X_{i,1} \in \mathbb{R}^{p_1}$. Suppose that for $d_1 \in \{0,1\}$, $P(D_{i,1} = d_1 | X_{i,1}) \in(\delta, 1 - \delta)$, for some $\delta \in (0,1)$. For finite constants $c_0, C < \infty$, with probability at least $1 - 5/n$, for 
$\log(2np_1)/n \le c_0$, $\delta_1(n,p_1) \ge C\sqrt{2\log(2np_1)/n}$, there exists a feasible $\hat{\gamma}_1$ satisfying the constraints in Algorithm 1. In addition, 
$$\lim_{n \rightarrow \infty} P\Big(n ||\hat{\gamma}_1||_2^2 \le \mathbb{E}\Big[\frac{1}{P(D_{i,1} = d_1 | X_{i,1})}\Big] \Big) = 1. $$  
\end{lem}

\begin{proof}[Proof of Lemma \ref{lem:gamma_1_star}]
This proof follows similarly to the one-period setting in \cite{athey2018approximate}. 

\paragraph{Feasible guess} To prove existence of a feasible weight, we use a feasible guess. 
We prove the claim for a general $d_1 \in \{0,1\}$. Consider first
\begin{equation} \label{eqn:guess}
\small  
\begin{aligned} 
\hat{\gamma}_{i,1}^* = \frac{1\{D_{i,1}= d_1\}}{nP(D_{i,1} = d_1 | X_{i,1})}\Big /\underbrace{\Big(\frac{1}{n}\sum_{i=1}^n\frac{1\{D_{i,1} = d_1\}}{P(D_{i,1} = d_1 | X_{i,1})}\Big)}_{(D)}.  
\end{aligned} 
\end{equation} 

\paragraph{(D) is bounded away from zero} For the guess in Equation \eqref{eqn:guess} to be well-defined, we need that the denominator is bounded away from zero.
We now provide bounds on the denominator. Since $P(D_{i,1} = d_1 | X_{i,1}) \in (\delta, 1 -\delta)$ by Hoeffding inequality  
$$
\small 
\begin{aligned} 
P\Big( \Big|\frac{1}{n}\sum_{i=1}^n\frac{1\{D_{i,1} = d_1\}}{P(D_{i,1} = d_1 | X_{i,1})} - 1\Big| > t\Big) \le 2 \exp\Big(-\frac{n t^2}{2  a^2}\Big), 
\end{aligned} 
$$
for a finite constant $a$ that only depends on the overlap constant $\delta$. With probability at least $1 - 1/n$,
\begin{equation} \label{eqn:den1}
\small  
\begin{aligned} 
\frac{1}{n}\sum_{i=1}^n\frac{1\{D_{i,1} = d_1\}}{P(D_{i,1} = d_1 | X_{i,1})} > 1 - \sqrt{2a^2 \log(2n)/n} . 
\end{aligned}
\end{equation} 
Therefore for $n$ large enough such that $\sqrt{2a^2 \log(2n)/n} < 1 - \kappa$, taking some $\kappa \in (0,1)$, weights are finite with probability at least $1 - 1/n$. 

\paragraph{Weights sum up to one and satisfies $C_{n,1}$ constraint} The weights in Equation \eqref{eqn:guess} sum up to one and with probability at least $1 - 1/n$
$$
\frac{1\{D_{i,1}= d_1\}}{nP(D_{i,1} = d_1 | X_{i,1})} \le\frac{1}{n \delta}  \Rightarrow \gamma_{i,1}^* \le \frac{\bar{C}}{n \delta} 
$$
for a finite constant $\bar{C} < \infty$, 
where the first inequality follows by the overlap assumption that $P(D_{i,1} = d_1 | X_{i,1}) \in (\delta, 1 - \delta)$, and the second by Equation \eqref{eqn:den1}. 

\paragraph{First constraint in Algorithm 1} 
We are left to show that the first constraint in Algorithm 1 is satisfied. 

Under Assumption \ref{ass:seqignm}, 
$
\mathbb{E}\Big[\frac{1}{n} \sum_{i=1}^n \frac{1\{D_{i,1} = d_1\}X_{i,1}^{(j)}}{P(D_{i,1} = 1 | X_{i,1})} | X_1\Big] = \bar{X}_1^{(j)}.
$
In addition, since $X_{i,1}$ is subgaussian, and $1/P(D_{i,1} = d_1 | X_{i,1})$ is uniformly bounded,
$$
\small 
\begin{aligned} 
P \Big(\Big| \Big| \bar{X}_1 - \frac{1}{n} \sum_{i=1}^n \frac{1\{D_{i,1} = d_1\}}{P(D_{i,1} = 1 | X_{i,1})} X_{i,1} \Big| \Big|_{\infty} > t\Big) \le  p_1 2\exp\Big(-\frac{n t^2}{2 a^2}\Big)  
\end{aligned} 
$$
for a finite constant $a^2$. With trivial rearrangement, with probability $1- 1/n$, 
\begin{equation} \label{eqn:numerator} 
\Big| \Big| \bar{X}_1 - \frac{1}{n} \sum_{i=1}^n \frac{1\{D_{i,1} = d_1\}}{P(D_{i,1} = 1 | X_{i,1})} X_{i,1} \Big| \Big|_{\infty} \le  a  \sqrt{2\log(2 n p)/n}
\end{equation} 
Consider now the denominator $(D)$ in Equation \eqref{eqn:guess}. We have shown that with probability $1 - 1/n$, for a finite constant $a < \infty$, 
\begin{equation} \label{eqn:concetration1}
\Big| \frac{1}{n} \sum_{i=1}^n \frac{ 1\{D_{i,1} = d_1\}}{P(D_{i,1} = d_1 | X_{i,1})} - 1 \Big| \le 2 a\sqrt{\log(2n)/n}. 
\end{equation} 
Therefore, with probability $1 - 2/n$, 
\begin{equation} 
\small 
\begin{aligned} 
&\Big| \Big| \bar{X}_1 - \frac{\frac{1}{n} \sum_{i=1}^n \frac{1\{D_{i,1} = d_1\}}{P(D_{i,1} = 1 | X_{i,1})} X_{i,1}}{\frac{1}{n} \sum_{i=1}^n \frac{1\{D_{i,1} = d_1\}}{P(D_{i,1} = 1 | X_{i,1})}} \Big| \Big|_{\infty} 
 =  \Big| \Big| \frac{\bar{X}_1 \frac{1}{n} \sum_{i=1}^n \frac{1\{D_{i,1} = d_1\}}{P(D_{i,1} = 1 | X_{i,1})} - \frac{1}{n} \sum_{i=1}^n \frac{1\{D_{i,1} = d_1\}}{P(D_{i,1} = 1 | X_{i,1})} X_{i,1}}{\frac{1}{n} \sum_{i=1}^n \frac{1\{D_{i,1} = d_1\}}{P(D_{i,1} = 1 | X_{i,1})}} \Big| \Big|_{\infty} \\ 
 &\qquad  =  \Big| \Big| \frac{\bar{X}_1 \frac{1}{n} \sum_{i=1}^n \frac{1\{D_{i,1} = d_1\}}{P(D_{i,1} = 1 | X_{i,1})} + \bar{X}_1 - \bar{X}_1 - \frac{1}{n} \sum_{i=1}^n \frac{1\{D_{i,1} = d_1\}}{P(D_{i,1} = 1 | X_{i,1})} X_{i,1}}{\frac{1}{n} \sum_{i=1}^n \frac{1\{D_{i,1} = d_1\}}{P(D_{i,1} = 1 | X_{i,1})}} \Big| \Big|_{\infty}    \\ 
&\qquad \le \Big| \Big| \frac{\bar{X}_1  - \frac{1}{n} \sum_{i=1}^n \frac{1\{D_{i,1} = d_1\}}{P(D_{i,1} = 1 | X_{i,1})} X_{i,1}}{\frac{1}{n} \sum_{i=1}^n \frac{1\{D_{i,1} = d_1\}}{P(D_{i,1} = 1 | X_{i,1})}} \Big| \Big|_{\infty} + \frac{2a \sqrt{\log(2n)/n}}{\frac{1}{n} \sum_{i=1}^n \frac{1\{D_{i,1} = d_1\}}{P(D_{i,1} = 1 | X_{i,1})}} \\ 
&\qquad  \le \frac{a \sqrt{2\log(2 n p)/n}  +   2a \sqrt{\log(2n)/n}}{\frac{1}{n} \sum_{i=1}^n \frac{1\{D_{i,1} = d_1\}}{P(D_{i,1} = 1 | X_{i,1})}},  
\end{aligned} 
\end{equation} 
where the first inequality follows by the triangular inequality and by concentration of the term $\frac{1}{n} \sum_{i=1}^n \frac{1\{D_{i,1} = d_1\}}{P(D_{i,1} = 1 | X_{i,1})}$ around one at exponential rate as in Equation \eqref{eqn:concetration1}. The second inequality follows by concentration of the numerator as in Equation \eqref{eqn:numerator}. 
With probability $1 -1/n$, the denominator $(D)$ is bounded away from zero. Therefore for a universal constant $C < \infty$, (here $3/n$ follows from the union bound)
\begin{equation} 
\small 
P\Big(\Big| \Big| \bar{X}_1 - \frac{\frac{1}{n} \sum_{i=1}^n \frac{1\{D_{i,1} = d_1\}}{P(D_{i,1} = 1 | X_{i,1})} X_{i,1}}{\frac{1}{n} \sum_{i=1}^n \frac{1\{D_{i,1} = d_1\}}{P(D_{i,1} = 1 | X_{i,1})}} \Big| \Big|_{\infty} \le C a \sqrt{2\log(2 n p)/n}\Big) \ge 1  - 3/n.  
\end{equation} 

\paragraph{Bound on $||\hat{\gamma}||_2^2$}
We are left to provide bounds on $||\hat{\gamma}_1||_2^2$. For $n$ large enough, with probability at least $1 - 5/n$, $||\hat{\gamma}_1||_2^2 \le ||\hat{\gamma}_1^*||_2^2$ since $\hat{\gamma}_1^*$ is a feasible solution. By overlap, the fourth moment of $1/P(D_{i,1} = d_1 | X_{i,1})$ is bounded. By the strong law of large numbers and Slutsky theorem, 
\begin{equation} 
\small 
\begin{aligned} 
n ||\hat{\gamma}_1^*||_2^2 &= \sum_{i=1}^n \frac{1\{D_{i,1} = d_1\}}{n P(D_{i,1} = d_1 | X_{i,1})^2} \Big/\Big(\sum_{i=1}^n \frac{1\{D_{i,1} = d_1\}}{n P(D_{i,1} = d_1 | X_{i,1})}\Big)^2  \rightarrow_{as}  \frac{\mathbb{E}[\frac{1\{D_{i,1} = d_1\}}{ P(D_{i,1} = d_1 | X_{i,1})^2}]}{\mathbb{E}[\frac{1\{D_{i,1} = d_1\}}{ P(D_{i,1} = d_1 | X_{i,1})}]^2} < \infty.
\end{aligned} 
\end{equation} 
 which completes the proof. 
\end{proof} 

\begin{lem}[Existence of a feasible $\hat{\gamma}_t$]  \label{lem:firststat}
Let 
$$
\small 
\begin{aligned} 
Z_{i,t}(d_t) = \frac{1\{D_{i,t} = d_t\}}{P(D_{i,t} = d_t | Y_{i,1},\cdots , Y_{i,t-1}, X_{i,1},  \cdots , X_{i,t-1}, D_{i,1}, \cdots , D_{i,t-1})}.
\end{aligned} 
$$ Let Assumption \ref{ass:weakoverlap} hold and let for finite constants $c_0, \bar{c}$,  
$$\delta_t(n,p_t) \ge c_0 \frac{\log^{3/2}(p_tn)}{n^{1/2}}.$$
Then for any $t \le T$, with probability $\eta_n \rightarrow 1$, for some $N >0$, $n \ge N$, there exists a feasible $\hat{\gamma}_t^*(\hat{\gamma}_{t-1})$ satisfying the constraints in Algorithm 1 (as a function of the solution $\hat{\gamma}_{t-1}$ in the previous period), with 
$$
\small 
\begin{aligned} 
\hat{\gamma}_{i,t}^* = \hat{\gamma}_{i,t-1} Z_{i,t}(d_t)\Big/ \sum_{i=1}^n \hat{\gamma}_{i,t-1} Z_{i,t}(d_t) 
\end{aligned} 
$$ In addition, 
\begin{equation} \label{eqn:asymptotic_1}
\small 
\begin{aligned} 
\lim_{n \rightarrow \infty} P \Big(n ||\hat{\gamma}_t||_2^2 \le \bar{C}_t\Big) = 1 
\end{aligned}  
\end{equation} 
for a constant $1 \le \bar{C}_t < \infty$ independent of $(p_t,  n)$.
\end{lem} 

\begin{proof} [Proof of Lemma  \ref{lem:firststat}]
The proof follows by induction. By Lemma \ref{lem:gamma_1_star} we know that there exist a feasible $\hat{\gamma}_1$, with $\lim_{n \rightarrow \infty} P(n ||\hat{\gamma}_1||_2^2 \le C_1) = 1$, for some finite $C_1 < \infty$. Suppose now that there exist feasible $\hat{\gamma}_{1}, ..., \hat{\gamma}_{t-1}$, such that 
\begin{equation} \label{eqn:induction}  
\lim_{n \rightarrow \infty} P(n ||\hat{\gamma}_s||_2^2 \le C_s) = 1
\end{equation} 
for some finite constant $C_s < \infty$ depends on $s$, and for all $s < t$. We want to show that the statement holds for $\hat{\gamma}_t$. We find $\gamma_{t}^*$ that satisfies the constraint, with  
\begin{equation} \label{eqn:guess2}
\hat{\gamma}_{i,t}^* = \hat{\gamma}_{i,t-1} \frac{1\{D_{i,t} = d_t\}}{P(D_{i,t} = d_t | H_{i,t})}\Big / \Big(\sum_{i=1}^n \hat{\gamma}_{i,t - 1} \frac{1\{D_{i,t} = d_t\}}{P(D_{i,t} = d_t | H_{i,t})}\Big).  
\end{equation} 
We break the proof into several steps. 

\paragraph{Finite and Bounded Weights}

To show that such weights are finite, with high probability, we need to impose bounds on the numerator and the denominator of the weights in Equation \eqref{eqn:guess2}. 
We want to bound  for some finite constant $\epsilon > 0$, 
$$
\small 
\begin{aligned} 
&P\Big(\Big\{\max_{i \in \{1, ..., n\}} \hat{\gamma}_{i,t-1} \frac{1\{D_{i,t} = d_t\}}{P(D_{i,t} = d_t | H_{i,t})}  > C_{n,t - 1}/\delta    \Big\} \bigcup  \Big\{\sum_{i=1}^n \hat{\gamma}_{i,t} \frac{1\{D_{i,t} = d_t\}}{P(D_{i,t} = d_t | H_{i,t})} < \epsilon \Big\}\Big) \\
&\le  \underbrace{P\Big(\max_{i \in \{1, ..., n\}} \hat{\gamma}_{i,t-1} \frac{1\{D_{i,t} = d_t\}}{P(D_{i,t} = d_t | H_{i,t})}  >  C_{n,t - 1}/\delta   \Big)}_{(i)} +   \underbrace{P\Big(\sum_{i=1}^n \hat{\gamma}_{i,t} \frac{1\{D_{i,t} = d_t\}}{P(D_{i,t} = d_t | H_{i,t})} < \epsilon \Big)}_{(ii)}.
\end{aligned}
$$

\paragraph{Bound on $(i)$} We start by $(i)$. Observe first that we can bound
$$
\max_{i \in \{1, ..., n\}} \hat{\gamma}_{i,t-1} \frac{1\{D_{i,t} = d_t\}}{P(D_{i,t} = d_t | H_{i,t})} \le C_{n,t-1} \max_{i \in \{1, ..., n\}} \frac{1\{D_{i,t} = d_t\}}{P(D_{i,t} = d_t | H_{i,t})} \le \frac{C_{n,t-1}}{\delta} 
$$
by strong overlap (Assumption \ref{ass:weakoverlap}). 
\paragraph{Bound on $(ii)$} We now provide bounds on $(ii)$. 
Since $\sigma(H_{t-1}) \subseteq \sigma(H_t)$  
$$ 
\small 
\begin{aligned}
\mathbb{E}\Big[\sum_{i=1}^n \hat{\gamma}_{i,t-1} \frac{1\{D_{i,t} = d_t\}}{P(D_{i,t} = d_t | H_{i,t})} \Big] &= \mathbb{E}\Big[\mathbb{E}\Big[\sum_{i=1}^n \hat{\gamma}_{i,t-1} \frac{1\{D_{i,t} = d_t\}}{P(D_{i,t} = d_t | H_{i,t})} \Big| H_{t-1}\Big] \Big] 
 \\ &= \mathbb{E}\Big[\sum_{i=1}^n \hat{\gamma}_{i,t-1}  \mathbb{E}\Big[\mathbb{E}\Big[\frac{1\{D_{i,t} = d_t\}}{P(D_{i,t} = d_t | H_{i,t})} \Big| H_{t}\Big] \Big| H_{t-1}\Big] \Big]  = \sum_{i=1}^n \hat{\gamma}_{i,t-1} = 1.
\end{aligned} 
$$ 
Let $\bar{C}_{t-1}$ be the upper limit on $n||\hat{\gamma}_{t-1}||_2^2$, and let 
\begin{equation} \label{eqn:eta_n}
\begin{aligned} 
c  := 1\Big/\bar{C}_{t-1}  \quad 
 \eta_{n,t} := P(||\hat{\gamma}_{t-1}||_2^2 \le 1/(cn)), 
\end{aligned} 
\end{equation} 
for some constant $c$, which depends on $t-1$ (the dependence with $t-1$ is suppressed for expositional convenience). Observe in addition that $\eta_{n,t} \rightarrow 1$ by the induction argument (see Equation \eqref{eqn:induction}). 
We write for a finite constant $a$, for any $h > 0$
\begin{equation} \label{eqn:eta_n2}
\begin{aligned}
&P\Big( \Big |\sum_{i=1}^n \hat{\gamma}_{i,t-1} \frac{1\{D_{i,t} = d_t\}}{P(D_{i,t} = d_t | H_{i,t})} - 1 \Big| > h \Big) \\ &\le P\Big( \Big |\sum_{i=1}^n \hat{\gamma}_{i,t-1} \frac{1\{D_{i,t} = d_t\}}{P(D_{i,t} = d_t | H_{i,t})} - 1 \Big| > h \Big|  ||\hat{\gamma}_{t-1}||_2^2 \le 1/(cn) \Big) \eta_{n,t} + (1 - \eta_{n,t}) \\ &\le  2\exp\Big(-\frac{a h^2 }{2||\hat{\gamma}_{t-1}||_2^2 }\Big| ||\hat{\gamma}_{t-1}||_2^2 \le 1/(cn) \Big)\eta_{n,t} + (1 - \eta_{n,t}) \\ &\le 
2\exp\Big(-\frac{c h^2 a n}{2}\Big)\eta_{n,t} + (1 - \eta_{n,t}).
\end{aligned}
\end{equation} 
The third inequality follows from the fact that $\hat{\gamma}_{t-1}$ is measurable with respect to $H_{t-1}$ and $\frac{1\{D_{i,t} = d_t\}}{P(D_{i,t} = d_t | H_{i,t})}$ is sub-gaussian conditional on $H_{i,t-1}$ (since uniformly bounded).  
Therefore with probability at least $1 - \kappa$, 
\begin{equation} \label{eqn:denominator} 
\begin{aligned}
&\Big|\sum_{i=1}^n \hat{\gamma}_{i,t-1} \frac{1\{D_{i,t} = d_t\}}{P(D_{i,t} = d_t | H_{i,t})} - 1\Big|   \le \sqrt{2 \log(2\eta_{n,t}/(\kappa + \eta_{n,t} - 1))/(acn)}.
\end{aligned} 
\end{equation} 

By setting $\kappa = \eta_{n,t}/n + (1 - \eta_{n,t})$, with probability at least $1 - \eta_{n,t}/n + (1 - \eta_{n,t})$, 
\begin{align*} 
&\Big|\sum_{i=1}^n \hat{\gamma}_{i,t-1} \frac{1\{D_{i,t} = d_t\}}{P(D_{i,t} = d_t | H_{i,t})} - 1\Big|  \le   \sqrt{2 \log(2n)/acn}, 
\end{align*}

and hence the denominator is bounded away from zero for $n$ large enough (recall that $\eta_{n,t} \rightarrow 1$ by induction).

\paragraph{First Constraint in Algorithm 1} 
We now show that the proposed weights in Equation \eqref{eqn:guess2} satisfy the first constraint in Algorithm 1. The second constraint trivially holds, while the third follows from the ``finite and bounded weights" argument discussed in the paragraph above. We write 
$$
\begin{aligned} 
&\mathbb{E}\Big[\sum_{i=1}^n \hat{\gamma}_{i,t-1} H_{i,t}^{(j)} - \sum_{i=1}^n \hat{\gamma}_{i,t-1} \frac{1\{D_{i,t} = d_t\}}{P(D_{i,t} = d_t | H_{i,t})} H_{i,t}^{(j)} \Big] \\ &= \mathbb{E}\Big[\mathbb{E}\Big[\sum_{i=1}^n \hat{\gamma}_{i,t-1} H_{i,t}^{(j)} - \sum_{i=1}^n \hat{\gamma}_{i,t-1} \frac{1\{D_{i,t} = d_t\}}{P(D_{i,t} = d_t | H_{i,t})} H_{i,t}^{(j)} \Big | H_t \Big] \Big]  = 0. 
\end{aligned} 
$$
We want to show concentration.  We write for any $h >0$, 
$$
\begin{aligned} 
&P\Big(\Big |\Big| \sum_{i=1}^n \hat{\gamma}_{i,t-1} H_{i,t} - \sum_{i=1}^n \hat{\gamma}_{i,t-1} \frac{1\{D_{i,t} = d_t\}}{P(D_{i,t} = d_t | H_{i,t})} H_{i,t}\Big|\Big|_{\infty} > h\Big)  \\ 
&\le \underbrace{P\Big(\Big |\Big| \sum_{i=1}^n \hat{\gamma}_{i,t-1} H_{i,t} - \sum_{i=1}^n \hat{\gamma}_{i,t-1} \frac{1\{D_{i,t} = d_t\}}{P(D_{i,t} = d_t | H_{i,t})} H_{i,t}\Big|\Big|_{\infty} > h\Big| ||\hat{\gamma}_{t-1}||_2^2 \le 1/cn\Big)\eta_{n,t}}_{(I)} + \underbrace{(1 - \eta_{n,t})}_{(II)}, 
\end{aligned}  
$$
where $\eta_{n,t} = P(||\hat{\gamma}_{t-1}||_2^2 \le 1/cn)$ for some constant $c$ (that depends on $t - 1$). We study $(I)$, whereas, by the induction argument $(II) \rightarrow 0$ (Equation \eqref{eqn:induction}). 

\paragraph{Bound on $(I)$} For a constant $\bar{c} < \infty$, sub-gaussianity of $H_{i,t} | H_{t-1}$ and overlap, we can write for any $\lambda, h >0$, 
\begin{equation} \label{eqn:exp} 
\begin{aligned} 
&(I) \le   \sum_{j=1}^p \mathbb{E}\Big[\mathbb{E} \Big[\exp\Big(\lambda \bar{c} ||\hat{\gamma}_{t-1}||_2^2 - \lambda h\Big) | H_{t-1}, ||\hat{\gamma}_{t-1}||_2^2 \le 1/cn\Big] \Big| ||\hat{\gamma}_{t-1}||_2^2 \le 1/cn\Big]\eta_{n, t}.  
\end{aligned}  
\end{equation} 
Since $\hat{\gamma}_{t-1}$ is measurable with respect to $H_{t-1}$, we can write 
\begin{equation} \label{eqn:above}
\begin{aligned} 
\eqref{eqn:exp} \le \eta_{n,t} p_t \exp\Big(\lambda^2/(cn)- \lambda h\Big). 
\end{aligned} 
\end{equation}  
Choosing $\lambda = h cn/2$ we obtain that the above equation converges to zero as $\log(p_t)/n = o(1)$. After trivial rearrangement, with probability at least $1 - (1 - \eta_{n,t}) - 1/n$ (recall that $\eta_{n,t} \rightarrow 1$ by induction) , 
\begin{equation} \label{eqn:boundsup}
 \Big |\Big| \sum_{i=1}^n \hat{\gamma}_{i,t-1} H_{i,t} - \sum_{i=1}^n \hat{\gamma}_{i,t-1} \frac{1\{D_{i,t} = d_t\}}{P(D_{i,t} = d_t | H_{i,t})} H_{i,t}\Big|\Big|_{\infty} \lesssim \sqrt{\log(n p_t)/n}.
\end{equation} 
As a result, we can write 
$$
 \begin{aligned} 
  &\Big |\Big| \sum_{i=1}^n \hat{\gamma}_{i,t-1} H_{i,t} - \frac{\sum_{i=1}^n \hat{\gamma}_{i,t-1} \frac{1\{D_{i,t} = d_t\}}{P(D_{i,t} = d_t | H_{i,t})} H_{i,t}}{\sum_{i=1}^n \hat{\gamma}_{i,t-1} \frac{1\{D_{i,t} = d_t\}}{P(D_{i,t} = d_t | H_{i,t})}}\Big|\Big|_{\infty} \\ &= 
\Big |\Big| \frac{\sum_{i=1}^n \hat{\gamma}_{i,t-1} H_{i,t} \sum_{i=1}^n \hat{\gamma}_{i,t-1} \frac{1\{D_{i,t} = d_t\}}{P(D_{i,t} = d_t | H_{i,t})} - \sum_{i=1}^n \hat{\gamma}_{i,t-1} \frac{1\{D_{i,t} = d_t\}}{P(D_{i,t} = d_t | H_{i,t})} H_{i,t}}{\sum_{i=1}^n \hat{\gamma}_{i,t-1} \frac{1\{D_{i,t} = d_t\}}{P(D_{i,t} = d_t | H_{i,t})}}\Big|\Big|_{\infty} \\
&\lesssim \underbrace{\Big| \Big| \frac{\sum_{i=1}^n \hat{\gamma}_{i,t-1} H_{i,t}\Big(1 - \sum_{i=1}^n \hat{\gamma}_{i,t-1} \frac{1\{D_{i,t} =d_t\} }{P(D_{i,t} = d_t | H_{i,t})}\Big) }{\sum_{i=1}^n \hat{\gamma}_{i,t-1} \frac{1\{D_{i,t} = d_t\}}{P(D_{i,t} = d_t| H_{i,t})}} \Big| \Big|_{\infty} }_{(l)} + \underbrace{\Big| \Big| \frac{\sum_{i=1}^n \hat{\gamma}_{i,t-1} H_{i,t} \Big(1 - \frac{1\{D_{i,t} = d_t\}}{P(D_{i,t} = d_t | H_{i,t})}\Big)}{\sum_{i=1}^n \hat{\gamma}_{i,t-1} \frac{1\{D_{i,t} = d_t\}}{P(D_{i,t} = d_t| H_{i,t})}}\Big|\Big|_{\infty}}_{(ll)}. 
  \end{aligned} 
$$
Observe now that the denominators of the above expressions are bounded away from zero with high probability as discussed in Equation \eqref{eqn:denominator}. The numerator of $(ll)$ is bounded by Equation \eqref{eqn:boundsup}. We are left with the numerator of $(l)$. 
  Note first that 
$$
\mathbb{E}\Big[\sum_{i=1}^n \hat{\gamma}_{i,t-1} \frac{1\{D_{i,t} =d_t\} }{P(D_{i,t} = d_t | H_{i,t})} \Big| H_{i,t} \Big] = 1.  
 $$  
 We can write 
 $$
 \small 
 \begin{aligned} 
&  \Big| \Big|\sum_{i=1}^n \hat{\gamma}_{i,t-1} H_{i,t}\Big(1 - \sum_{i=1}^n \hat{\gamma}_{i,t-1} \frac{1\{D_{i,t} =d_t\} }{P(D_{i,t} = d_t | H_{i,t})}\Big)\Big|\Big|_{\infty}  \le \underbrace{\max_k \Big|\sum_{i=1}^n \hat{\gamma}_{i,t-1} H_{i,t}^{(k)}\Big|}_{(j)}  \underbrace{\Big| 1 - \sum_{i=1}^n \hat{\gamma}_{i,t-1} \frac{1\{D_{i,t} =d_t\} }{P(D_{i,t} = d_t | H_{i,t})}\Big|}_{(jj)}. 
 \end{aligned} 
 $$ 
 Here $(jj)$ is bounded as in Equation  \eqref{eqn:denominator}, with probability $1 - 1/n$ at a rate $\sqrt{\log(n)/n}$. The component $(j)$ instead is bounded as 
 $$
 (j) \le \max_{k, i} |H_{i,t}^{(k)}| \lesssim \log(p_t n) 
 $$ 
 with probability $1 - 1/n$ using subgaussianity of $H_{i,t}^{(k)}$ \citep{wainwright2019high}. Therefore, all constraints in Algorithm 1 are satisfied with probability converging to one.

 \paragraph{Finite Norm}

 We now need to show that Equation \eqref{eqn:asymptotic_1} holds.   With probability converging to one, 
 \begin{equation} \label{eqn:finite_norm} 
 \begin{aligned}
 n||\hat{\gamma}_t||_2^2  & \le n||\hat{\gamma}_t^*(\hat{\gamma}_{t-1})||_2^2 
= \sum_{i=1}^n n\hat{\gamma}_{i,t-1}^{2} \frac{1\{D_{i,t} = d_t\}}{P(D_{i,t} = d_t | H_{i,t})^2}\Big/ \Big(\sum_{i=1}^n \hat{\gamma}_{i,t-1} \frac{1\{D_{i,t} = d_t\}}{P(D_{i,t} = d_t | H_{i,t})}\Big)^2.
 \end{aligned}  
 \end{equation} 
 The denominator converges in probability to one by Equation \eqref{eqn:denominator}. The numerator can instead be bounded by $n ||\hat{\gamma}_{t-1}||^2$ up-to a finite multiplicative constant by Assumption \ref{ass:weakoverlap}. By the recursive argument $n||\hat{\gamma}_t||^2 =  O_p(1)$. 
\end{proof}

\begin{lem} \label{lem:bound_n} The weights solving the optimization problem in Algorithm 1 are such that 
$
||\hat{\gamma}_t||_2^2 \ge 1/n. 
$
\end{lem} 
\begin{proof} Observe that for either algorithms, weights sum to one. The minimum under this constraint only is obtained at $\hat{\gamma}_{i,t} = 1/n$ for all $i$ concluding the proof. 
\end{proof} 

\begin{lem}[Sub-gaussianity] \label{lem:sub-gaussian} Suppose that $Y_{i,T}$ is a sub-gaussian random variable. Then $\varepsilon_{i,T}, \nu_{i,t}, t\in \{0, \cdots, T\}$ for finite $T$ are also sub-gaussian random variables. 
\end{lem} 

\begin{proof}[Proof of Lemma \ref{lem:sub-gaussian}] First, note that for generic random variables $Z, X$, $\mathbb{E}[Z | X]$ is sub-gaussian if $Z$ is sub-gaussian. The reason is because $\mathbb{E}[Z | X]$ is a contraction in $L_p$ spaces. Because subgaussianity is satisfied if $E[|Z|^p] < K^P p^{p/2}$ for a constant $K$, it follows that $\mathbb{E}[Z | X]$ is subgaussian as $\mathbb{E}[|\mathbb{E}[Z | X]|^p \le \mathbb{E}[|Z|^p]$ for any $p \ge 1$. In addition, for two sub-gaussian random variables $X_1, X_2$, $X_1 + X_2$ is also sub-gaussian. 
To show this we can use the definition of sub-gaussianity using the moment generating function. In particular, for any $\lambda > 0$, we have $\mathbb{E}[e^{\lambda [(X_1 + X_2) - \mathbb{E}[X_1 + X_2] ]}] \le \sqrt{\mathbb{E}[e^{2 \lambda (X_1 - \mathbb{E}[X_1])}]} \sqrt{\mathbb{E}[e^{2 \lambda (X_2 - \mathbb{E}[X_2])}]}$. The result directly follows from the definition of sub-gaussianity using the moment generating function \citep{wainwright2019high}. 
Lemma \ref{lem:sub-gaussian} directly follows from these two properties as it is simple to show that $\varepsilon_{i,T}, \nu_{i,t}$ are defined as sums and differences of sub-gaussian random variables.   
\end{proof}

\section{Proofs of the Main Theorems} \label{sec:main_theorem}

\subsection*{Proof of Theorem \ref{thm:residual_final}} 

The first part of the theorem is a direct corollary of Lemma \ref{lem:lemmam1} and the second part of the theorem is a direct corollary of Lemma \ref{lem:lemmam2} for multiple periods, whose proofs are contained in Appendices \ref{app:lem22} and \ref{lem:33} respectively.

\subsection*{Proof of Theorem \ref{thm:overlap} and Corollary \ref{cor:comparison}} \label{app:overlap_proof}

  Theorem \ref{thm:overlap} is a direct corollary of 
Lemmas \ref{lem:gamma_1_star} and \ref{lem:firststat}.

Corollary \ref{cor:comparison} follows directly from the fact that we can find a feasible solution $\gamma_t^*(\hat{\gamma}_{t-1})$ as a function of the \textit{solution} $\hat{\gamma}_{t-1}$ in the previous step. In addition,  
 
 \begin{equation} 
 \begin{aligned}
 n||\hat{\gamma}_t||_2^2  & \le n||\hat{\gamma}_t^*(\hat{\gamma}_{t-1})||_2^2 
= \sum_{i=1}^n n\hat{\gamma}_{i,t-1}^{2} \frac{1\{D_{i,t} = d_t\}}{P(D_{i,t} = d_t | H_{i,t})^2}\Big/ \Big(\sum_{i=1}^n \hat{\gamma}_{i,t-1} \frac{1\{D_{i,t} = d_t\}}{P(D_{i,t} = d_t | H_{i,t})}\Big)^2.
 \end{aligned}  
 \end{equation} 
 The denominator converges in probability to one by Equation \eqref{eqn:denominator} (following verbatim the paragraph ``Bound on $(ii)$ in the proof of Lemma \ref{lem:firststat}). The numerator can instead be bounded by $n ||\hat{\gamma}_{t-1}||^2$ up-to a finite multiplicative constant by Assumption \ref{ass:weakoverlap} by Holder's inequality. Taken together, with probability converging to one, $n||\hat{\gamma}_t||_2^2 \le n ||\hat{\gamma}_{t-1}||_2^2 c_0$ for a finite constant $c_0 < \infty$. 

\subsection*{Proof of Theorems \ref{thm:convergence_rate_first} and \ref{thm:thm_asym_t}}

We first prove Theorem \ref{thm:thm_asym_t} and then show that 
Theorem \ref{thm:convergence_rate_first} is a direct corollary of Theorem \ref{thm:thm_asym_t} at the end of the proof. 

\paragraph{Overview of the proof strategy} 
Throughout recall the definition of 
$$
H_t = (D_1, \cdots, D_{t-1}, X_1, \cdots, X_t, Y_1, \cdots, Y_{t-1}).
$$  We will consider filtrations generated by the sigma algebra $\sigma(H_t, D_t)$. 

We organize the proof of Theorem \ref{thm:thm_asym_t} as follows. As the first step we decompose the estimation error in the one where the variance is known and  the estimation error of the variance. For the first component, 
we then derive conditional Linderberg conditions for each random variables 
$$
\small 
\begin{aligned} 
&Z_{t} = \frac{\sum_{i= 1}^n \hat{\gamma}_{i,t} \nu_{i,t}}{\sqrt{\sum_{i= 1}^n \hat{\gamma}_{i,t}^2 \mbox{Var}(\nu_{i,t}|H_t, D_t)}} , \quad t \in \{1, \cdots, T -1\}, \\ &  Z_{T} = \frac{\sum_{i= 1}^n\hat{\gamma}_{i,T} \varepsilon_{i,T}}{\sqrt{\sum_{i= 1}^n \mbox{Var}(\varepsilon_{i,T} | H_{i,T}, D_{i,T}) \hat{\gamma}_{i,T}^2}}, \quad Z_0 = \frac{\bar{X}_1 \beta^1 - \mu(d_{1:T})}{\sqrt{\frac{1}{n} \mathrm{Var}(X_{i,1} \beta^1)}}, 
 \end{aligned} 
$$
for each $Z_t$ conditional on $\sigma(H_t,D_t)$, 
and show that each random variable $Z_t$ is mean zero conditional on $\sigma(H_t,D_t)$ and asymptotically Gaussian. We then leverage the fact that $Z_{t-1}$ is measurable with respect to the filtration $\sigma(H_t, D_t)$ to establish a joint Gaussian asymptotic result. 

Finally, we show that our estimators can be written asymptotically as a weighted average of $Z_t$ with weights $W_t$ as defined in Equation \eqref{eqn:weights_W}, where $W_t$ are measurable with respect to $\sigma(H_{t},D_{t})$, and with the property that $||W||^2 = 1$ almost surely. In this last step, we leverage a martingale property of the function $M_t(\lambda) = \exp(\sum_{s=0}^t W_s Z_s - \frac{\lambda^2}{2} \sum_{k=0}^t W_k^2)$, and use to show that the moment generating function of $\sum_t W_t Z_t$ corresponds to the one of a standard normal random variable. 

Finally, we complete the proof by controlling the estimation error for the variance. This latter step is where we use Assumption \ref{ass:4m}(i), where either  $\max_t \| \hat{\beta}_{d_{1:T}}^{(t)} - \beta_{d_{1:T}}^{(t)} \|_1 = O_p(n^{-1/4})$ and $H_t$ are sub-gaussian; or (b) $\max_t \| \hat{\beta}_{d_{1:T}}^{(t)} - \beta_{d_{1:T}}^{(t)} \|_1 = o_p(1/\log(n))$ and $||H_t||_{\infty} \le h$  We now proceed with the formal proof. 

\paragraph{Weights do not diverge to infinity}

First note that by Lemmas \ref{lem:gamma_1_star}, \ref{lem:firststat}, there exist a $\hat{\gamma}_t^*$ such that for $N$ large enough, with probability converging to one, for some $N > 0$, and $n > N$ 
\begin{equation} \label{eqn:convergence} 
\begin{aligned} 
n||\hat{\gamma}_{t}||_2^2 \le n||\hat{\gamma}_t^*||_2^2 = O_p(1).
\end{aligned} 
\end{equation}
 Similarly, $n \sum_{i=1}^n \gamma_{i,t}^2 \mbox{Var}(\nu_{i,t} | H_{i,t}, D_{i,t}) = O_p(1)$ and  $n\sum_{i=1}^n \gamma_{i,T}^2 \mbox{Var}(\varepsilon_{i,T}| H_T) = O_p(1)$ since the conditional variances are uniformly bounded by the finite fourth moment condition.

\paragraph{Error Decomposition} 

 We denote $\bar{\sigma}^2$ the lower bound on the conditional variances and $\sigma_{up}^2$ a the upper bound on the variances under Assumption \ref{ass:4m}.  Recall $\nu_{i,t}  =H_{i,t+1}\beta^{t+1}_{d_{1;T}} - H_{i,t} \beta^t_{d_{1:T}}, \nu_{i, 0} = X_{i,1} \beta^1 - \mathbb{E}[X_{i,1}] \beta^1$ and $\hat{\nu}_{i,t}$  for estimated coefficients, $\hat \nu_{i,t}  =H_{i,t+1}\hat \beta^{t+1}_{d_{1;T}} - H_{i,t} \hat \beta^t_{d_{1:T}}, \hat{\nu}_{i, 0} = X_{i,1} \hat{\beta}^1 - \bar{X}_1 \hat{\beta}^1$. 
We write 
\begin{equation} 
\small 
 \begin{aligned} & \frac{\hat{\mu}(d_{1:T}) - \mu(d_{1:T})}{\sqrt{\hat{V}_T(d_{1:T})}}  \\ & = 
\underbrace{\frac{\hat{\mu}(d_{1:T}) - \mu(d_{1:T})}{\sqrt{\sum_{i= 1}^n \hat{\gamma}_{i,T}^2 \mbox{Var}(\varepsilon_{i, T} | H_{i,T}, D_{i,T}) + \sum_{i= 1}^n \sum_{t=1}^{T-1} \hat{\gamma}_{i,t}^2 \mbox{Var}(\nu_{i,t} | H_{i,t}, D_{i,t}) +  \frac{1}{n^2} \sum_{i=1}^n \mathrm{Var}(X_{i,1} \beta^1) } }}_{(I)} \times \\ &\times \underbrace{ \frac{\sqrt{\sum_{i= 1}^n \hat{\gamma}_{i,T}^2 \mbox{Var}(\varepsilon_{i,T} | H_{i,T}, D_{i,T}) + \sum_{i= 1}^n \sum_{t=0}^{T-1} \hat{\gamma}_{i,t}^2 \mbox{Var}(\nu_{i,t} | H_{i,t}, D_{i,t}) + \frac{1}{n} \mathrm{Var}(X_{i,1} \beta^1) }}{\sqrt{\sum_{i= 1}^n \Big\{\hat{\gamma}_{i,T}^2 (Y_{i,T} - H_{i,T} \hat{\beta}_{d_{1:T}}^T)^2 +  \sum_{t=0}^{T-1}  \hat{\gamma}_{i,t}^2 \hat{\nu}_{i,t}^2 + \frac{1}{n^2} \hat{\nu}_{i,0}^2 \Big\} }}}_{(II)}. 
\end{aligned} 
\end{equation} 
\paragraph{Term $(I)$} 
We consider the term $(I)$. By Lemma \ref{lem:lemmam1}, we have 
$$
\small 
\begin{aligned}  
& (I) = \underbrace{\frac{\sum_{t=1}^T (\beta^t - \hat{\beta}^t)^\top(\hat{\gamma}_t H_t - \hat{\gamma}_{t-1} H_t) + (\beta^1 - \hat{\beta}^1)(\hat{\gamma}_1 X_1 - \bar{X}_1)}{
\sqrt{\sum_{i= 1}^n \hat{\gamma}_{i,T}^2 \mbox{Var}(\varepsilon_{i,T} | H_{i,T}, D_{i,T}) + \sum_{i= 1}^n \sum_{t=1}^{T-1} \hat{\gamma}_{i,t}^2 \mbox{Var}(\nu_{i,t} | H_{i,t}, D_{i,t}) + \frac{1}{n} \mathrm{Var}(X_{i,1} \beta^1) }}}_{(j)} \\ &+ \underbrace{\frac{\sum_{i= 1}^n \hat{\gamma}_{i,T} \varepsilon_{i, T} + \sum_{t=1}^{T-1} \hat{\gamma}_{i,t} \nu_{i,t} + (\bar{X}_1\beta^1 - \mu(d_{1:T})) }{
\sqrt{\sum_{i= 1}^n \hat{\gamma}_{i,T}^2 \mbox{Var}(\varepsilon_{i,T} | H_{i,T}, D_{i,T}) + \sum_{i= 1}^n \sum_{t=1}^{T-1} \hat{\gamma}_{i,t}^2 \mbox{Var}(\nu_{i,t} | H_{i,t}, D_{i,t}) + \frac{1}{n} \mathrm{Var}(X_{i,1} \beta^1) }}}_{(jj)} .
\end{aligned} 
$$
We start from $(j)$. Since $\sum_{i=1}^n \hat{\gamma}_{i,t} = 1$ and the variances are bounded from below (see Lemma \ref{lem:bound_n}), it follows that 
$$
\sum_{i=1}^n \hat{\gamma}_{i,T}^2 \mbox{Var}(\varepsilon_{i,T} | H_{i,T}, D_{i,T}) + \sum_{i=1}^n \sum_{t=1}^{T-1} \hat{\gamma}_{i,t}^2 \mbox{Var}(\nu_{i,t} | H_{i,t}, D_{i,t}) + \frac{1}{n} \mathrm{Var}(X_{i,1} \beta^1) \ge  T \bar{\sigma}^2 \sum_{i=1}^n \frac{1}{n^2} = T \bar{\sigma}^2/n. 
$$
Therefore, since the denominator is bounded from below by $\bar{\sigma} \sqrt{T/n}$, and since, by Holder's inequality 
\begin{equation} \label{eqn:eqn_j} 
\sum_{t=1}^T (\beta^t - \hat{\beta}^t)^\top(\hat{\gamma}_t H_t - \hat{\gamma}_{t-1} H_t) + (\beta^1 - \hat{\beta}^1)^\top (\hat{\gamma}_1 X_1 - \bar{X}_1) \lesssim T \max_t \delta_t(n,p) ||\beta^t - \hat{\beta}^t||_1  = o_p(n^{-1/2})
\end{equation} 
under Assumption \ref{ass:4m} and the fact that $T$ is fixed. We can now write 
$$
\scriptsize 
\begin{aligned}
  & (I) = 
 \underbrace{\frac{\sum_{i = 1}^n \hat{\gamma}_{i,T} \varepsilon_{i,T}}{\sqrt{\sum_{i=1}^n \hat{\gamma}_{i,T}^2 \mbox{Var}(\varepsilon_{i,T} | H_{i,T}, D_{i,T})}}}_{(i)} \times \underbrace{\sqrt{\frac{\sum_{i=1}^n \hat{\gamma}_{i,T}^2 \mbox{Var}(\varepsilon_{i,T} | H_{i,T}, D_{i,T})}{\sum_{i=1}^n \hat{\gamma}_{i,T}^2 \mbox{Var}(\varepsilon_{i,T} | H_{i,T}, D_{i,T}) + \sum_{i=1}^n \sum_{t=1}^{T-1} \hat{\gamma}_{i,t}^2 \mbox{Var}(\nu_{i,t} | H_{i,t}, D_{i,t}) + \frac{1}{n} \mathrm{Var}(X_{i,1} \beta^1)}}}_{(ii)} \\ &+ \sum_{t=1}^{T-1} \underbrace{\frac{\sum_{i = 1}^n \hat{\gamma}_{i,t} \nu_{i,t}}{ \sqrt{\sum_i \mbox{Var}(\nu_{i,t} | H_{i,t}, D_{i,t} ) \hat{\gamma}_{i,t}^2}} }_{(iii)}  \times \underbrace{\frac{\sqrt{\sum_i \mbox{Var}(\nu_{i,t} | H_{i,t}, D_{i,t}) \hat{\gamma}_{i,t}^2}}{\sqrt{\sum_{i=1}^n \hat{\gamma}_{i,T}^2 \mbox{Var}(\varepsilon_{i,T} | H_{i,T}, D_{i,T}) + \sum_{i=1}^n \sum_{t=1}^{T-1} \hat{\gamma}_{i,t}^2 \mbox{Var}(\nu_{i,t} | H_{i,t}, D_{i,t}) + \frac{1}{n} \mathrm{Var}(X_{i,1} \beta^1)}}}_{(iv)} \\ 
 & + \underbrace{\frac{\bar{X}_1 \beta^1 - \mu(d_{1:T})}{\sqrt{\frac{1}{n} \mathrm{Var}(X_{i,1} \beta^1)}}}_{(v)} \times \underbrace{\frac{\sqrt{\frac{1}{n} \mathrm{Var}(X_{i,1} \beta^1)}}{\sqrt{\sum_{i=1}^n \hat{\gamma}_{i,T}^2 \mbox{Var}(\varepsilon_{i,T} | H_{i,T}, D_{i,T}) + \sum_{i=1}^n \sum_{t=1}^{T-1} \hat{\gamma}_{i,t}^2 \mbox{Var}(\nu_{i,t} | H_{i,t}, D_{i,t}) + \frac{1}{n} \mathrm{Var}(X_{i,1} \beta^1)}}}_{(vi)}
 + o_p(1).  
\end{aligned} 
$$  
First, notice that $\sigma(\hat{\gamma}_T) \subseteq \sigma(D_T, H_T)$, and by Assumption \ref{ass:seqignm},  $\mathbb{E}[\varepsilon_T| H_T, D_T] = 0$. Therefore, 
$$
 \mathbb{E}[\hat{\gamma}_{i,T} \varepsilon_{i, T} | H_T, D_T] =  0, \quad  \bar{\sigma}^2 ||\hat{\gamma}_T||_2^2 \le \mbox{Var}\Big(\sum_{i=1}^n \hat{\gamma}_{i,T} \varepsilon_{i,T} |H_T, D_T\Big) \le ||\hat{\gamma}_T||_2^2 \sigma_{\varepsilon}^2, 
$$
where the first statement follows directly from Lemma \ref{lem:lemmam2} and the second statement holds 
for a finite constant $\sigma_{\varepsilon}^2$ by the third moment condition in Assumption \ref{ass:4m}.
By the fourth moment conditions in Assumption \ref{ass:4m}, for a constant $0 < C < \infty$,
$$
 \begin{aligned}  
 \mathbb{E}\Big[\Big(\sum_{i= 1}^n \hat{\gamma}_{i,T} \varepsilon_{i,T}\Big)^3  \Big | H_T , D_T\Big] &=  \sum_{i= 1}^n \hat{\gamma}_{i,T}^3 \mathbb{E}[\varepsilon_{i,T}^3 |H_T, D_T] \\ &\le C \sum_{i= 1}^n \hat{\gamma}_{i,T}^3 \le C ||\hat{\gamma}_T||_2^2 \max_i |\hat{\gamma}_{i,T}|  \lesssim  \log(n) n^{-2/3} ||\hat{\gamma}_T||_2^2 .
 \end{aligned}  
$$
 Thus, 
$$
\mathbb{E}\Big[\sum_{i= 1}^n \hat{\gamma}_{i,T}^3 \varepsilon_{i,T}^3 \Big | H_T, D_T\Big ] \Big / \mbox{Var}\Big(\sum_{i= 1}^n \hat{\gamma}_{i,T} \varepsilon_{i,T} \Big | H_T, D_T \Big)^{3/2} = O(\log(n) n^{-2/3}  ||\hat{\gamma}_T||_2^{-1}) = o(1).
$$
 By Lyapunov  theorem, we have 
$$
 \frac{\sum_{i= 1}^n \hat{\gamma}_{i,T} \varepsilon_{i,T} }{\sqrt{\sum_{i= 1}^n \hat{\gamma}_{i,T} \mbox{Var}(\varepsilon_{i,T} | H_T, D_T)}} \Big|   \sigma(H_T, D_T)\rightarrow_d \mathcal{N}(0,\sigma^2). 
$$
Consider now $(iii)$ for a generic time $t$. 
We study the behaviour of $\sum_{i= 1}^n \hat{\gamma}_{i,t} \nu_{i, t}$ conditional on $\sigma(H_t, D_t)$. Since $\sigma(\hat{\gamma}_t) \subseteq \sigma( H_t, D_t)$, $\hat{\gamma}_t$ is deterministic given $\sigma( H_t, D_t)$. By Lemma \ref{lem:lemmam2}, 
$
\mathbb{E}[\hat{\gamma}_{i,t} \nu_{i,t} | H_t, D_t] =0. 
$
We now study the second moment. Notice that 
$$
 \begin{aligned} 
\bar{\sigma}^2 ||\hat{\gamma}_t||_2^2 \le  \mbox{Var}(\sum_{i= 1}^n \hat{\gamma}_{i,t} \nu_{i,t} \big| H_t, D_t) &= \sum_{i= 1}^n \hat{\gamma}_{i,t}^2 \mbox{Var}(\nu_{i,t} |H_t, D_t) \le \sum_{i= 1}^n \hat{\gamma}_{i,t}^2 \sigma_{ub}^2.
 \end{aligned} 
$$
Finally, we consider the third moment. Under Assumption \ref{ass:4m},
$$
\begin{aligned} 
&\mathbb{E}\Big[\sum_{i= 1}^n \hat{\gamma}_{i,t}^3 \nu_{i,t}^3 \Big | H_t, D_t\Big] = \sum_{i= 1}^n \hat{\gamma}_{i,t}^3 \mathbb{E}[\nu_{i,t}^3 | H_t, D_t]  \le \sum_{i= 1}^n \hat{\gamma}_{i,t}^3 u_{max}^3 \lesssim \log(n) n^{-2/3} ||\hat{\gamma}_t||_2^2  
.\end{aligned} 
$$
Since $||\hat{\gamma}_t||_2 \ge 1/\sqrt{n}$ by Lemma \ref{lem:bound_n} and since $\mbox{Var}(\nu_{i, t}|H_t, D_t) > u_{min}$, 
$$
\mathbb{E}\Big[\sum_{i= 1}^n \hat{\gamma}_{i,t}^3 \nu_{i,t}^3 \Big | H_t, D_t \Big] \Big / \mbox{Var}\Big(\sum_{i= 1}^n \hat{\gamma}_{i,t} \nu_{i,t} \Big | H_t, D_t \Big)^{3/2} = O(\log(n) n^{-2/3} ||\hat{\gamma}_t||_2^{-1} ) = o(1).
$$ 
$$
\Rightarrow \frac{\sum_{i= 1}^n \hat{\gamma}_{i,t} \nu_{i,t}}{\sqrt{\sum_{i= 1}^n \hat{\gamma}_{i,t}^2 \mbox{Var}(\nu_{i,t} | H_t, D_t)} } \Big| \sigma( H_t, D_t) \rightarrow_d \mathcal{N}(0, 1). 
$$
The same reasoning applies verbatim to $(v)$. Therefore, collecting our results
\begin{equation} \label{eqn:general} 
\small 
\begin{aligned} 
&\frac{\sum_{i= 1}^n\hat{\gamma}_{i,T} \varepsilon_{i,T}}{\sqrt{\sum_{i= 1}^n \mbox{Var}(\varepsilon_{i,T} | H_{i,T}, D_{i,T}) \hat{\gamma}_{i,T}^2}} \quad  \Big |\quad  \sigma(H_T, D_T) \quad \rightarrow_d \mathcal{N}(0,1) \\
& \frac{\sum_{i= 1}^n \hat{\gamma}_{i,t} \nu_{i,t}}{\sqrt{\sum_{i= 1}^n \hat{\gamma}_{i,t}^2 \mbox{Var}(\nu_{i,t}|H_t, D_t)}} \quad  \Big |\quad   \sigma(H_t, D_t) \quad \rightarrow_d \mathcal{N}(0,1), \quad \forall t \in \{1, ..., T-1\} \\ 
& \frac{\bar{X}_1 \beta^1 - \mu(d_{1:T})}{\sqrt{\frac{1}{n} \mathrm{Var}(X_{i,1} \beta^1)}} \rightarrow_d \mathcal{N}(0,1). 
\end{aligned} 
\end{equation} 
In addition, to complete the characterization of the joint distribution, observe that 
\begin{equation} 
\small 
\begin{aligned} 
&\mathbb{E}[\hat{\gamma}_{i,T} \varepsilon_{i, T} \hat{\gamma}_{i,t}\nu_{i,t} |H_T, D_T] =
 \hat{\gamma}_{i,t}\nu_{i,t} \hat{\gamma}_{i,T}  \mathbb{E}[\varepsilon_{i, T} |H_T, D_T] = 0 \\
 &\mathbb{E}[\hat{\gamma}_{i,s} \nu_{i,s} \hat{\gamma}_{i,t}\nu_{i,t} |H_{\max\{s,t\}}, D_{\max\{s,t\}}] =
 \hat{\gamma}_{i,t} \hat{\gamma}_{i,s} \nu_{i,\min\{t, s\}} \mathbb{E}[ \nu_{i,\max\{s, t\}}  |H_{\max\{s,t\}}, D_{\max\{s,t\}}] = 0. \end{aligned} 
\end{equation} 

\paragraph{Introducing the $Z_t$, $W_t$ notation} Since each component at time $t$ is measurable with respect to $\sigma(H_{t+1}, D_{t+1})$, it follows the joint convergence result 
$$
\small 
\begin{aligned} 
& \Big[Z_0, Z_{1},
   \cdots
    Z_{T}\Big]^\top 
 \rightarrow_d \mathcal{N}\left(0, I \right), \\ &Z_{t} = \frac{\sum_{i= 1}^n \hat{\gamma}_{i,t} \nu_{i,t}}{\sqrt{\sum_{i= 1}^n \hat{\gamma}_{i,t}^2 \mbox{Var}(\nu_{i,t}|H_t, D_t)}} , \quad t \in \{1, \cdots, T -1\}, \\ &  Z_{T} = \frac{\sum_{i= 1}^n\hat{\gamma}_{i,T} \varepsilon_{i,T}}{\sqrt{\sum_{i= 1}^n \mbox{Var}(\varepsilon_{i,T} | H_{i,T}, D_{i,T}) \hat{\gamma}_{i,T}^2}}, \quad Z_0 = \frac{\bar{X}_1 \beta^1 - \mu(d_{1:T})}{\sqrt{\frac{1}{n} \mathrm{Var}(X_{i,1} \beta^1)}}. 
 \end{aligned} 
$$
We are left to consider the components $(ii)$, $(iv)$, $(vi)$. Define 
\begin{equation} \label{eqn:weights_W} 
\small 
\begin{aligned} 
W_T & = \sqrt{\frac{\sum_{i=1}^n \hat{\gamma}_{i,T}^2 \mbox{Var}(\varepsilon_{i,T} | H_{i,T}, D_{i,T})}{\sum_{i=1}^n \hat{\gamma}_{i,T}^2 \mbox{Var}(\varepsilon_{i,T} | H_{i,T}, D_{i,T}) + \sum_{i=1}^n \sum_{t=1}^{T-1} \hat{\gamma}_{i,t}^2 \mbox{Var}(\nu_{i,t} | H_{i,t}, D_{i,t}) + \frac{1}{n} \mathrm{Var}(X_{i,1} \beta^1)}}, \\ W_t & = \frac{\sqrt{\sum_i \mbox{Var}(\nu_{i,t} | H_{i,t}, D_{i,t}) \hat{\gamma}_{i,t}^2}}{\sqrt{\sum_{i=1}^n \hat{\gamma}_{i,T}^2 \mbox{Var}(\varepsilon_{i,T} | H_{i,T}, D_{i,T}) + \sum_{i=1}^n \sum_{t=1}^{T-1} \hat{\gamma}_{i,t}^2 \mbox{Var}(\nu_{i,t} | H_{i,t}, D_{i,t}) + \frac{1}{n} \mathrm{Var}(X_{i,1} \beta^1)}}, \quad t \in \{1, \cdots, T - 1\} \\ 
W_0 & = \frac{\sqrt{\frac{1}{n} \mathrm{Var}(X_{i,1} \beta^1)}}{\sqrt{\sum_{i=1}^n \hat{\gamma}_{i,T}^2 \mbox{Var}(\varepsilon_{i,T} | H_{i,T}, D_{i,T}) + \sum_{i=1}^n \sum_{t=1}^{T-1} \hat{\gamma}_{i,t}^2 \mbox{Var}(\nu_{i,t} | H_{i,t}, D_{i,t}) + \frac{1}{n} \mathrm{Var}(X_{i,1} \beta^1)}}
\end{aligned} 
\end{equation}

\paragraph{Final asymptotic normality step} Note that $||W||_2^2 = 1$ almost surely. In addition because $W_t, Z_{t-1}$ are measurable with respect to $\sigma(H_t, D_t)$ it follows from Equation \eqref{eqn:general} that for all $t \le T$, $Z_t | Z_{t-1}, \cdots, Z_0, W_t, \cdots, W_{0} \rightarrow_d \mathcal{N}(0,1)$. 

To show that $\sum_{t=1}^T W_t Z_t$ is asymptotically normal, we will study its (asymptotic) moment generating function, for   $Z_t | Z_{t-1}, \cdots, Z_0, W_t, \cdots, W_{0} \sim \mathcal{N}(0,1)$. Define 
$$
M_t(\lambda):= \exp\Big(\lambda \sum_{s=0}^t W_s Z_s - \frac{\lambda^2}{2} \sum_{k=0}^t W_k^2\Big). 
$$ 
Denote $F_t = \Big(W_t, W_{t-1}, \cdots, W_0, Z_{t-1}, Z_{t-2}, \cdots, Z_0\Big)$. We have 
$$
\begin{aligned} 
\mathbb{E}\Big[M_t(\lambda) | F_t\Big] & = \mathbb{E}\Big[\exp\Big(\lambda \sum_{s=0}^{t-1} W_s Z_s  + \lambda W_t Z_t - \frac{\lambda^2}{2} \sum_{k=0}^t W_k^2\Big) | F_{t}  \Big] \\ 
&= \exp\Big(\lambda \sum_{s=0}^{t-1} W_s Z_s - \frac{\lambda^2}{2} \sum_{k=0}^{t-1} W_k^2\Big) \mathbb{E}\Big[e^{\lambda W_t Z_t - \frac{\lambda^2}{2} W_t^2} | F_{t} \Big] = M_{t-1}(\lambda)
\end{aligned} 
$$ 
where we used the properties of the Gaussian distribution so that $\mathbb{E}\Big[e^{\lambda W_t Z_t - \frac{\lambda^2}{2} W_t^2} | F_{t} \Big] = 1$, where we used the fact that $W_t$ is measurable with respect to $F_t$ and $Z_t|F_t \sim \mathcal{N}(0,1)$. In particular, by the tower property of the expectation, $\mathbb{E}[M_t(\lambda) ] = \mathbb{E}[M_0(\lambda)] = 1$.

Now take 
$$
\lambda \sum_{t=0}^T Z_t W_t = \log(M_T(\lambda)) + \frac{\lambda^2}{2} \sum_{k=1}^T W_k^2
$$ 
where the equality follows by definition of $M_t(\lambda)$. We can now write the moment generating function of $\sum_{t=0}^T Z_t W_t$ as (remember that $||W||^2 = 1$ almost surely by construction), 
$$
\begin{aligned} 
\mathbb{E}\Big[e^{\lambda \sum_{t=0}^T Z_t W_t}  \Big] & = \mathbb{E}\Big[M_T(\lambda) e^{\frac{\lambda^2}{2} \sum_{k=1}^T W_k^2} \Big] \\
&= \mathbb{E}\Big[M_T(\lambda) e^{\frac{\lambda^2}{2}} \Big]
\end{aligned} 
$$ 
since $||W||^2 = 1$ almost surely. Because we have proven that $\mathbb{E}[M_t(\lambda)] = 1$ for all $t$, we have $\mathbb{E}\Big[e^{\lambda \sum_{t=0}^T Z_t W_t}  \Big] = e^{\frac{\lambda^2}{2}}$ corresponding to the moment generating function of a Gaussian random variable with variance one. Therefore as $n\rightarrow \infty$, $\sum_{t=0}^T W_t Z_t \sim \mathcal{N}(0,1)$.

\paragraph{Term $(II)$} We are only left to show that $(II) \rightarrow_p 1$. We can then invoke Slutsky theorem to complete the proof. We can write
\begin{equation} \label{eqn:II}
\small 
\begin{aligned} 
 |(II)^2 - 1| &= \Big|\frac{\sum_{i= 1}^n\hat{\gamma}_{i,T}^2 (Y_{i, T} - H_{i,T} \hat{\beta}^T)^2 +  \sum_{t=1}^{T-1}  \hat{\gamma}_{i,t}^2 \hat{\nu}_{i,t}^2 + \frac{1}{n^2} \sum_{i=1}^n \hat{\nu}_{i, 0}^2}{\sum_{i= 1}^n \hat{\gamma}_{i,T}^2 \mbox{Var}(\varepsilon_{i,T} | H_{i,T}, D_{i,T}) + \sum_{i= 1}^n \sum_{t=1}^{T-1} \hat{\gamma}_{i,t}^2 \mbox{Var}(\nu_{i,t} | H_{i,t}, D_{i,t}) + \frac{1}{n} \mathrm{Var}(X_{i,1} \beta^1)}
 - 1 \Big| 
\end{aligned} 
\end{equation} 
\begin{equation}
\small 
\begin{aligned} 
\eqref{eqn:II}  &\lesssim \underbrace{\Big|\frac{n \sum_{i= 1}^n\hat{\gamma}_{i,T}^2 \varepsilon_{i, T}^2 +  n \sum_{t=1}^{T-1}  \hat{\gamma}_{i,t}^2 \nu_{i,t}^2 + \frac{1}{n} \sum_{i=1}^n (X_{i,1} \beta^1 - \mathbb{E}[X_{i,1}]\beta^1)^2}{n\sum_{i= 1}^n \hat{\gamma}_{i,T}^2 \mbox{Var}(\varepsilon_{i,T} | H_{i,T}, D_{i,T}) + n\sum_{i= 1}^n \sum_{t=1}^{T-1} \hat{\gamma}_{i,t}^2 \mbox{Var}(\nu_{i,t} | H_{i,t}, D_{i,t}) +  \mathrm{Var}(X_{i,1} \beta^1)}
 - 1 \Big|}_{(A)}  \\ 
& +  \Big|\underbrace{\frac{n \sum_{i= 1}^n \hat{\gamma}_{i,T}^2 \Big[(Y_{i,T} - H_{i,T}\hat{\beta}^T)^2 -  (Y_{i,T} - H_{i,T}\beta^T)^2\Big]  }{n \sum_{i= 1}^n \hat{\gamma}_{i,T}^2 \mbox{Var}(\varepsilon_{i,T} | H_{i,T}, D_{i,T}) + n \sum_{i= 1}^n \sum_{s=1}^T \hat{\gamma}_{i,s}^2 \mbox{Var}(\nu_{i,s} | H_{i,s}, D_{i,s}) + \mathrm{Var}(X_{i,1} \beta^1)}}_{(B)} \Big| \\  &+ \sum_{t=1}^{T-1}
\Big|\underbrace{\frac{n \sum_{i= 1}^n \hat{\gamma}_{i,t}^2 \Big[(H_{i,t + 1} \hat{\beta}^{t+1} - H_{i,t}\hat{\beta}^{t+1})^2 -  (H_{i,t + 1} \beta^{t+1} - H_{i,t} \beta^t)^2\Big]  }{n \sum_{i= 1}^n \hat{\gamma}_{i,T}^2 \mbox{Var}(\varepsilon_{i,T} | H_{i,T}, D_{i,T}) + n \sum_{i= 1}^n \sum_{s=1}^T \hat{\gamma}_{i,s}^2 \mbox{Var}(\nu_{i,s} | H_{i,s}, D_{i,s}) + \mathrm{Var}(X_{i,1} \beta^1)}}_{(C)} \Big| \\ & + 
\underbrace{\Big| \frac{\frac{1}{n} \sum_{i = 1}^n (X_{i,1} \hat{\beta}^1 - \bar{X}_1 \hat{\beta}^1)^2 -  (X_{i,1} \beta^1 - \mathbb{E}[X_{i,1}] \beta^1)^2}{n \sum_{i= 1}^n \hat{\gamma}_{i,T}^2 \mbox{Var}(\varepsilon_{i,T} | H_{i,T}, D_{i,T}) + n \sum_{i= 1}^n \sum_{s=1}^T \hat{\gamma}_{i,s}^2 \mbox{Var}(\nu_{i,s} | H_{i,s}, D_{i,s}) +  \mathrm{Var}(X_{i,1} \beta^1) } \Big|}_{(D)}. 
\end{aligned} 
\end{equation} 
To show that (A) converges it suffices to note that the denominator is bounded from below by a finite positive constant by Lemmas \ref{lem:gamma_1_star}, \ref{lem:firststat} and the fact that each variance component is bounded away from zero under Assumption \ref{ass:4m}. 

The conditional variance of each component in the numerator reads as follows (recall by the above lemmas that $n||\hat{\gamma}_t||^2 = O_p(1)$)
$$
\small 
\begin{aligned} 
& \mathrm{Var}\Big(n \sum_{i=1}^n \hat{\gamma}_{i,T}^2 \varepsilon_{i,T}^2 \Big| H_T\Big) \le n^2  \bar{C} ||\hat{\gamma}_T||_4^4 \le \log^2(n) n^2 \bar{C} n^{-4/3} ||\hat{\gamma}_T||_2^2 = O_p(1) \log^2(n) n n^{-4/3} = o_p(1), \\ 
& \mathrm{Var}\Big(n \sum_{i=1}^n \hat{\gamma}_{i,t}^2 \nu_{i,t}^2 \Big| H_t\Big) \le \bar{C} n^2 ||\hat{\gamma}_T||_4^4 \le n^2 \log^2(n) \bar{C} n^{-4/3} ||\hat{\gamma}_t||_2^2  = O_p(1) \log^2(n) n n^{-4/3} = o_p(1) \\ 
& \frac{1}{n} \mathrm{Var}\Big((X_{i,1} \beta^1 - \mathbb{E}[X_{i,1}] \beta^1)^2\Big) = o(1). 
\end{aligned} 
$$ 
and hence (A) converges to zero by the continuous mapping theorem. 

For the term (B), the denominator is bounded similarly to (A). The numerator is 
\begin{equation} \label{eqn:aa} 
\small 
\begin{aligned} 
n \sum_{i= 1}^n \hat{\gamma}_{i,T}^2 \Big[(Y_{i,T} - H_{i,T}\hat{\beta}^T)^2 -  (Y_{i,T} - H_{i,T}\beta^T)^2\Big] \le   
& n \sum_{i= 1}^n \hat{\gamma}_{i,T}^2  \Big(H_{i,T}(\hat{\beta}^T - \beta^T)\Big)^2 \\
&+ 2 n \sum_{i=1}^n \hat{\gamma}_{i, T}^2 \varepsilon_{i, T} H_{i, T}(\hat{\beta}^T - \beta^T). 
\end{aligned} 
\end{equation} 
We can now write 
$$
\small 
\begin{aligned} 
n \sum_{i= 1}^n \hat{\gamma}_{i,T}^2  \Big(H_{i,T}(\hat{\beta}^T - \beta^T)\Big)^2 & \le ||\hat{\beta}^T - \beta^T||_1^2 n ||\hat{\gamma}_T||^2 ||\max_{i} |H_{i,T}|||^2_{\infty} \\ 
n \sum_{i=1}^n \hat{\gamma}_{i, T}^2 \varepsilon_{i, T} H_{i, T}(\hat{\beta}^T - \beta^T) & \le  ||\hat{\beta}^T - \beta^T||_1 \max_i || H_{i,T}||_{\infty} \max_j |\varepsilon_{j,T}| n ||\hat{\gamma}_T||^2
\end{aligned} 
$$ 
Notice now that by Lemma \ref{lem:sub-gaussian} and Assumption \ref{ass:weakoverlap}, with probability $1 - 1/n$, we have (for sub-gaussian $H_{i,T}$) $||\max_{i} H_{i,T}||_{\infty} = O(\log(np)), \max_j |\varepsilon_{j,T}| = \mathcal{O}(\log(n))$. (To note this, we can write $P(\max_{i, j} |H_{i,T}^{(j)}| > t) \le n p P(|H_{i,T}^{(j)}| > t) \le np e^{-t^2 v}$ for some finite constant $v$. Setting $n p e^{-t^2 v} = 1/n$ the claim holds.) 

Under Assumption \ref{ass:4m}(i), whenever $||\hat{\beta}^T - \beta^T||_1 = O_p(n^{-1/4}), n ||\hat{\gamma}_T||^2 = O_p(1)$ and $\log(np)/n^{1/4} = o(1)$ the above expression is $o_p(1)$ (Condition (a) in Assumption \ref{ass:4m}(i)). 
If instead $H_{i,T}$ is uniformly bounded, $||\max_{i} H_{i,T}||_{\infty} = O(1)$. Therefore it suffices that $||\hat{\beta}^T - \beta^T||_1 = o_p(1/\log(n))$ for the expression to be $o_p(1)$ (Condition (b) in Assumption \ref{ass:4m}(i)). 

Consider now (C), namely
$$
\small 
\begin{aligned} 
 & n \sum_{i= 1}^n \hat{\gamma}_{i,t}^2 \Big[(H_{i,t + 1} \hat{\beta}^{t+1} - H_{i,t}\hat{\beta}^{t})^2 -  (H_{i,t + 1} \beta^{t + 1} - H_{i,t} \beta^t)^2\Big] \le 
n \sum_{i= 1}^n \hat{\gamma}_{i,t}^2 \Big(H_{i,t} (\beta^t - \hat{\beta}^t)\Big)^2  \\
& + 2 \Big|n \sum_{i = 1}^n \hat{\gamma}_{i,t}^2 (H_{i,t+1} \beta^{t+1} - H_{i,t} \beta^t) \Big(H_{i,t} (\beta^t - \hat{\beta}^t)\Big)\Big| + 2 \Big| n \sum_{i=1}^n \hat{\gamma}_{i, t}^2 (H_{i,t + 1} (\hat{\beta}^{t+1} - \beta^{t+1})) (H_{i,t+1} \beta^{t+1} - H_{i,t} \beta^t) \Big| \\
&+ n \sum_{i=1}^n \hat{\gamma}_{i,t}^2 \Big(H_{i, t+1} (\beta^{t+1} - \hat{\beta}^{t+1})\Big)^2 + 2 \Big| n \sum_{i=1}^n \hat{\gamma}_{i,t}^2 \Big(H_{i,t+1}(\beta^{t+1} - \hat{\beta}^{t+1})\Big)\Big(H_{i,t}(\beta^t - \hat{\beta}^t)\Big)\Big| 
\end{aligned} 
$$ 
The first component $n \sum_{i= 1}^n \hat{\gamma}_{i,t}^2 \Big(H_{i,t} (\beta^t - \hat{\beta}^t)\Big)^2 = o_p(1)$ following verbatim the argument for the first-term on the right-hand side of Equation \eqref{eqn:aa} with $H_{i,t}$ in lieu of $H_{i,T}$ and $(\beta^t - \hat{\beta}^t)$ in lieu of $(\hat{\beta}^T - \beta^T)$, and similarly $n \sum_{i= 1}^n \hat{\gamma}_{i,t}^2 \Big(H_{i,t+1} (\beta^{t+1} - \hat{\beta}^{t+1})\Big)^2 = o_p(1)$.  The component $n \sum_{i = 1}^n \hat{\gamma}_{i,t}^2 (H_{i,t+1} \beta^{t+1} - H_{i,t} \beta^t) \Big(H_{i,t} (\beta^t - \hat{\beta}^t)\Big) = o_p(1)$ following verbatim the argument for the second term in the right-hand side of Equation \eqref{eqn:aa} with $\nu_{i,t+1} = H_{i,t+1} \beta^{t+1} - H_{i,t} \beta^t$ in lieu of $\varepsilon_{i,T}$, which is sub-gaussian by Lemma \ref{lem:sub-gaussian}, and $H_{i,t}, (\beta^t - \hat{\beta}^t)$ in lieu of $H_{i,T}, (\hat{\beta}^T - \beta^T)$. The same reasoning applies to  n $\sum_{i=1}^n \hat{\gamma}_{i, t}^2 (H_{i,t + 1} (\hat{\beta}^{t+1} - \beta^{t+1})) (H_{i,t+1} \beta^{t+1} - H_{i,t} \beta^t)$  with $\nu_{i,t}$ in lieu of $\varepsilon_{i,T}$ and $(H_{i,t + 1} (\hat{\beta}^{t+1} - \beta^{t+1})$ in lieu of $H_{i,T}(\hat{\beta}^T - \beta^T)$. Finally, for the last term, we can bound $n \sum_{i=1}^n \hat{\gamma}_{i,t}^2 \Big(H_{i,t+1}(\beta^{t+1} - \hat{\beta}^{t+1})\Big)\Big(H_{i,t}(\beta^t - \hat{\beta}^t)\Big) \le n ||\hat{\gamma}_t||^2 ||\hat{\beta}^{t+1} - \beta^{t+1}||_1 ||\hat{\beta}^t - \hat{\beta}_t||_1 ||H_t||_{\infty} ||H_{t+1}||_{\infty} = o_p(1)$ under Assumption \ref{ass:4m}(i) and \ref{ass:weakoverlap}, since by Lemma \ref{lem:firststat}, $n ||\hat{\gamma}_t||^2 = O_p(1)$, and $||\hat{\beta}^{t+1} - \beta^{t+1}||_1 ||H_{t+1}||_{\infty} = o_p(1)$ if either $||\hat{\beta}^{t+1} - \beta^{t+1}||_1 = O_p(1/n^{1/4})$ and $H_{t+1}$ is subgaussian, so that $||H_{t+1}||_{\infty} = O_p(\log(np))$ (since by assumption $\log(pn)/n^{1/4} = o(1)$) or if $||\hat{\beta}^{t+1} - \beta^{t+1}||_1 = o_p(1)$ and $||H_{t+1}||_{\infty} = O(1)$. 
 
\paragraph{Rate of convergence is $n^{-1/2}$: proof of Theorem \ref{thm:convergence_rate_first}.} To study the rate of convergences it suffices to show that (for fixed $T$)
$$
n  \Big[\sum_{i=1}^n \hat{\gamma}_{i,T}^2 \mbox{Var}(\varepsilon_{i,T} | H_{i,T}, D_{i,T}) + \sum_{i=1}^n \sum_{t=1}^{T-1} \hat{\gamma}_{i,t}^2 \mbox{Var}(\nu_{i,t} | H_{i,t}, D_{i,t})\Big] + \mathrm{Var}(X_{i,t} \beta^1) = O_p(1). 
$$ 
This follows directly from Lemma \ref{lem:firststat}, \ref{lem:gamma_1_star} and the bounded conditional third moment assumption in Assumption \ref{ass:4m}.

 \subsection{Inference on ATE} \label{app:robust}

\begin{thm}[Inference on ATE] \label{cor:ate}
Let Assumptions \ref{ass:seqignm}, \ref{ass:weakoverlap},  \ref{ass:4m} hold. Let $d_1 \neq d_1'$.
Then, whenever $\log(np_T)/n^{1/4} \to 0$ with $n,p_1, \cdots, p_T\to \infty$, 
$$
\small 
\begin{aligned} 
\lim_{n \rightarrow \infty} { (\hat{V}_T(d_{1:T}) + \hat{V}_T(d_{1:T}') - \hat{v}(d_{1:T}, d_{1:T}'))^{-1/2}\sqrt{n} \Big(\hat{\mu}(d_{1:T}) - \hat{\mu}(d_{1:T}') - \mathrm{ATE}(d_{1:T}, d_{1:T}')\Big)} \rightarrow_d \mathcal{N}(0,1), 
 \end{aligned} 
$$
where 
$$
\small 
\begin{aligned} 
\hat{v}(d_{1:T}, d_{1:T}') = \frac{1}{n} \sum_{i=1}^n \left\{ \Big((X_{i,1} - \bar{X}_1)^\top \hat{\beta}_{d_{1:T}}^{(1)}\Big)^2 + \Big((X_{i,1} - \bar{X}_1)^\top \hat{\beta}_{d_{1:T}'}^{(1)}\Big)^2 -  \Big((X_{i,1} - \bar{X}_1)^\top (\hat{\beta}_{d_{1:T}}^{(1)} - \hat{\beta}_{d_{1:T}'}^{(1)})\Big)^2 \right\}. 
\end{aligned} 
$$  
\end{thm} 
  
  Following the decomposition of Theorem \ref{thm:thm_asym_t},
we can write 
$$
\small 
\begin{aligned} 
& \hat{\mu}(d_{1:T}) - \hat{\mu}(d_{1:T}') - \Big(\mu_T(d_{1:T}) - \mu_T(d_{1:T}')\Big) = o_p(n^{-1/2}) \\ 
&  + \sum_{d \in \{d_{1:T}, d_{1:T}'\} }\sum_{i=1}^n \left\{\hat{\gamma}_{i,T}(d) \varepsilon_{i,T}(d) + \sum_{t=1}^T \hat{\gamma}_{i,t}(d) \nu_{i,t}(d)\right\} + \Big(\bar{X}_1 (\beta_{d_{1:T}} - \beta_{d_{1:T}'}) - \mu_T(d_{1:T}) - \mu_T(d_{1:T}')\Big)
\end{aligned} 
$$ 
 Because $\mathbb{E}[\bar{X}_1 (\beta_{d_{1:T}} - \beta_{d_{1:T}'})] =  \mu_T(d_{1:T}) - \mu_T(d_{1:T}')$ and since each $\hat{\gamma}_{i,t}(d), \hat{\gamma}_{i,t}(d'), d \neq d'$ sum over a different set of (independent individuals), we can follow verbatim the proof of Theorem \ref{thm:thm_asym_t} with the variance appropriately adjusted by the sum over the variances of each element $\sum_i \hat{\gamma}_{i,T}(d) \varepsilon_{i,T}, \sum_i \hat{\gamma}_{i,t}(d) \nu_{i,t}(d), d \in \{d_{1:T}, d_{1:T}'\}$ and of $\bar{X}_1 (\beta_{d_{1:T}} - \beta_{d_{1:T}'})$ in lieu of $\bar{X}_1 \beta_{d_{1:T}}$ in the proof of Theorem \ref{thm:thm_asym_t}.



\section{Additional numerical studies} \label{sec:num}

\subsection{Longer time horizon} 

Next, we compare DCB and AIPW with high dimensional covariates with a longer time period. Namely, in Figure \ref{fig:long_T},  we collect results for $T \in \{1,\cdots, 10\}$. We generate data using a sparse model, $p = 100, n = 400$ over two-hundred replications. The outcome at time $t$ depends on the contemporaneous treatment, covariates, and previous outcome at time $t - 1$. To simulate a scenario where a strong correlation occurs between treatments over a long time period, we generate 
$
\mathbb{E}[D_{i,t} | D_{t-1}, X_{t}] = (1 - \alpha) D_{i,t-1} + \alpha (1 + e^{\iota_{i,t}})^{-1}, 
$
where similarly to the propensity score model Equation \eqref{eqn:propmodel}, $\iota_{i,t} = \frac{\eta}{\alpha} X_{i,t-1} \phi + \frac{\eta}{\alpha} X_{i,t} \phi +  \frac{1}{2}(D_{i,t-1} - \bar{D}_{t-1}) + \xi_{i,t}, \xi_{i,t} \sim \mathcal{N}(0,1)$. Here $\eta$ controls overlap together with $\alpha$, where $\eta/\alpha$ has a similar role of the overlap constant in previous simulations. (We take $\eta/\alpha$ as this plays approximately the same role of $\eta$ in previous simulations from a simple linear approximation of $(1 + e^{\iota_{i,t}})^{-1}$ with respect to $\eta \approx 0$.)
In the figure we report results for $\alpha \in \{0.9, 0.7, 0.5\}$ (denoted as ``High, Medium and Low correlation" respectively), and $\eta \in \{0.3, 0.5\}$. In Figure \ref{fig:long_T}, we observe that for very strong time dependence between treatments (i.e., there are limited or no dynamics in assignments) the two methods are comparable. When instead, there are relatively more dynamics in treatment assignments the proposed method significantly improves in mean-squared error, with larger improvements in the presence of poorer overlap.  

\begin{figure}[!ht]
\centering
\includegraphics[scale=0.5]{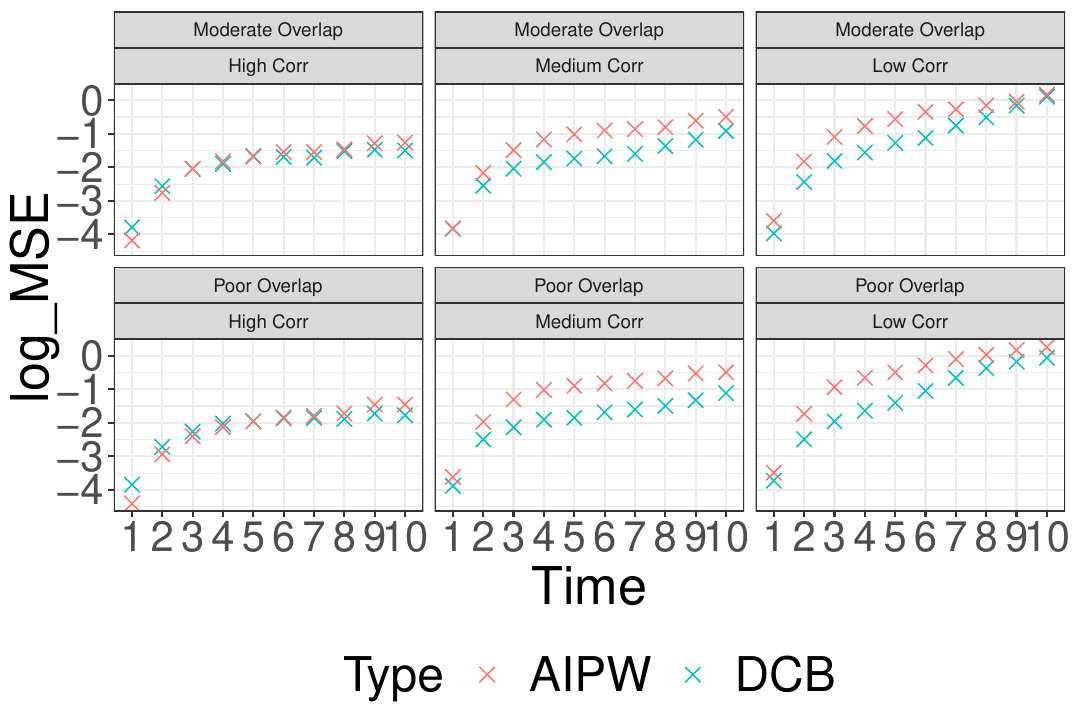}
\caption{Mean-squared error in log-scale. Simulations for $T \le 10, p = 100, n = 400$, two-hundred replications. Here high-correlation denotes strong serial depedence between treatment assignments with $\alpha = 0.9$, medium with $\alpha = 0.7$ and weak with $\alpha = 0.5$. $\eta \in \{0.3, 0.5\}$ for moderate and poor overlap, respectively. } 
\label{fig:long_T}
\end{figure}

\begin{table*}[!ht]\centering
\caption{Summary statistics of the distribution of the propensity score in two and three periods in a sparse setting with $\mathrm{dim}(X) = 300$.}    \label{tab:summaries} 
\ra{1.3}
\scalebox{0.9}{\begin{tabular}{@{}lrrrcrrrcrrr@{}}\toprule
& \multicolumn{2}{c}{$\eta=0.1$} & \ & \multicolumn{2}{c}{$\eta=0.3$} & \ & \multicolumn{2}{c}{$\eta=0.5$} \\
\cmidrule{2-3} \cmidrule{5-6}  \cmidrule{8-9}  
&T=2& T=3& &T=2& T=3 &&T=2& T=3  \\ \midrule
Min & 0.012 & $0.003 $ && 0.004&0.0002 && 0.001& 0.00000 \\
1st Quantile  &0.126 &0.049 && 0.105&0.031&& 0.079& 0.018 \\
 Median & 0.218& 0.097 & &0.216& 0.097& & 0.216 & 0.094  \\
3rd Quantile & 0.248& 0.126 && 0.259 & 0.153&&0.277 & 0.183\\
Max& $0.352$ &0.175& & 0.377 & 0.226 && 0.429 &0.286    \\
\bottomrule
\end{tabular}}
\end{table*}

\subsection{Computational complexity}

Figure \ref{fig:long_T2} reports the computational complexity. 

\begin{figure}[!ht]
\centering
\includegraphics[scale=0.4]{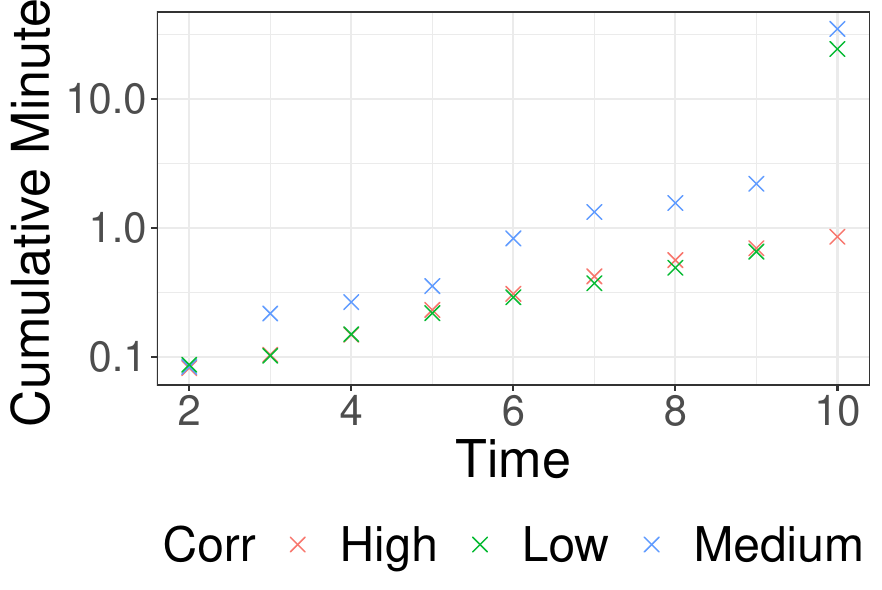}
\caption{ Example of cumulative computational time in minutes to run the first (one) simulation on a personal laptop as $T$ varies
High-correlation denotes strong serial dependence between treatments ($\alpha = 0.9$), medium corresponds to $\alpha = 0.7$ and weak to $\alpha = 0.5$. We set $\eta = 0.1$. } 
\label{fig:long_T2}
\end{figure}

\subsection{Simulations under misspecification} \label{sec:app:misspecified}

We simulate the outcome model over each period using non-linear dependence between the outcome, covariates, and past outcomes. The function that we choose for the dependence of the outcome with the past outcome and covariates follows similarly to \cite{athey2018approximate}, where, differently, here, such dependence structure is applied not only to the first covariate only (while keeping a linear dependence with the remaining ones) but to all covariates, making the scenarios more challenging for the DCB method. Formally, the DGP is the following: 
$$
\begin{aligned} 
Y_2(d_1, d_2) &= \log(1 + \exp(-2-2X_1\beta_{d_1, d_2})) + \log(1 + \exp(-2-2 X_2\beta_{d_1, d_2})) 
\\
&+  \log(1 + \exp(-2-2 Y_1 )) + d_1 + d_2+\varepsilon_2 ,
\end{aligned} 
$$
and similarly for $Y_3(d_1, d_2, d_3)$, with also including covariates and outcomes in period $T = 2$. Coefficients $\beta$ are obtained from the sparse model formulation discussed in the main text. Results are collected in Table \ref{tab:misspecified} for the MSE and for the bias and variance in the subsequent tables below. Interestingly, we observe that DCB performs relatively well under the misspecified model, even if our method does not use any information on the propensity score. We also note that our adaptation of the double lasso to dynamic setting performs comparable or better in the presence of two periods only or a sparse structure. However, as the number of periods increase or sparsity decreases Double Lasso's performance deteriorates.  

Finally, we observe that DCB outperforms AIPW, with a known propensity score. The main reason is due to the instability of the inverse probability weights in dynamic settings. Because these weights define the joint probability of treatment assignments, these can exhibit instability (poor overlap) in small samples, increasing the variance of the AIPW estimator. This behavior can be observed as we decompose the bias and variance of AIPW: whereas AIPW has a smaller finite sample bias than DCB with a misspecified model, its variance is substantially larger.

\begin{table}[ht]\centering
\caption{MSE under misspecified model in a sparse setting.}    \label{tab:misspecified} 
\ra{1.3}
\begin{tabular}{@{}lrrrrrrrrrrrrrr@{}}\toprule
& \multicolumn{2}{c}{$T=2$} & \ & \multicolumn{2}{c}{$T=3$} \\
\cmidrule{2-3}   \cmidrule{5-6}   
 & $\eta = 0.3$ & $\eta = 0.5$   &   &    $\eta = 0.3$ & $\eta = 0.5$   \\
\Xhline{.8pt} 
DCB & $0.238$ & $0.354$     &&  $0.751$ & $0.402$ \\ 
aIPW* & $0.434$ & $0.802$ && $1.363$ & $1.622$ \\
aIPWh &  $0.863$ & $1.363$ && $1.882$ & $2.464$ \\ 
CAEW (MSM)  & $0.815$ & $1.364$ && $7.889$ & $8.675$ \\ 
D. Lasso& $0.121$ & $0.142$ && $0.689$ & $0.503$ \\ 
Seq.Est & $0.811$ & $0.346$ && $2.288$ & $2.031$ \\ 
DiD switchback & $60.05$ & $100.5$ && $795.5$ & $1104$ \\ 
Local Projection & $0.639$ & $0.382$ && $1.995$ & $1.922$ \\ 
 \bottomrule
\end{tabular}
\end{table}

\begin{table}[ht]\centering
\caption{MSE under misspecified model in a moderately sparse setting.}    \label{tab:misspecified} 
\ra{1.3}
\begin{tabular}{@{}lrrrrrrrrrrrrrr@{}}\toprule
& \multicolumn{2}{c}{$T=2$} & \ & \multicolumn{2}{c}{$T=3$} \\
\cmidrule{2-3}   \cmidrule{5-6}   
 & $\eta = 0.3$ & $\eta = 0.5$   &   &    $\eta = 0.3$ & $\eta = 0.5$   \\
\Xhline{.8pt} 
DCB & $0.212$ & $0.256$     &&  $0.326$ & $0.384$ \\ 
aIPW* & $0.428$ & $0.789$ && $1.364$ & $1.616$ \\
aIPWh &  $0.826$ & $1.313$ && $1.857$ & $2.434$ \\ 
CAEW (MSM) & $0.781$ & $1.317$ && $7.833$ & $8.616$ \\ 
D. Lasso& $0.115$ & $0.133$ && $0.675$ & $0.494$ \\ 
Seq.Est & $0.847$ & $0.366$ && $2.316$ & $2.058$ \\ 
DiD switchback & $59.71$ & $100.1$ && $795$ & $1104$ \\ 
Local Projection & $0.670$ & $0.408$ && $2.023$ & $1.950$ \\ 
 \bottomrule
\end{tabular}
\end{table}

\begin{table}[!htbp] \centering 
  \caption{Bias for sparse setting under misspecified model.} 
  \label{tab:tab1} 
\begin{tabular}{@{\extracolsep{5pt}} ccccc} 
\\[-1.8ex]\hline 
\hline \\[-1.8ex] 
 & $t = 2, \eta = 0.3$ & $t = 2,\eta = 0.5$ & $t = 3, \eta = 0.3$ & $t = 3, \eta = 0.5$ \\ 
\hline \\[-1.8ex] 
DCB & $0.227$ & $0.340$ & $$-$0.467$ & $$-$0.199$ \\ 
AIPW - Known Prop & $0.146$ & $0.288$ & $$-$0.0003$ & $0.318$ \\ 
AIPW - High Prop & $0.852$ & $1.119$ & $1.245$ & $1.459$ \\ 
AIPW - Low Prop & $0.551$ & $1.045$ & $1.378$ & $2.057$ \\ 
CAEW & $0.760$ & $1.086$ & $2.718$ & $2.872$ \\ 
Double Lasso & $0.156$ & $0.225$ & $0.671$ & $0.469$ \\ 
Seq.Est. & $$-$0.793$ & $$-$0.448$ & $$-$1.391$ & $$-$1.276$ \\ 
DiD Switchback & $7.66$ & $9.98$ & $28.02$ & $33.15$ \\ 
Local Projection & -0.746 & -0.566 & -1.370 & -1.343 \\ 
\hline \\[-1.8ex] 
\end{tabular} 
\end{table}

\begin{table}[!htbp] \centering 
  \caption{Variance for sparse setting under misspecified model.} 
  \label{tab:tab2} 
\begin{tabular}{@{\extracolsep{5pt}} ccccc} 
\\[-1.8ex]\hline 
\hline \\[-1.8ex] 
 & $t = 2, \eta = 0.3$ & $t = 2, \eta = 0.5$ & $t = 3, \eta = 0.3$ & $t = 3,\eta = 0.5$ \\ 
\hline \\[-1.8ex] 
DCB & $0.187$ & $0.239$ & $0.533$ & $0.363$ \\ 
AIPW - Known Prop & $0.413$ & $0.719$ & $1.364$ & $1.521$ \\ 
Naive Lasso & $0.273$ & $0.259$ & $1.058$ & $1.194$ \\ 
AIPW - High Prop & $0.138$ & $0.111$ & $0.333$ & $0.336$ \\ 
AIPW - Low Prop & $0.612$ & $0.225$ & $0.827$ & $0.438$ \\ 
CAEW & $0.237$ & $0.184$ & $0.500$ & $0.425$ \\ 
Double Lasso & $0.098$ & $0.092$ & $0.239$ & $0.284$ \\ 
Seq.Est. & $0.183$ & $0.145$ & $0.354$ & $0.404$ \\ 
DiD Switchback & $1.25$ & $0.86$ & $10.17$ & $5.39$ \\ 
Local Projection & 0.079 & 0.061 & 0.118 & 0.118 \\ 
\hline \\[-1.8ex] 
\end{tabular} 
\end{table}

\begin{table}[!htbp] \centering 
  \caption{Bias for moderately sparse model under misspecification.} 
  \label{tab:tab3} 
\begin{tabular}{@{\extracolsep{5pt}} ccccc} 
\\[-1.8ex]\hline 
\hline \\[-1.8ex] 
 & $t = 2, \eta = 0.3$ & $t = 2, \eta = 0.5$ & $t = 3, \eta = 0.3$ & $t = 3, \eta = 0.5$ \\ 
\hline \\[-1.8ex] 
This Paper & $0.202$ & $0.358$ & $0.096$ & $0.323$ \\
AIPW - Known Prop & $0.123$ & $0.266$ & $$-$0.010$ & $0.308$ \\ 
AIPW - High Prop & $0.830$ & $1.097$ & $1.235$ & $1.449$ \\ 
AIPW - Low Prop & $0.529$ & $1.023$ & $1.367$ & $2.047$ \\ 
CAEW & $0.738$ & $1.064$ & $2.708$ & $2.862$ \\ 
Double Lasso & $0.134$ & $0.202$ & $0.661$ & $0.459$ \\ 
Seq.Est. & $$-$0.815$ & $$-$0.470$ & $$-$1.401$ & $$-$1.286$ \\ 
DiD Switchback & $7.64$ & $9.96$ & $28.01$ & $33.15$ \\ 
Local Projection & -0.768 & -0.588 & -1.380 & -1.353 \\ 
\hline \\[-1.8ex] 
\end{tabular} 
\end{table}

\begin{table}[!htbp] \centering 
  \caption{Variance for moderately sparse model under misspecification.} 
  \label{tab:tab4} 
\begin{tabular}{@{\extracolsep{5pt}} ccccc} 
\\[-1.8ex]\hline 
\hline \\[-1.8ex] 
 & $t = 2, \eta = 0.3$ & $t = 2, \eta = 0.5$ & $t = 3, \eta = 0.3$ & $t = 3, \eta = 0.5$ \\ 
\hline \\[-1.8ex] 
This Paper & $0.171$ & $0.129$ & $0.317$ & $0.280$ \\ 
AIPW - Known Prop & $0.413$ & $0.719$ & $1.364$ & $1.521$ \\ 
AIPW - High Prop & $0.138$ & $0.111$ & $0.333$ & $0.336$ \\ 
AIPW - Low Prop & $0.612$ & $0.225$ & $0.827$ & $0.438$ \\ 
CAEW & $0.237$ & $0.184$ & $0.500$ & $0.425$ \\ 
Double Lasso & $0.098$ & $0.092$ & $0.239$ & $0.284$ \\ 
Seq.Est. & $0.183$ & $0.145$ & $0.354$ & $0.404$ \\ 
DiD Switchback & $1.25$ & $0.86$ & $10.17$ & $5.39$ \\ 
Local Projection & 0.079 & 0.061 & 0.118 & 0.118 \\ 
\hline \\[-1.8ex] 
\end{tabular} 
\end{table} 

\subsection{Simulations with low dimensional covariates}

In Figure \ref{fig:lowdim}, we explore the performance of the method in low dimensional scenarios where $p \in \{10, 20\}$. DCB outperforms AIPW and IPW uniformly except for $p = 10$ and strong overlap of the propensity score, where DCB performs comparably or slightly worst than AIPW. However, in all other scenarios where overlap decays (both moderate and weak overlap), DCP outperforms AIPW in low-dimensional settings. Improvements of DCB over AIPW increase with the number of periods.  These results justify DCB also in low dimensional scenarios in the presence of poor or moderately poor overlap of the propensity score. 

\begin{figure}[!ht] 
\centering
\includegraphics[scale=0.42]{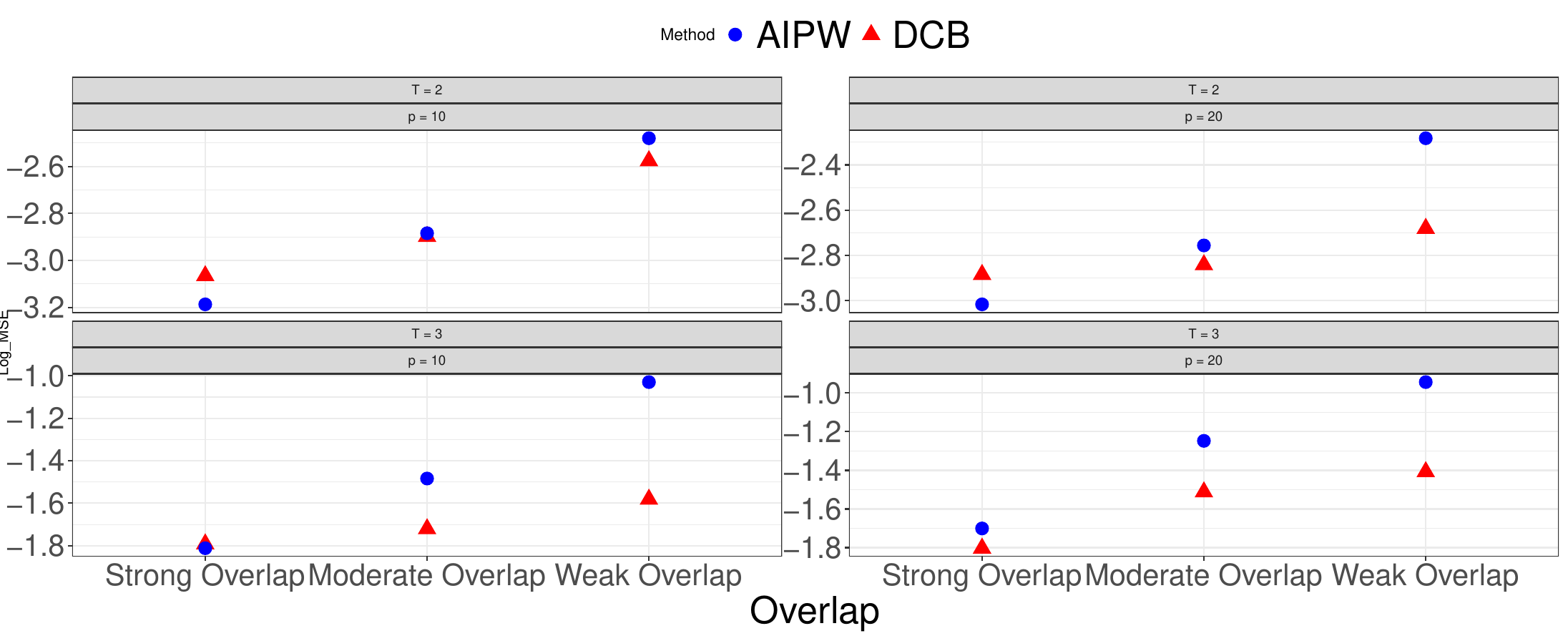}
\includegraphics[scale=0.42]{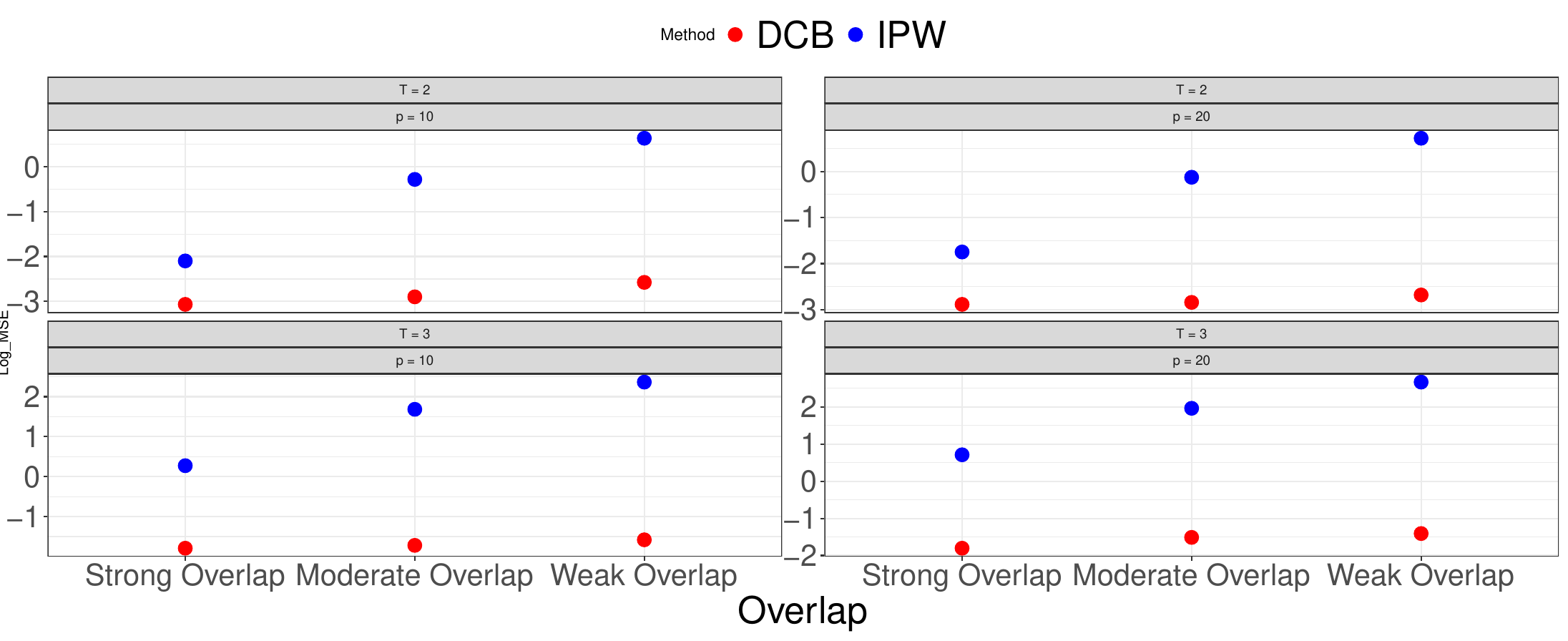}
\caption{Comparison of DCB with AIPW and IPW (with estimated propensity score and estimated conditional mean function) over 200 replications. Low dimensional scenarios with $p \in \{10, 20\}$. The y-axis report the MSE in log-scale the the x-axis reports different scenarios in terms of overlap (strong, moderate and weak overlap). } \label{fig:lowdim} 
\end{figure}

\subsection{Comparisons as signal strength varies}

In this section, we present simulations results with the same design as in Section \ref{sec:numerics} for a sparse setting ($p = 10$), but with one distinction: we multiply the signal of the coefficients by $\eta \in \{0.1, 0.3, 0.5\}$ (overlap constant), varying the strength of the signal in the linear regression. This approach follows in spirit with simulations in Section 3 in \cite{wuthrich2023omitted}. Similarly to what observed in  \cite{wuthrich2023omitted}, as the signal decreases ($\eta$ moves from $0.5$ to $0.1$ in our case), the performance of AIPW that uses lasso deteriorates. The relative improvements of DCB over AIPW increase as the signal strength decreases, further motivating the proposed method. 

\begin{figure}[!ht] 
\centering
\includegraphics[scale=0.45]{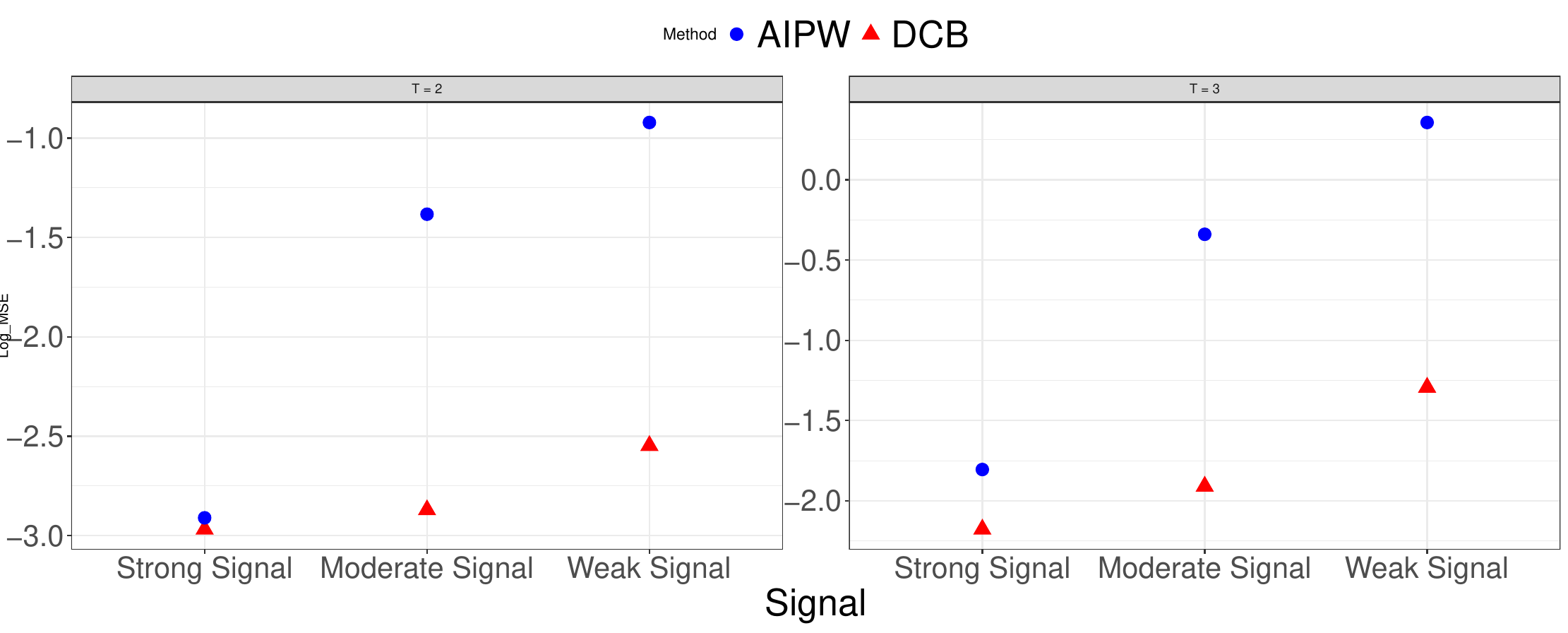}
\includegraphics[scale=0.45]{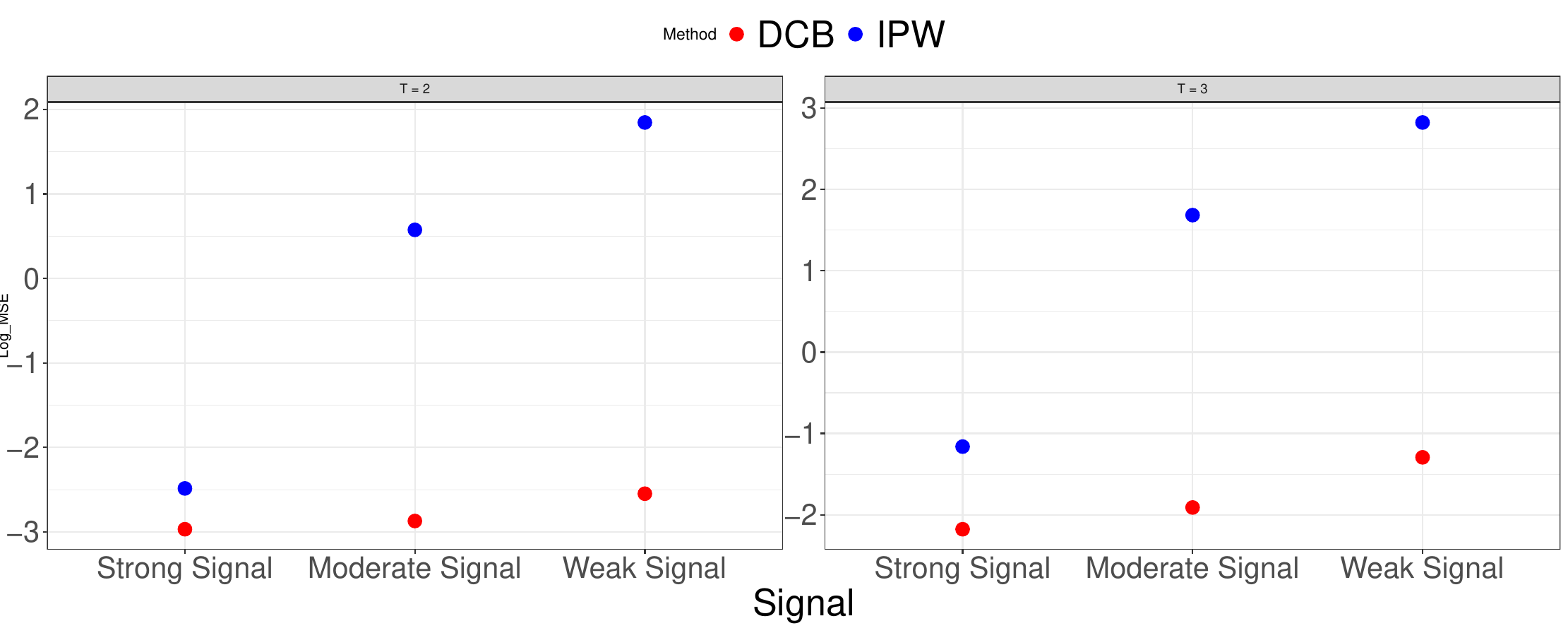}
\caption{$p = 100$, simulation design as in Section \ref{sec:numerics} with one difference: the coefficients $\beta$ multiply by $\eta$ with $\eta \in \{0.1, 0.3, 0.5\}$ (strong, moderate and weak signal, respectively). y-axis reports the MSE in log-scale. } 
\end{figure}

\subsection{Additional empirical results}

Figure \ref{fig:dispersion} illustrates the instability of the inverse probability weights compared to DCB weights in our empirical application.

\begin{figure}[!ht]
\centering 
\includegraphics[scale=0.5]{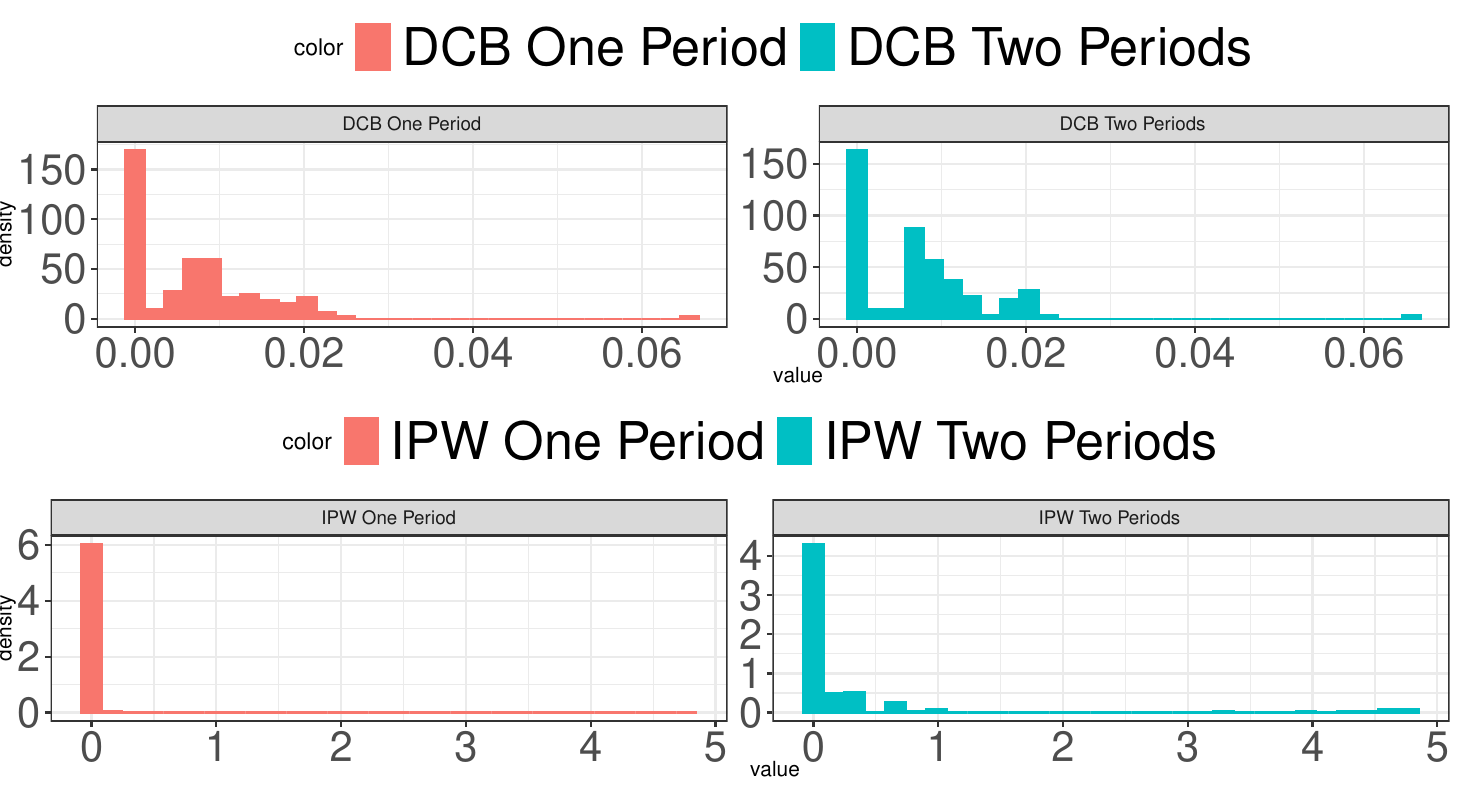} 
\caption{Comparisons between DCB weights and inverse probability weights when estimating the effect of treatment for two consecutive periods. Dispersion in the weights is associated with higher variance of the estimator. Zeros for DCB weights correspond to observations who are not weighted as in Algorithm 1, because their treatment differs from one. Both weights are estimated on the same sample over the last two periods, with inverse probability weights controlling for four lagged outcomes and past assignment. The figure illustrates the larger sensitivity of inverse probability weights to longer periods.} \label{fig:dispersion}
\end{figure}

 \bibliography{mybibliography}
\bibliographystyle{chicago}

 	\end{document}